\documentclass[12pt]{article}

\usepackage[round,authoryear]{natbib}
\usepackage{amsthm,amsmath,amsfonts,amssymb,mathrsfs}
\usepackage{bm}
\usepackage{dsfont,stmaryrd}

\usepackage{graphicx}
\usepackage{epsfig}
\usepackage{epstopdf}
\usepackage{float}
\usepackage{subcaption}
\usepackage{booktabs}
\usepackage{multirow}
\usepackage{array}

\usepackage{tikz}
\usepackage{pgfplots}
\pgfplotsset{compat=1.18}

\usepackage{enumerate}

\usepackage{algorithm}
\usepackage{algpseudocode}

\usepackage{url}

\usepackage{bbm}        
\usepackage{extarrows}
\usepackage{framed}
\usepackage{lscape}
\usepackage{pifont}
\usepackage{stackengine}

\usepackage[normalem]{ulem}

\usepackage{hyperref}
\hypersetup{
    colorlinks=true,
    linkcolor=blue,
    filecolor=magenta,
    urlcolor=cyan
}

\allowdisplaybreaks

\newtheorem{definition}{Definition}[section]
\newtheorem{assumption}{Assumption}[section]
\newtheorem{theorem}{Theorem}[section]
\newtheorem{example}{Example}[section]
\newtheorem{proposition}{Proposition}[section]
\newtheorem{remark}{Remark}[section]
\newtheorem{lemma}{Lemma}[section]
\newtheorem{corollary}{Corollary}[section]

\makeatletter
\renewcommand{\section}{\@startsection{section}{1}{0pt}%
  {\baselineskip}{0.5\baselineskip}{\large\bfseries}}
\renewcommand{\subsection}{\@startsection{subsection}{2}{0pt}%
  {0.5\baselineskip}{0.3\baselineskip}{\normalsize\bfseries}}
\makeatother

\newcommand{\E}{\mathbb{E}}

\newcommand{\R}{\mathbb{R}}

\newcommand{\sgn}{\text{sgn}}

\DeclareMathOperator*{\argmin}{arg\,min}

\DeclareMathOperator{\Var}{Var}
\DeclareMathOperator{\rhomax}{\rho_{\max}}
\newcommand{\fdp}{\text{FDP}}
\newcommand{\estfdp}{\widehat{\text{FDP}}}

\newcommand{\blind}{0} 

\usepackage[affil-it]{authblk}  
\usepackage{mathtools}

\newcommand{\FDP}{\text{FDP}}

\newcommand{\sj}{s_j^{(2)}}
\newcommand{\tsj}{\widetilde s_j^{(2)}}
\newcommand{\Xk}[1]{\boldsymbol{X}^{(#1)}}
\newcommand{\yk}[1]{\boldsymbol{y}^{(#1)}}

\newcommand{\tyk}[1]{\widetilde{\boldsymbol{y}^{(#1)}}}
\newcommand{\tXkj}[1]{\widetilde{\boldsymbol{X}^{(#1,j)}}}
\newcommand{\tykj}[1]{\widetilde{\boldsymbol{y}^{(#1,j)}}}
\newcommand{\Djr}{D_j^{\mathrm{r}}}
\newcommand{\Dk}[1]{D^{(#1)}}
\newcommand{\pr}{\mathbbm{P}}
\newcommand{\betakj}[1]{\widehat{\beta}^{(#1)}_j}
\newcommand{\absbetakj}[1]{|\widehat{\beta}^{(#1)}_j|}
\newcommand{\Djrt}{{D_{j}^{\mathrm{r},t}}}
\newcommand{\Djkt}[1]{D_{j}^{(#1,t)}}
\newcommand{\Djktau}[1]{D_{j}^{(#1,\tau)}}
\newcommand{\Drt}{{D^{\mathrm{r},t}}}
\newcommand{\Dkt}[1]{D^{(#1,t)}}
\newcommand{\Dktau}[1]{D^{(#1,\tau)}}

\begin{document}

\newcommand{\bX}{\boldsymbol{X}}
\def\spacingset#1{\renewcommand{\baselinestretch}%
{#1}\small\normalsize} \spacingset{1}

\if0\blind
{
   \title{\bf PRADAS: PRior-Assisted DAta Splitting for False Discovery Rate Control}
   \author[1]{Yuanchuan Guo,{$^\ast$}}
  \author[1]{Buyu Lin,\footnote{These authors contribute equally to this work.}}
\author[2]{Jun S. Liu}
\affil[1]{Department of statistics, Harvard University}
\affil[2]{Department of statistics, Tsinghua University}
\date{\today}

  \maketitle
} \fi

\if1\blind
{
  \bigskip
  \bigskip
  \bigskip
  \begin{center}
    {\LARGE\bf{PRADAS: PRior-Assisted DAta Splitting for False Discovery Rate Control}}
\end{center}
  \medskip
} \fi

\bigskip


\begin{abstract}
In the FDR-controlling literature, mirror statistics offer a flexible alternative to $p$-value based procedures. 
When prior information is available, however, it is unclear how to incorporate mirror statistics in a principled way, and the standard equal split used by data-splitting methods can be inefficient. In this paper, we characterize a broader class of mirror statistics for any fixed splitting scheme and establish asymptotic FDR control under mild weak-dependence conditions using a two-stage procedure inspired by \cite{li2021whiteout}. Within this class, we derive a Bayes-optimal mirror statistic. Theoretically, we demonstrate its power advantage through analyses in the Rare/Weak signal model. Building upon this Bayes-optimal mirror statistic, we propose \textsc{PRADAS} (PRior-Assisted DAta Splitting) that treats split ratio as a stopping time and recasts the data-splitting as an optional stopping over a natural filtration; the optimal stopping rule is characterized by the Snell envelope and computed efficiently via a Longstaff--Schwartz regression approximation. Both simulations and real data examples demonstrate the effectiveness of our proposed framework.
\end{abstract}

\textbf{Keywords:} data-splitting, feature selection, mirror statistic, power analysis, optimal stopping

\newpage
\spacingset{1.75}

\section{Introduction}

In large-scale multiple comparisons where hundreds or thousands of hypotheses are tested simultaneously, the goal is to identify all non-null hypotheses simultaneously while controlling the false discovery rate (FDR) \citep{benjamini1995controlling}, which is defined as:
\[
    \text{FDR} = \mathbbm{E}[\text{FDP}], \quad 
    \text{FDP} = \frac{\# \{j:j\in S_0, j\in \widehat{S}\}}{\# \{j:j\in \widehat{S}\}\lor 1}
\]
where $S_0$ is the set of null features, $\widehat{S}$ is the set of selected features and FDP stands for false discovery proportion.
The expectation is taken with respect to the randomness in both the data and the selection procedure if it is a stochastic algorithm.

The foundation of modern multiple testing lies in the seminal work of \cite{benjamini1995controlling}, who introduced the Benjamini-Hochberg (BH) procedure to control FDR.  While originally proven under the assumption of independence, \citet{benjaminiyekutieli2001} later extended its validity to $p$-values having positive regression dependence on a subset (PRDS) property. For more general dependence structures where PRDS cannot be assumed, the Benjamini-Yekutieli (BY) procedure \citep{benjaminiyekutieli2001} offers control but at the cost of significantly reduced power (introducing a $\log m$ penalty, where $m$ is the number of hypotheses). Subsequent research has focused on improving power by incorporating adaptivity to the proportion of null hypotheses \citep{storey2002direct, storey2004strong} or weighting hypotheses based on prior domain knowledge \citep{genovese2006false}. However, all these methods share a common bottleneck: they require valid $p$-values. In high-dimensional settings with unknown nuisance parameters or complex dependence structures, deriving such valid $p$-values is often intractable.

To circumvent the reliance on valid $p$-values, \citet{barber2015} introduced the knockoff filter. By constructing synthetic `knockoff' variables that mimic the correlation structure of the original features without influencing the response, the method constructs statistics that satisfy a specific swap symmetry property. While the original framework was limited to the low-dimensional setting ($n \ge 2p$), \citet{candes2018panning} generalized this to the Model-X knockoff framework, which allows for valid inference in high-dimensional regimes ($p > n$) by assuming knowledge of the covariate distribution. However, despite being able to control FDR in finite sample, the knockoff framework faces practical instability in the presence of high correlations. As noted in recent literature \citep{li2021whiteout}, when features are highly correlated, the requirement that knockoffs remain exchangeable with the original features forces the knockoffs to be nearly identical to the original variables. This leads to a phenomenon known as ``statistical whiteout", where the method loses the ability to distinguish signal from noise, resulting in a severe loss of power. Furthermore, the construction of valid model-X knockoffs for complex, non-Gaussian distributions remains a computationally intensive and often intractable challenge.

\subsection{Controlling FDR through data-splitting}

Generalizing the symmetry principle popularized by knockoffs, \citet{dai2023false, dai2023scale} proposed a data-splitting  framework, which splits the data into two halves. Through data-splitting, they were able to construct a mirror statistic $M_j$ that is symmetric about zero under the null for any feature $j$. 
They also argued that the symmetry property does not necessarily have to be exact. Asymptotic symmetry is enough for controlling FDR. Formally, they require 
\begin{equation*}
    \max_{j\in S_0, t>0} |\mathbbm{P}(M_j>t)-\mathbbm{P}(M_j<-t)| \overset{n,p\to \infty}{\rightarrow} 0
\end{equation*}
as the sample size $n$ and the number of features $p$ increase.

The core idea in this line of research for controlling FDR, including mirror statistics \citep{dai2023false} and Gaussian mirror \citep {xing2023controlling}, is the symmetry property of the mirror statistic under the null. Once we have computed the mirror statistics satisfying the symmetry property, we can estimate the number of false positives for any threshold $\tau>0$. That is, we can conservatively estimate the right tail $\# \{j\in S_0:M_j >\tau\}$, by the left tail $\#\{j:M_j<-\tau\}$. Then,  we can obtain a conservative estimate $\widehat{\text{FDP}}$ for the FDP: 
\begin{equation*}
    \text{FDP} := \frac{\# \{j\in S_0: M_j>t\}}{\# \{j: M_j>t\}\lor 1}, \quad \ \
    \widehat{\text{FDP}} = \frac{\# \{j: M_j<-t\}}{\# \{j: M_j>t\} \lor 1}.
\end{equation*}
For any nominal FDR level $q$, we can choose a data-driven cut-off $\tau_q$ as follows:
\begin{equation}
\label{eq:tauq}
    \tau_q = \min \{t>0: \widehat{\text{FDP}}\leq q\},
\end{equation}
and select non-null features (aka ``discoveries'') as:
$\widehat{S} = \{j:M_j> \tau_q\}$. 
We name this procedure as {\it SeqStep} and summarize it in Algorithm \ref{alg:selecseq}. It has been implemented in \cite{dai2023false, dai2023scale, spector2022asymptotically, barber2015} as a key step for controlling FDR. 

\begin{algorithm}[t]
\caption{SeqStep Algorithm}
\label{alg:selecseq}
\begin{algorithmic}[1] 
\State Randomly split the data into two groups $(\boldsymbol{X}^{(1)},\boldsymbol{y}^{(1)})$ and $(\boldsymbol{X}^{(2)},\boldsymbol{y}^{(2)})$.
\State Calculate summary statistic and $M_j$ based on \eqref{general-form-mr-stats}.
\State Given a nominal FDR level $q$, calculate the data-driven cut-off $\tau_q$ as:
\begin{equation*}
    \tau_q = \min \{t>0: \widehat{\text{FDP}} = \frac{\# \{j: M_j<-t\}}{\# \{j: M_j>t\} \lor 1} \leq q\}
\end{equation*}
\State Select the features $\{j:M_j>\tau_q\}$.
\end{algorithmic}
\end{algorithm}

For linear models, \citet{dai2023false, dai2023scale} first obtain independent estimates from the two half-sized datasets, $\widehat{\beta}^{(1)}_j$ and $\widehat{\beta}^{(2)}_j$, for the unknown coefficient vector $\beta_j$, using potentially two different methods. They require only that  one of the estimates (i.e., either $\widehat{\beta}^{(1)}_j$ or $\widehat{\beta}^{(2)}_j$) to be symmetric about zero under the null. 
The choice for two estimates can be flexible, as long as they are asymptotically symmetric about zero under the null, such as the debiased lasso estimator \citep{van2014asymptotically, zhang2014confidence, javanmard2014confidence} and feature gradient in neural networks \citep{xing2020neuralgaussianmirrorcontrolled}.
In high-dimensional linear and graphical models, \citet{dai2023false, dai2023scale} further introduced a lasso + OLS procedure: lasso is applied to one half of the data for feature screening, and OLS is then applied to the selected features using the other half. The mirror statistic is defined as $M_j = \text{sgn}(\widehat{\beta}_j^{(1)}\widehat{\beta}_j^{(2)}) \cdot f(|\widehat{\beta}_j^{(1)}|,|\widehat{\beta}_j^{(2)}|)$, where $f(u,v)$ is non-decreasing in both arguments $u$ and $v$.

\subsection{Increasing power through improving hypotheses ranking}

In many modern applications, rich domain knowledge is available and serves as valuable prior information for inference. For example, the correlation structure of genetic features or the sparsity pattern of disease-causing mutations provides critical structural constraints \citep{cui2021fused}. External knowledge from historical studies or auxiliary datasets can be formally integrated via Bayesian transfer learning, allowing current experiments to ``borrow strength" from source domains to improve precision \citep{suder2025bayesian}.  
Therefore, how to leverage prior information to improve the power of making discoveries while controlling for FDR becomes the central question in feature selection literature.
Most of the existing FDR-controlling methods first rank all candidate hypotheses or features and then select top candidates via a data-driven threshold. In light of this, maximizing power is equivalent to solving a ranking problem: if a procedure can successfully rank true signals ahead of null features, the FDR-controlling threshold will naturally adapt to select more discoveries. Consequently, the literature has evolved to develop algorithms that leverage the prior information to improve the ranking of hypotheses or features and go beyond viewing each hypotheses the same a priori.

For p-value based methods, the optimal ranking is grounded in the local false discovery rate (Lfdr) \citep{efron2001empirical, efron2002empirical} in terms of Bayes optimality. \cite{oraclecompoundfdr2007} developed a compound decision theory framework for multiple testing and derived an oracle rule based on the transformed $p$-values, i.e., $z$-values. They showed that the decision rule that rejects hypothese with the lowest Lfdr values is optimal in the sense that it minimizes the false non-discovery rate for a predefined FDR level under independence. Later, \cite{optimalfdrdependent2011} extended this optimality to the week dependence case. 
\citet{lei2018adapt} proposed an adaptive $p$-value thresholding procedure that
utilizes side information to estimate the Lfdr in an adaptive matter, effectively re-ordering hypotheses to prioritize those with higher signal probability. 
Apart from the full Bayesian approach, \cite{genovese2006false} pioneered the concept of weighted $p$-values. Their idea is to relax the p-value thresholds for hypotheses that
are more likely to be alternative and tighten the thresholds for the other hypotheses so that the
overall FDR level can be controlled.
\cite{hu2010false} developed a
weighed p-value procedure under grouped hypotheses structure by estimating the proportion of null hypotheses for each group separately.
Later, \cite{li2019multiple} proposed the structure-adaptive BH algorithm that generalized this idea by using the censored $p$-values. \cite{ignatiadis2016data,ignatiadis2017covariate} proposed the independent hypothesis weighting
(IHW) for multiple testing with covariate information. Overall, these approaches can be viewed as different ways of ranking hypotheses, either through Lfdr or prior-informed weighting, to focus power on the most promising signals under FDR control.

For knockoff-based methods, the ranking of hypotheses lies in the quality of the feature importance statistic $W$. 
\cite{li2021whiteout} recasted the fixed-X knockoff procedure as a conditional post-selection inference method. They introduced an optimal fixed-X knockoff procedure $\textsc{knockoff}^*$ under Gaussian linear model. 
\cite{spector2022asymptotically}  formalized the theory of asymptotically optimal knockoff statistics, proving that a ranking based on the masked likelihood ratio (MLR) is Bayes-optimal. 
\cite{ren2023knockoffs} introduced an adaptive knockoff filter that adaptively order the variables using progressively revealed information and focus on those that are the most promising.

\subsection{Our contribution}

For mirror statistic, \citet{dai2023false,dai2023scale} noticed that the sign-sum mirror statistic, defined as $M_j = \sgn(\hat{\beta}_j^{(1)}\cdot\hat{\beta}_j^{(2)})(|\hat{\beta}_j^{(1)}|+|\hat{\beta}_j^{(2)}|)$, achieves asymptotic optimality when not considering prior information. While \citet{lin2026correlationcaveat} further substantiated the power superiority of the sign-sum mirror statistic, it is critical to note that such optimality does not hold when prior information is available. Furthermore, this optimality is conditioned on the underlying data-splitting scheme (fixed split ratio), which dictates the distribution of information for two splits. Consequently, a methodological gap remains regarding the integration of prior information to jointly optimize the construction of the mirror statistic and the data-splitting proportion, thereby moving beyond the constraints of predefined fixed-ratio partitioning. Our contributions are threefold:

\begin{itemize}
    \item \emph{General class of mirror statistics}: 
    We extend the class of mirror statistics introduced in \cite{dai2023false, dai2023scale} via a two stage procedure inspired by \cite{li2021whiteout}. We rigorously demonstrate that any mirror statistic derived from this framework is capable of controlling FDR asymptotically.
    
    \item \emph{Optimal mirror statistic derivation}: We derive the optimal form of the mirror statistic for any given data-splitting scheme and show that it is Bayes-optimal. As a supplement, we explore its error rate under the Rare/Weak signal model proposed by \citet{donoho2004higher}, demonstrating that it is strictly more powerful than the method suggested by \citet{dai2023false,dai2023scale}, particularly in scenarios where features exhibit high correlation.

    \item \emph{Prior-assisted Data Splitting}:
    We develop a prior-assisted strategy for sample partitioning and propose the ADaptive Mirror Statistic (ADMS). Unlike predefined splitting schemes, ADMS enables the dynamic allocation of samples between the two splits, based on how much prior information we have. Methodologically, we recast the determination of the optimal splitting ratio as an optimal stopping time problem. We then solve this problem by invoking the theory of the Snell envelope \citep{snell1952applications}, thereby deriving a stopping rule that maximizes the expected power of the procedure.
\end{itemize}
Finally, we empirically demonstrate the effectiveness of our proposed methods through
both simulations and real-world data analyses.

\subsection{Notation and outline}
Throughout the paper, we denote the set of feature variables as $(X_1, X_2, \ldots, X_p)$, which are assumed to follow a $p$-dimensional distribution. We denote the (random) design matrix formed by the $n$ observations as $\boldsymbol{X} = [\boldsymbol{X}_1, \boldsymbol{X}_2, \ldots, \boldsymbol{X}_p] \in \mathbb{R}^{n \times p}$, where $\boldsymbol{X}_j = (X_{1j}, X_{2j}, \ldots, X_{nj})^\top$ is the vector that contains the $n$ independent realizations of the feature $X_j$. For each observation $i$, its feature vector is $(X_{i1}, X_{i2}, \ldots, X_{ip})$, and there is an associated response variable $y_i$.
We assume that the response variable depends only on a small set of features indexed by $S_1$, i.e., $X_{S_1} = \{X_j: j \in S_1\}$. Our goal is to identify the features in $X_{S_1}$ under a prespecified FDR control. Let $S_0 = \{1, 2, \ldots, p\} \setminus S_1$ denote the index set of null (irrelevant) features. We denote the numbers of null and relevant features as $|S_0|$ and $|S_1|$, respectively.
Let $\theta\in \Theta$ represent the parameter(s) of the conditional distribution $y\mid (X_1,X_2,\cdots,X_p)$. The prior of $\theta$ is denoted as $\theta\sim\pi(\theta)$. In spike-and-slab model, $\nu_t$ denotes a distribution concentrated at $t$.
In the context of generalized linear models, we denote the true coefficient vector as $\boldsymbol{\beta}$ and its estimate as $\widehat{\boldsymbol{\beta}} = (\widehat{\beta}_1,\widehat{\beta}_2, \ldots, \widehat{\beta}_p)$, which will be the maximum likelihood estimate unless otherwise specified.
We use superscripts $(1)$ and $(2)$ to represent quantities derived from each of the two splits of the data, respectively. For example, $(\boldsymbol{X}^{(1)}, \boldsymbol{y}^{(1)})$ and $(\boldsymbol{X}^{(2)}, \boldsymbol{y}^{(2)})$ are the samples from the two splits, and $\widehat{\boldsymbol{\beta}}^{(1)}$ and $\widehat{\boldsymbol{\beta}}^{(2)}$ are the corresponding maximum likelihood estimators. Throughout, we define $\sigma(\cdot)=\text{sigmoid}(\cdot)$.

The rest of the paper is organized as follows. In Section 2, we discuss how to use prior information to construct mirror statistic. In Section 3, we theoretically prove that when features are correlated, the proposed mirror statistic incorporating prior information are strictly more powerful than the vanilla ones in the Rare/Weak model considered in \cite{ke2024power}. In Section 4, we utilize prior information to determine the optimal splitting scheme and introduce our Adaptive Data Splitting (ADMS) framework. Empirical analyses are presented in Sections 5 and 6, where we demonstrate the effectiveness of our methods.

\section{Bayes-optimal Mirror Statistic}
\label{sec:pads}
\subsection{A general class of mirror statistics via two-stage procedure}

\citet{dai2023false,dai2023scale} proposed a class of mirror statistics and recommended the use of sign-sum mirror statistic. However, the sign-sum mirror statistic is limited to linear models and does not leverage any prior information for feature ranking. To take advantage of the prior, we first introduce a broader class of mirror statistics having asymptotical FDR control and derive the Bayes-optimal mirror statistic. To begin, we assume  to have a batch of $n$ i.i.d. observations from $f_\theta(x, y)$, denoted as $(\boldsymbol{X},\boldsymbol{y})$, where $f_\theta$ denotes a specific parametric family. We divide the data into two parts of equal size: $(\boldsymbol{X}^{(1)}, \boldsymbol{y}^{(1)})$ and $(\boldsymbol{X}^{(2)}, \boldsymbol{y}^{(2)})$. 
For the second half of the data, suppose that for every feature $j$, we have a \emph{summary statistic} $s_j^{(2)}$ and its \emph{mirror} $\widetilde{s}_j^{(2)}$, both independent of the first half of the data and satisfying Assumption~\ref{ass:exchanablity}. Intuitively, one cannot distinguish the {summary statistic} of any null feature from its corresponding \emph{mirror} by looking at values of the other features. 

\begin{assumption}[Exchangeability]
\label{ass:exchanablity}
    For each null feature index $j\in S_0$, $(s_j^{(2)}, \widetilde{s}_j^{(2)}) \overset{d}{=} (\widetilde{s}_j^{(2)}, s_j^{(2)})$.
\end{assumption}

\begin{example}[Sign-flip]
\label{ex:negation}
    In linear regression, we can simply choose $s_j^{(2)}$ to be the $j$-th coefficient of the OLS estimator, i.e., $\widehat{\beta}_j^{(2)}$, and set its \emph{mirror} as $\widetilde s_j^{(2)}=-\widehat{\beta}_j^{(2)}$. In this case, Assumption~\ref{ass:exchanablity} holds since $\widehat{\beta}_j^{(2)}$ is symmetric around zero when the $j$-th feature is null.
\end{example}

\begin{example}[Model-X]
\label{ex:modelX}
    If the joint distribution for $(X_1,\cdots,X_p)$ is accessible, we may proceed differently. For each feature $j$, we construct \emph{mirror feature} $\tXkj{2}$ by replacing  the $j$-th entry of $\boldsymbol{X}^{(2)}$, i.e., $X_j^{(2)}$, with an independent draw from the conditional distribution $[X_j \mid \bX_{-j}=\boldsymbol{X}_{-j}^{(2)}]$ while keeping $\boldsymbol{X}_{-j}^{(2)}$ unchanged. Then, for arbitrary user-defined measurable functions $f_1,f_2,\cdots,f_p$, the pair $(\sj,\tsj)=(f_j(\Xk{2},\yk{2}),f_j(\tXkj{2},\yk{2}))$ satisfies Assumption \ref{ass:exchanablity}.
\end{example}

\begin{example}[Model-y]
\label{ex:modely}
    Suppose that $\boldsymbol{y} \mid \boldsymbol{X}$ follows a generalized linear model with coefficient vector $\boldsymbol{\beta}$. Let $S_{-j}$ denote the sufficient statistics for the submodel excluding feature $j$; when $\beta_j = 0$, these coincide with the sufficient statistics for the full model. For each feature $j$, we resample $n$ samples $(\tXkj{2}, \tykj{2})$ conditioned on $S_{-j}$ assuming $\beta_j = 0$. Specifically, for linear regression, let $S_{-j} = (\boldsymbol{X}_{-j}^{(2)\top} \yk{2}, \lVert \boldsymbol{M}_{-j} \yk{2} \rVert^2)$, where $\boldsymbol{P}_{-j}$ is the projection matrix onto the column space of $\boldsymbol{X}_{-j}^{(2)}$ and $\boldsymbol{M}_{-j} = \boldsymbol{I} - \boldsymbol{P}_{-j}$. We can set $(\tXkj{2}, \tykj{2}) = (\Xk{2}, \yk{2} - 2 \boldsymbol{M}_{-j} \yk{2})$, i.e., we flip $\yk{2}$ to the opposite direction of the signal while keeping $\boldsymbol{P}_{-j} \yk{2}$ unchanged. Then, the pair $(s_j, \tsj) = \bigl( \log \Pr(\Xk{2}, \yk{2} \mid \Dk{1}), \log \Pr(\tXkj{2}, \tykj{2} \mid \Dk{1}) \bigr)$ satisfies Assumption~\ref{ass:exchanablity}.

\end{example}

Next, we set \( \Djr \subseteq (\boldsymbol{X}^{(1)}, \boldsymbol{y}^{(1)}, \{s_j^{(2)}, \widetilde{s}_j^{(2)}\}) \) as the \emph{ranking data} for feature \( j \), where \( \{s_j^{(2)}, \widetilde{s}_j^{(2)}\} \) is an unordered pair, meaning we do not know which value corresponds to the $j$-th summary statistics or its mirror. Now, we are ready to introduce the two stage procedure inspired by \cite{li2021whiteout}. As shown below, the \emph{ranking data} $\Djr$ will be used to  order the features in consideration, and we will only select features that rank higher. 

\begin{remark}
To make the ranking more accurate, we should include as much information in $\Djr$ as possible, i.e., setting \( \Djr = (\boldsymbol{X}^{(1)}, \boldsymbol{y}^{(1)}, \{s_j^{(2)}, \widetilde{s}_j^{(2)}\}) \). We must be careful with the information contained in $\Djr$. One might consider including summary statistics from other features and their mirrors. However, this risks violating FDR control. For example, in Example \ref{ex:negation}, the unordered pair \( \{s_j^{(2)}, \widetilde{s}_j^{(2)}\} \) is equivalent to \( |\widehat{\beta}_j^{(2)}| \). Adding other statistics effectively introduces \( |\widehat{\beta}_i^{(2)}| \) for some \( i \ne j \). When \( j \) is null but \( i \) is not, \( \sgn(\widehat{\beta}_j^{(2)}) \) becomes correlated with \( |\widehat{\beta}_i^{(2)}| \), breaking the required symmetry and compromising FDR control. 
\end{remark}

\noindent 
\textbf{A Two-Stage procedure:}
\begin{itemize}
\item[Stage 1] (Exploration). For each feature $j$, the analyst calculates a non-negative \emph{ranking score} $\eta_j\geq 0$ based on $\Djr$ 
to make a \textit{\textit{guess}} $\psi_j=\psi_j(\Djr)$ as to which of the two candidates, $s_j^{(2)}$ or $\widetilde{s}_j^{(2)}$, is the true summary statistic instead of the mirror.
\item[Stage 2] (Confirmation). The analyst examines whether the \textit{guess} $\psi_j$ made in Stage 1 is correct for each $j$, proceeding in the decreasing order of the \emph{ranking score} $\eta_j$.
Define $r_j=1$ if the analyst's \textit{guess} $\psi_j$ is correct (i.e., if $s_j^{(2)}$ is indeed the chosen summary statistic for feature $j$), and $r_j=-1$ otherwise.
The FDP after examining $k$ features is estimated as
\begin{equation*}
    \widehat{\text{FDP}}_k = \frac{\sum_{i=1}^k \mathbbm{1}(r_j \neq 1)}
    {\sum_{i=1}^k \mathbbm{1}(r_j = 1)\lor 1}.
\end{equation*}
\end{itemize}
Finally, the analyst may choose the largest $k$ such that $\widehat{\text{FDP}}_k \leq q$, where $q$ is the nominal FDR level, and selects the first $k$ features examined. 

Based on the two stage procedure, we are able to introduce a \emph{general class of mirror statistics} in the following theorem.
\begin{theorem}
\label{thm:equi-twostage-generalclass}
    The above two stage procedure is equivalent to applying Algorithm \ref{alg:selecseq} with the following \emph{general class of mirror statistic}:
    \begin{equation}
    \label{general-form-mr-stats}
        M_j = r_j \cdot \eta_j(\Djr),
    \end{equation}
    where $r_j\in\{\pm 1\}$ represents whether the \textit{guess} $\psi_j=\psi_j(\Djr)$ is correct and $\eta_j(\Djr) \geq 0$ represents the magnitude of the mirror statistic and can be chosen as an arbitrarily function of $\Djr$.
\end{theorem}

\begin{remark}
    Let us revisit Example \ref{ex:negation}. The \emph{ranking} data $\Djr$ can be written as $(\Xk{1},\yk{1},\absbetakj{2})$. The corresponding \emph{general class of mirror statistic} can be written as
    \begin{equation}
    \label{general-form-mr-stats-linear-model}
        M_j = \text{sgn}(\widehat{\beta}_j^{(2)})\cdot \psi_j(\Djr) \cdot \eta_j(\Djr),
    \end{equation}
    where $r_j = \text{sgn}(\widehat{\beta}_j^{(2)})\cdot \psi_j(\Djr)$ with the decision function $\psi_j(\Djr)\in \{\pm 1\}$ being the \textit{guess} for the sign of $\widehat{\beta}_j^{(2)}$, effectively representing the choice of the true summary statistic between $s_j^{(2)}$ and $\widetilde{s}_j^{(2)}$. Note that if we choose $\psi_j(\Djr) = \sgn(\betakj{1})$ and $\eta_j(\Djr) = f(|\betakj{1}|,|\betakj{2}|)$, we recover the class of mirror statistics in \cite{dai2023false,dai2023scale}. 
\end{remark}

To this end, we have derived a general class of mirror statistics in \eqref{general-form-mr-stats} and \eqref{general-form-mr-stats-linear-model}. The power of a mirror statistic depends on two factors:
\begin{itemize}
    \item The \textit{guess} $\psi_j$: Its accuracy in identifying the summary statistic corresponding to the true  $s_j^{(2)}$ 
  impacts the likelihood of correctly detecting relevant features.
    \item The \emph{ranking score} $\eta_j$: Properly ranking features according $\eta_j$ ensures that the most promising features are prioritized for examination.
\end{itemize}

To achieve optimal power in feature selection while controlling the FDR, it is sufficient to optimize $\psi_j$ and $\eta_j$ separately. 
Under Assumption \ref{ass:exchanablity}, it is clear to see that, for any null feature $j\in S_0$, $\psi_j(\Djr)$ is unaffected if we swap $s_j^{(2)}$ and $\widetilde{s}_j^{(2)}$, and 
$ r_j \sim \text{Unif}(\{\pm 1\})$.
To control FDP asymptotically with high probability, we need the variance of $\widehat{\text{FDP}}$ to be reasonably small. Therefore, the mirror statistics for the null features cannot be too correlated. Additionally, we require that the variance of each mirror statistic be uniformly bounded away from both zero and infinity to avoid triviality. These requirements are summarized as follows:
\begin{assumption}[Weak correlation among features]
\label{ass:weak-cor-among-nulls}
    The summary statistics $s_j^{(2)}$'s are continuous random variables and there exist constants $\alpha<1$ and $c>0$ such that
    \begin{equation*}
        \sum_{1\leq i<j\leq p} \rhomax\left(\left(s_{i}^{(2)},\widetilde s_{i}^{(2)}\right),\left(s_{j}^{(2)},\widetilde s_{j}^{(2)}\right) \right)\leq c \ p^{2\alpha},
    \end{equation*}
    where $\rhomax$ denote the {\it maximal correlation}. 
\end{assumption}
Recall that the {\it maximal correlation} \citep{Rnyi1959OnMO} between r.v.'s $X$ and $Y$ is defined as 
    \begin{equation*}
    \begin{aligned}
    \rho_{\text{max}}(X,Y)
    &=
    \sup \left\{\mathbb{E}[f(X)g(Y)]\text{:} \quad
       \begin{aligned}
          &f,g \text{ are Borel measurable functions,} \\
          &\text{with} \ \mathbbm{E}[f(X)^2]\leq 1, \mathbbm{E}[g(Y)^2]\leq 1.
       \end{aligned}
    \right\}.
    \end{aligned}
\end{equation*}
In Assumption \ref{ass:weak-cor-among-nulls}, we use the maximal correlation rather than just Pearson correlation because we need to control dependence strengths of \emph{nonlinear} functions of the null statistics. When the features are jointly Gaussian, the maximal correlation reduces to the absolute correlation \citep{Rnyi1959OnMO} so that the condition coincides with the familiar correlation-based notion in that setting.
If we can arrange the $s_j^{(2)}$'s in such a way that $\rho_{\text{max}}(s_i^{(2)},s_j^{(2)})\leq r_1^{|i-j|}$ for a constant $r_1<1$, then 
Assumption \ref{ass:weak-cor-among-nulls} holds. This assumption essentially restricts the null mirror statistics from being nearly perfectly dependent. Except for some extreme cases, we believe that this assumption holds in fairly broad settings.
In linear models, \citet{dai2023false} showed that the weak dependence assumption holds as long as the covariance matrix of the null features satisfies some regularity conditions. In Section \ref{sec:simulations}, we also empirically found that even when the features have constant pairwise correlations, our method effectively controls FDR while maintaining a high power. 
Similarly to Proposition 2.2 in \cite{dai2023false}, the proposition below guarantees that Algorithm \ref{alg:selecseq} with any proper mirror statistic of the form \eqref{general-form-mr-stats} controls FDP asymptotically.

Recall that $\text{FDP}(t)$ refers to the FDP of the selection $\widehat{S}_t = \{j:M_j>t\}$. 
The following proposition shows that for any nominal level $q\in(0,1)$, $\text{FDP}(\tau_q)$ and the corresponding $\text{FDR}(\tau_q)$ are under control, where $\tau_q$ is the data-dependent cutoff chosen according to \eqref{eq:tauq}.

\begin{proposition}[Valid Asympototic FDR control]
\label{prop:fdr-control}
    For any nominal FDR control level $q\in (0,1)$, assume that there exists a constant $t_q>0$ such that $\mathbbm{P}(\text{FDP}(t_q)\leq q) \to 1$ as $p\to \infty$. Then, under Assumptions \ref{ass:exchanablity} and \ref{ass:weak-cor-among-nulls}, the procedure in Algorithm \ref{alg:selecseq} satisfies
    \begin{equation*}
        \text{FDP}(\tau_q) \leq q + o_p(1).
    \end{equation*}
\end{proposition}

The existence of such $t_q$  guarantees the asymptotic feasibility of FDR control based upon the rankings of features according to their mirror statistics. It essentially guarantees that the cutoff threshold $\tau_q$ is bounded and does not diverge to infinity. 
Proposition~\ref{prop:fdr-control} can be generalized when we also prescreen features using the first half of the data, provided that the prescreening algorithm has the sure screening property \citep{fan2008sureindependencescreeningultrahigh}.

\subsection{Bayes-optimal mirror statistic based on the two-stage procedure}
Now we assume that parameter $\theta$ of the family $f_\theta(x,y)$ has a prior distribution $\pi(\theta)$.  
Unless otherwise stated, our analysis accounts for randomness in both the prior and the observed data.  Now we discuss how to construct the Bayes-optimal mirror statistics.

\medskip

\textbf{Optimality of $\psi_j$.} We aim to choose $\psi_j$, as a function of $\Djr$, to maximize the probability $\mathbbm{P}(r_j = 1 \mid \Djr)$.
Since $\mathbbm{P}(M_j>0 \mid \Djr) = \mathbbm{P}(r_j = 1 \mid \Djr)$, maximizing $\mathbbm{P}(M_j>0 \mid \Djr)$ increases the number of positive mirror statistics, which then increases the number of potential discoveries. Let $\psi^\star_j$ and $r^\star_j$ denote the \emph{optimal} choices of $\psi_j$ and $r_j$, respectively.

\medskip

\textbf{Optimality of the ranking.}  Given the  \textit{optimal guess} $\psi_j^\star$, we seek the best ranking of the features. The FDR control procedure aims to find the largest $k$ such that $\widehat{\text{FDP}}_k\leq q$, where
\begin{equation*}
\widehat{\text{FDP}}_k = \frac{\sum_{i<k}\mathbbm{1}(r^\star_{(i)}= -1)}{\sum_{i<k}\mathbbm{1}(r^\star_{(i)}=1)\vee 1}.
\end{equation*}
For now, we assume that $\widehat{\text{FDP}}_k$ converges to a smoother version: 
\begin{equation*}
    \widetilde{\text{FDP}}_k = \frac{\sum_{i<k}\mathbbm{P}(r^\star_{(i)}=-1 \mid \Djr)}{\sum_{i<k}\mathbbm{P}(r^\star_{(i)}=1 \mid \Djr)\vee 1}.
\end{equation*}
Now, we choose the \emph{optimal ranking score} as
\begin{equation}
\label{eq:etaj}
    \eta_j^\star = \log \frac{\mathbbm{P}(r^\star_j= 1 \mid \Djr)}{\mathbbm{P}(r^\star_j= -1 \mid \Djr)}.
\end{equation}
By the construction of $\psi_j^\star$, we know that $\mathbbm{P}(r^\star_j= 1 \mid \Djr) \geq \mathbbm{P}(r^\star_j= -1 \mid \Djr)$. Therefore, $\eta_j^\star\geq 0$ is well defined. If, for some $j$, we have $\eta_{(j)}^\star<\eta_{(j+1)}^\star$, we can switch the ordering of $(j)$ and $(j+1)$ so that, for any $k$, the numerator of $\widehat{\text{FDP}}_k$ decreases and the denominator of $\widehat{\text{FDP}}_k$ increases. Therefore, ranking according to $\eta_j^\star$ minimizes $\widehat{\text{FDP}}_k$ for any $k$, which, by the definition of Algorithm \ref{alg:selecseq}, leads to the most rejections among all possible rankings. 



Thus, the Bayes-optimal mirror statistic with \( \Djr = (\boldsymbol{X}^{(1)}, \boldsymbol{y}^{(1)}, \{s_j^{(2)}, \widetilde{s}_j^{(2)}\}) \), which maximizes the number of rejections is given by:
\begin{align}
\label{opt-mr-stats-general}
M_j^\star &= r_j^\star \cdot \eta_j^\star 
= r_j^\star \cdot \log \frac{\mathbbm{P}(r^\star_j = 1 \mid \Djr)}{\mathbbm{P}(r^\star_j = -1 \mid \Djr)} \
= \log \frac{\mathbbm{P}_{(s_j^{(2)}, \widetilde{s}_j^{(2)}\mid D^{(1)})}(s_j^{(2)}, \widetilde{s}_j^{(2)})}
{\mathbbm{P}_{(s_j^{(2)}, \widetilde{s}_j^{(2)} \mid D^{(1)})}(\widetilde{s}_j^{(2)}, s_j^{(2)})},
\end{align}
where $\mathbbm{P}_{(s_j^{(2)},\widetilde{s}_j^{(2)} \mid D^{(1)})}$ is the joint probability density of $(s_j^{(2)},\widetilde{s}_j^{(2)})$ conditional on $D^{(1)}$. 
The first two equalities follow directly from the definition of the mirror statistic provided in \eqref{general-form-mr-stats} and the choice  specified in \eqref{eq:etaj}.
We establish the last equality as follows:
If $\mathbbm{P}_{(s_j^{(2)},\widetilde{s}_j^{(2)}\mid D^{(1)})}(s_j^{(2)}, 
\widetilde{s}_j^{(2)})\geq \mathbbm{P}_{(s_j^{(2)},\widetilde{s}_j^{(2)} \mid D^{(1)})}(\widetilde{s}_j^{(2)},s_j^{(2)})$, then the test $\psi^\star_j$ selects $s_j^{(2)}$, yielding in $r_j^\star=1$. Otherwise it selects  $\widetilde s_j^{(2)}$, yielding $r_j^\star=-1$. This ensures that the signs of both sides match. Moreover, the magnitude on the left-hand side can be written explicitly as: 
\[\log \frac{\mathbbm{P}(r^\star_j= 1 \mid \Djr)}{\mathbbm{P}(r^\star_j= -1 \mid \Djr)} = \log \frac{\max\{\mathbbm{P}(s_j^{(2)},\widetilde{s}_j^{(2)}\mid \Djr),\mathbbm{P}(\widetilde{s}_j^{(2)},s_j^{(2)}\mid \Djr)\}}{\min\{\mathbbm{P}(s_j^{(2)},\widetilde{s}_j^{(2)}\mid \Djr),\mathbbm{P}(\widetilde{s}_j^{(2)},s_j^{(2)}\mid \Djr)\}},\]
which matches the magnitude on the right-hand side, thereby confirming their equality.

Since the constructed mirror statistic \eqref{opt-mr-stats-general} is within the general mirror statistics class \eqref{general-form-mr-stats}, Proposition \ref{prop:fdr-control} ensures its asymptotic FDR control. For optimality, it suffices to state that \eqref{opt-mr-stats-general} makes the highest number of rejections among the class of \eqref{general-form-mr-stats}.
Proposition~\ref{prop:optimality} provides the Bayes-optimality statement when features are weakly correlated.

\begin{assumption}
    \label{ass:exp-corr}
    $\exists$ $\alpha\in(0,1)$ and a (possibly random) permutation $\sigma=\sigma(\Dk{1})$ such that with probability at least $1-o(1)$ over the first half of data $\Dk{1}$
    \begin{equation*}
    \rhomax\left(\left(s_{\sigma(i)}^{(2)},\widetilde s_{\sigma(i)}^{(2)}\right),\left(s_{\sigma(j)}^{(2)},\widetilde s_{\sigma(j)}^{(2)}\right) \mid \Dk{1}\right)\leq \alpha^{|i-j|} \quad \text{for all } 1\leq i,j\leq p, 
    \end{equation*}
    where $\rhomax\left(\cdot,\cdot \mid \Dk{1}\right)$ is the maximal correlation conditioned on $\Dk{1}$. 
\end{assumption}

Assumption \ref{ass:exp-corr} is slightly stronger than Assumption \ref{ass:weak-cor-among-nulls} in that it requires the conditional maximal correlation to be short-ranged. Nevertheless, Assumption \ref{ass:exp-corr} remains mild in many common designs; it is satisfied under independence and also holds for block-wise orthogonal structures with a fixed block size in Example \ref{ex:negation} with Gaussian noise, as discussed in \cite{ke2024power}.


\begin{proposition}[Bayes-optimality of $M_j^\star$]
\label{prop:optimality}
    For any fixed $\epsilon>0$ and ranking data $\Djr$, we run Algorithm \ref{alg:selecseq} with the nominal FDR level $q$ and any mirror statistics $M=(M_1,M_2,\cdots,M_p)$ of form \eqref{general-form-mr-stats}. Let $\beta(M,q)$ be 
    the fraction of the true positives returned by the algorithm over all non-null features. 
    Suppose $p_1/p \geq \gamma \in (0,1)$. Then, under Assumptions \ref{ass:exchanablity} and \ref{ass:exp-corr}, 
    \begin{equation*}
        \sup_{0<\epsilon<1-q}\underset{p\to\infty}{\liminf} \ \beta(M^\star,q+\epsilon) - \beta(M,q) \geq 0,
    \end{equation*}
    where $M^\star = (M^\star_1,M^\star_2,\cdots,M^\star_p)$ is the Bayes-optimal mirror statistics defined in \eqref{opt-mr-stats-general}.
\end{proposition}

\begin{example}
    In Example \ref{ex:negation}, the Bayes-optimal mirror statistic \eqref{opt-mr-stats-general} can be written as
    \begin{equation}
    \label{opt-mr-stats}
        M_j^\star = \log\left(\frac{\mathbbm{E}_{\boldsymbol{\beta}\mid D^{(1)}}p(\widehat\beta_j^{(2)}\mid \boldsymbol{\beta})}{\mathbbm{E}_{\boldsymbol{\beta}\mid D^{(1)}}p(-\widehat\beta_j^{(2)}\mid \boldsymbol{\beta})}\right).
    \end{equation}
    Note that \eqref{opt-mr-stats} assumes \( \Djr = (\boldsymbol{X}^{(1)}, \boldsymbol{y}^{(1)}, \{s_j^{(2)}, \widetilde{s}_j^{(2)}\}) \). Alternatively, if $\Djr = (\widehat{\beta}_j^{(1)},|\widehat{\beta}_j^{(2)}|)$,
    the Bayes-optimal mirror statistic becomes 
    \begin{equation}
    \label{subopt-mr-stats}
        \widetilde M_j = \log\left(\frac{\mathbbm{E}_{\beta_j\mid \widehat\beta_j^{(1)}}p(\widehat\beta_j^{(2)}\mid \beta_j)}{\mathbbm{E}_{\beta_j\mid \widehat\beta_j^{(1)}}p(-\widehat\beta_j^{(2)}\mid \beta_j)}\right).
    \end{equation}
    The key difference between the two statistics lies in the expectation: \eqref{opt-mr-stats} takes it with respect to the full posterior of $\boldsymbol{\beta}$, whereas \eqref{subopt-mr-stats} uses the marginal posterior of $\beta_j$, both conditioned on $D^{(1)}$.
\end{example}

\begin{example}
    In Example \ref{ex:modely}, the Bayes-optimal mirror statistic \eqref{opt-mr-stats-general} can be written as 
    \begin{equation}
    \label{opt-mr-stat-self-contrast}
    M_j^\star=\log\left(
    \frac{\pr_{(\Xk{2},\yk{2})}(\Xk{2},\yk{2}\mid \Dk{1})}
    {\pr_{(\Xk{2},\yk{2})}(\tXkj{2},\tykj{2}\mid \Dk{1})}
    \right).
    \end{equation}
    By construction, pairs $(\Xk{2},\yk{2})$ and $(\tXkj{2},\tykj{2})$ are exchangeable when feature $j$ is null, guaranteeing the symmetry property. \eqref{opt-mr-stat-self-contrast} gives us a self-contrasting view of the Bayes-optimal mirror statistic. The numerator is essentially the likelihood of the data whereas the denominator is the likelihood of the ``mirror" data. 
    Let $\boldsymbol{M}_{-j}=\boldsymbol{I}-\boldsymbol{P}_{-j}$ be the residual projection matrix onto the column $\boldsymbol{X}_{-j}^{(2)}$ as in Example \ref{ex:modely}.
    Thus, $\boldsymbol{M}_{-j}\yk{2}\sim N(\beta_j \boldsymbol{M}_{-j}\boldsymbol{X}_j, \sigma^2\boldsymbol{M}_{-j})$. Under alternative ($\beta_j\neq0$), $\boldsymbol{M}_{-j}\yk{2}$ tens to have a positive component along the direction $\sgn(\beta_j)\boldsymbol{M}_{-j}\boldsymbol{X}_j$. To make $\tyk{2}$ and $\tykj{2}$ exchangeable, one must make $\|\boldsymbol{M}_{-j}\yk{2}\|=\|\boldsymbol{M}_{-j}\tykj{2}\|$. Among all choices, $\boldsymbol{M}_{-j}\tykj{2}=-\boldsymbol{M}_{-j}\yk{2}$ makes the \emph{mirror} summary statistic land on the opposite side of the signal and therefore has a  low likelihood under the alternative. 
\end{example}



\begin{remark}
    \label{rmk:anti-conservatism}
    For any mirror statistic $M_j$ of the form \eqref{general-form-mr-stats}, 
    $\mathbbm{P}(M_j>0\mid |M_j|) = \sigma(|M_t|) =  \frac{e^{|M_j|}}{1+e^{|M_j|}}$, where $\sigma$ is the sigmoid function. In other word, larger magnitude suggests that the mirror statistic is more likely to be positive. Therefore, intuitively, ranking $M_j$ by magnitude is optimal. 
\end{remark}

\begin{remark}
When assigning all samples to the second half so that $D^{(1)}$ is empty, \eqref{opt-mr-stats} reduces to $M_j^\star=\log\!\big(m(\hat\beta_j)/m(-\hat\beta_j)\big)$, where $m(t)$ is the marginal density of $\hat\beta_j$. Suppose $\beta_j$ follows a spike-and-slab prior $p_0\nu_0+p_1F$ with some slab part $F$, and $\hat\beta_j\sim f_0(\hat\beta_j)$ under the null and $\hat\beta_j\sim f_1(\hat\beta_j)$ under alternative. The symmetry property guarantees $f_0(\hat\beta_j)=f_0(-\hat\beta_j)$.
Then, the two-groups prior $m(t)=p_0 f_0(t)+(1-p_0)f_1(t)$, the Lfdr \citep{efron2001empirical,efron2002empirical} is $\mathrm{Lfdr}(t)=p_0 f_0(t)/m(t)$, so $-\log \mathrm{Lfdr}(t)=\log m(t)-\log(p_0 f_0(t))$. Under alternative, when signal is strong, one would expect $f_1(-\hat\beta_j)$ to be small and $m(-\hat\beta_j)\approx p_0f_0(-\hat\beta_j)=p_0f_0(\hat\beta_j)$. Thus, $M^\star(\hat\beta_j)=\log m(\hat\beta_j)-\log m(-\hat\beta_j)\approx \log m(\hat\beta_j)-\log(p_0 f_0(\hat\beta_j))=-\log \mathrm{Lfdr}(\hat\beta_j) + \text{Constant}$. Importantly, unlike the oracle Lfdr procedure \citep{Storey2007TheOD}, applying Algorithm \ref{alg:selecseq} with \eqref{opt-mr-stats} enjoys asymptotic FDR control even when the prior is misspecified.
\end{remark}



\section{Power Analysis}
In this section, we compare the sign-sum mirror statistic studied in \cite{dai2023false, dai2023scale}:
\begin{align}
\label{vanilla-mr-stats}
M_j &= \sgn(\widehat\beta_j^{(1)}\widehat\beta_j^{(2)})(|\widehat\beta_j^{(1)}|+|\widehat\beta_j^{(2)}|),
\end{align} 
which was shown to be optimal among three choices \citep{ke2024power},
with the Bayes-optimal mirror statistics defined in \eqref{opt-mr-stats} and \eqref{subopt-mr-stats} for Example \ref{ex:negation}.
Here, $M_j$ represents the sign-sum mirror statistic without prior information in \eqref{vanilla-mr-stats}, $\widetilde M_j$ represents an optimal way to combine the two estimates $\widehat\beta_j^{(1)}$ and $\widehat\beta_j^{(2)}$ in \eqref{subopt-mr-stats}, and $M_j^\star$ is the Bayes-optimal mirror statistic incorporating prior information in \eqref{opt-mr-stats}. Note that in the independent case where ${\boldsymbol{X}^{(1)}}^\intercal \boldsymbol{X}^{(1)} = {\boldsymbol{X}^{(2)}}^\intercal \boldsymbol{X}^{(2)} \propto I_p$, $M_j^\star$ is equivalent to $\widetilde M_j$.

To compare different mirror statistics, we assume that the working model is linear:
\begin{equation*}
    \boldsymbol{y}^{(k)} = {\boldsymbol{X}^{(k)}}^\intercal \boldsymbol{\beta} + \boldsymbol{\epsilon}^{(k)} \in \mathbb{R}^{n_k}, \quad \boldsymbol{\epsilon}^{(k)}_i \overset{iid}{\sim} \mathcal{N}(0,I_{n_k}), \quad k=1,2.
\end{equation*}
Here, $n_k$ denotes the number of samples in the $k$-th split, so that $(\boldsymbol{X}^{(1)}, \boldsymbol{y}^{(1)})$ denotes the first half of the data and $(\boldsymbol{X}^{(2)}, \boldsymbol{y}^{(2)})$,  the second half. We consider the balanced split case, i.e., $n_1=n_2$, and the Rare/Weak signal model in \cite{donoho2004higher}, where $\boldsymbol{\beta}=(\beta_1,\beta_2,\cdots,\beta_p)$ follows a spike-and-slab prior distribution:
\begin{align}\label{RWmodel1} 
&\beta_j\; \overset{iid}{\sim}\; (1-\epsilon_p)\nu_0 + \epsilon_p\nu_{\tau_p}, \qquad 1\leq j\leq p,
\end{align}
where $\nu_a$ denotes a point mass at $a$. Here, $\epsilon_p\in (0,1)$ is the expected fraction of signals and $\tau_p>0$ is the signal strength. We let $p$ be the driving asymptotic parameter and tie $(\epsilon_p, \tau_p)$ with $p$ through fixed constants $\vartheta\in (0,1)$ and $r>0$:
\begin{equation} \label{RWmodel2}
\epsilon_p = p^{-\vartheta}, \qquad\tau_p = \sqrt{2r\log(p)}. 
\end{equation}
Parameters $\vartheta$ and $r$ characterize the rarity of the signal and the strength of the signal, respectively. Note that $\sqrt{\log p}$ is the correct scale, as it is the minimax order for a successful inference of the support of $\beta$ \citep{Genovese2012ACO}. 

The power of Algorithm \ref{alg:selecseq} with different mirror statistics depends on the nominal FDR control level $q$. Instead of fixing a specific level $q$, we adopt the phrase diagram in \cite{ke2024power} and compare the diagrams for different mirror statistics. In Algorithm \ref{alg:selecseq}, we choose the threshold  as $
    \tau_q = \min \{t: \widehat{\text{FDP}}(t)\leq t\}.$
Since we are not fixing any $q$, $\tau_q$ is not directly available to us. Instead, we vary the threshold and consider the optimal threshold that maximizes the power. For \eqref{vanilla-mr-stats}, like \cite{ke2024power}, we reparameterize the threshold as $t$ through
 $t(u) = \sqrt{2u\log p}.$
For the optimal mirror statistics \eqref{opt-mr-stats} and \eqref{subopt-mr-stats}, we reparameterize the threshold as 
$t(u) = 2u\log p.$

For a feature importance measure $W_j$, which can be one of the three choices: $M_j$, $\widetilde{M_j}$, or 
$M_j^\star\}$, we consider the selected features  $\widehat{S}_1(u) = \{1\leq j\leq p: W_j \geq t(u)\}.$
Let $S_1$ denote the set of truly relevant features, and define
$\text{FP}(u) = \mathbbm{E}(|\widehat{S}_1(u)-S_1|) \text{ and } \text{FN}(u) = \mathbbm{E}(|S_1-\widehat{S}_1(u)|)$
as the false positive and false negative, respectively. The expectation is taken with respect to the randomness of both prior and the data. Define
\begin{align*}
    &\text{Hamm}(u) = \text{FP}(u) + \text{FN}(u), \quad \text{Hamm}^\star = \inf_{u>0} \{\text{Hamm}(u)\}.
\end{align*}
The first quantity is the expected Hamming error, and the second quantity is its minimum.
For each mirror statistic, there exists a bivariate function $f_{\text{Hamm}}^\star(\vartheta,r)\in [0,1]$ such that $\text{Hamm}^\star = L_p p^{f_{\text{Hamm}}^\star(\vartheta,r)},$
where $L_p$ is the multi-log term, commonly used in the study of Rare/Weak models \citep{Genovese2012ACO, ji2012ups}. The phase diagram is defined below.
\begin{definition}[Phrase Diagram]
    The phase diagram is defined to the partition of the two-dimensional space $(\vartheta,r)$ into three regions:
    \begin{itemize}
        \item Region of Exact Recovery (ER): $\{(\vartheta,r):f_{\text{Hamm}}^\star(\vartheta,r)<0\}$,
        \item Region of Almost Full Recovery (AFR): $\{(\vartheta,r): 0\leq f_{\text{Hamm}}<1-\vartheta\}$,
        \item Region of No Recovery (NR): $\{(\vartheta,r): f_{\text{Hamm}}\geq 1-\vartheta\}$.
    \end{itemize}
    The curves separating different regions are called phase curves. We use $h_{\text{AFR}}(\vartheta)$ to denote
the curve between NR and AFR, and $h_{\text{ER}}(\vartheta)$
the curve between AFR and ER.
\end{definition}
In the ER region, the $\text{Hamm}^\star$ tends to zero. Therefore, with
high probability, the support of $\beta$ is exactly recovered. In the AFR region, $\text{Hamm}^\star$ does not tend to zero but is much smaller than the expected number of relevant features. In this region, the majority of signals are recovered. In the region of NR, $\text{Hamm}^\star$ is comparable with the number of relevant features, suggesting failure of feature selection algorithm.


\subsection{The oracle prior}

We first consider the oracle case, where we know how the true $\boldsymbol{\beta} =(\beta_1,\beta_2,\cdots,\beta_p)$ is generated \textit{a priori}. The following proposition compares the optimal Hamming error of $M_j$ and $\widetilde M_j$.
\begin{proposition}[Hamming error for \eqref{subopt-mr-stats}]
\label{prop:M-tildeM-comparison}
Consider a linear regression model where $\boldsymbol{\beta}$ satisfies Models~\eqref{RWmodel1}-\eqref{RWmodel2}. 
 Let $\omega_j = \Sigma^{-1}_{jj}$ where $\boldsymbol{X}^{(1)^\intercal}\boldsymbol{X}^{(1)} = \boldsymbol{X}^{(2)^\intercal}\boldsymbol{X}^{(2)} = \Sigma$. Suppose $\omega_j=\theta$ for all $j$.
Then the Hamming rates for $M$ and $\widetilde{M}$ are the same: 
\begin{equation}
    \label{eq:prop-cmp-tilde-ms-with-vanilla-ms-rate}
    f^\star_{\text{Hamm}, \widetilde M} = f^\star_{\text{Hamm}, M}.
\end{equation}
However, for the Hamming errors, we have
\begin{equation}
    \label{eq:prop-cmp-tilde-ms-with-vanilla-ms-hamming}
    \underset{p\to \infty}{\limsup} \ \frac{\text{Hamm}_{\widetilde M}^\star}{\text{Hamm}_{M}^\star} \leq 1.
\end{equation}

\end{proposition}

In the proof of Proposition \ref{prop:M-tildeM-comparison}, a crucial geometric insight is that there exist two normal random variables $(U,V)$ with shrinking variance such that the rejection rule $\{1\leq j\leq p: W_j\geq t(u)\}$ with optimal $u$ can be characterized by some region in $(U,V) \in \mathbb{R}^2$ as illustrated in Figure \ref{fig:hamming-errors-diagram}. Furthermore, if $\beta_j$ is null, $(U,V)$ are two centered normal random variables, while if $\beta_j$ is non-zero (specifically if $\beta_j=\tau_p$), 
$(U,V)$ are two normal r.v.s centered at $(2\tau_p,0)$, where $\tau_p = \sqrt{2r\log p}$. From this observation, we can show that:
\begin{align*}
    &\mathbbm{P}(j\in \widehat{S}_{\widetilde{M}} \mid \beta_j=0) \leq 
    \mathbbm{P}(j\in \widehat{S}_{M} \mid \beta_j=0) \\
    &\mathbbm{P}(j\in \widehat{S}_{\widetilde{M}} \mid \beta_j=\tau_p) \geq 
    \mathbbm{P}(j\in \widehat{S}_{M} \mid \beta_j=0),
\end{align*}
where $\widehat{S}_{M}$ and $\widehat{S}_{\widetilde{M}}$ are selected sets with optimal $u$ for the mirror statistic $M_j$ and $\widetilde{M}_j$ respectively. From this, we see that  $\text{Hamm}_{\widetilde M}^\star\leq\text{Hamm}_{M}^\star.$

For $M^\star$, we utilize the entire first half of the data to construct the posterior distribution of $\beta_j$ on all entire second half of the data instead of only on the summary statistic $\widehat\beta_j^{(1)}$. This allows us to leverage the correlation structure among the design matrix and the prior information on the $\beta_j$'s. Therefore, the ranking of  the $M_j^\star$ can potentially be better than that of the $\widetilde M_j$ when the design matrix is correlated and when the prior distribution is nontrivial. Note that it is impossible to obtain an explicit formula for $M_j^\star$ under a general covariance structure since it involves summing over $2^p$ terms. Thus, we consider an idealized setting where the Gram matrices of two halves of the data, $G^{(1)} = {\boldsymbol{X}^{(1)}}^\intercal \boldsymbol{X}^{(1)}$ and $G^{(2)} = {\boldsymbol{X}^{(2)}}^\intercal \boldsymbol{X}^{(2)}$, are block-wise diagonal, consisting of $2\times 2$ blocks:
\begin{equation}
\label{eq:blockwise-definition}
    G^{(1)} = G^{(2)} =\frac{1}{2}\mathrm{diag}\bigl(B, B,\ldots, B\bigr), \quad \text{where}\quad B = \begin{pmatrix}
        1 & \rho\\
        \rho & 1
    \end{pmatrix}. 
\end{equation}
From now on, we consider only even value of $p$. 

\begin{figure}[tb!]
\centering
\includegraphics[width=1.0\textwidth]{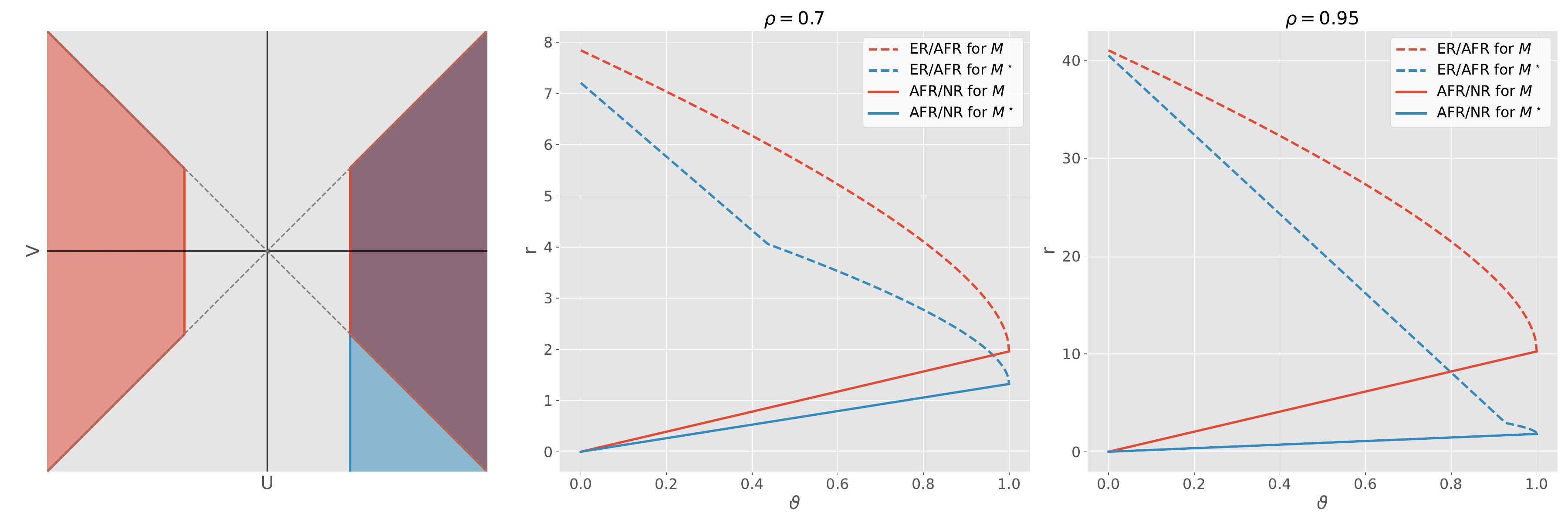}
\caption{Left: The optimal rejection rule for $M$ and $\widetilde{M}$. Orange and purple areas mark the rejection regions for $M$ while blue and purple areas mark the rejection region for $\widetilde{M}$. Middle: Phase diagram for $M$ and $M^\star$ with $\rho=0.7$. Right: Phase diagram for $M$ and $M^\star$ with $\rho=0.95$. In both the middle and right figures, dashed lines separate ER and AFR regions, while solid lines separate AFR and NR region. Red lines correspond to $M_j$, and blue lines correspond for $M^\star$. For both $M$ and $M^\star$, the regions from top to bottom are ER, AFR, NR, respectively.}
\label{fig:hamming-errors-diagram}
\end{figure}

\begin{proposition}[Hamming error for \eqref{opt-mr-stats}]
\label{prop:tildeM-Mstar-comparison}
    Consider a linear regression model where $\boldsymbol{\beta}$ satisfies Models \eqref{RWmodel1}-\eqref{RWmodel2}. Suppose the Gram matrices of the two halves of the data have the form of \eqref{eq:blockwise-definition}. Then, as $p\to\infty$,
    \begin{align}
        &f^\star_{\text{Hamm}, M^\star}
        \leq
        1 - \vartheta
        -\min \left(\frac{((1-\rho^2/2)r-\vartheta)_+^2}{4(1-\rho^2/2)r},
        \frac{(1-|\rho|)(3+|\rho|)r}{8}
        \right).
    \end{align}
    This leads to
    \begin{equation}
        f^\star_{\text{Hamm}, M^\star}\leq f^\star_{\text{Hamm}, \widetilde M} = f^\star_{\text{Hamm}, M} = 1-\vartheta-\frac{((1-\rho^2)r-\vartheta)_+^2}{4(1-\rho^2)r}.
    \end{equation}   
    where the inequality becomes strict when $r>\vartheta/(1-\rho^2/2)$ and $\rho \neq 0$.
\end{proposition}

Although Proposition \ref{prop:M-tildeM-comparison} states that the Hamming error of $\widetilde M$ is smaller than that of $M$, it also indicates that the former is only better than the latter by some constant (or at most a multi-log term), and their convergence rates of Hamming error remain the same. Thus, $\widetilde M_j$ and $M_j$ have the same phase diagram. On the other hand, the convergence rate of Hamming error of $M^\star_j$ is smaller than that of $\widetilde M_j$ and $M_j$ thus $M^\star_j$ will have a better phase diagram compared to the other two mirror statistic. As shown in Figure \ref{fig:hamming-errors-diagram}, the power gain becomes larger as $|\rho|$ increases. 



\subsection{Structured prior}
Another scenario where $M^\star$ is useful is when the prior of $\boldsymbol{\beta}$ is structured, e.g., the $\beta_j$'s have a group structure such that the $\beta_j$'s in the same group are either all relevant or all null. Using the terminology of the Rare/Weak signal model, we assume
\begin{equation}
\label{RWmodel-grouped} 
\beta_{2j}\; \overset{iid}{\sim}\; (1-\epsilon_p)\nu_0 + \epsilon_p\nu_{\tau_p},\quad \beta_{2j-1} = \beta_{2j},
\quad 1\leq j\leq p/2.
\end{equation}

\begin{proposition} \label{prop:mcmc-best-group-case}
Consider a linear regression model where $\beta$ satisfies Models \eqref{RWmodel2} and \eqref{RWmodel-grouped}. 
Suppose the Gram matrices of the two halves of the data have the form $G^{(1)} = G^{(2)} = I/2$. 
Then as $p\to\infty$,
\begin{equation}
    f_{\text{Hamm}, M^\star}^\star = 1-\vartheta-
    \frac{(3r - 2\vartheta)_+^2}{24r}\leq  f_{\text{Hamm}, \widetilde M}^\star = f_{\text{Hamm}, M}^\star = 1-\vartheta- \frac{(r-\vartheta)_+^2}{4r}.
\end{equation}
\end{proposition}

Similar to Proposition \ref{prop:tildeM-Mstar-comparison}, $M^\star_j$ has a better convergence rate of Hamming error compared to $\widetilde M_j$ and $M_j$ when the true $\boldsymbol{\beta}$ has a group structure, since the latter two ignore the group structure in the prior knowledge, thus their Hamming rates stay the same. 





\subsection{Misspecified prior}
Instead of setting the oracle prior as in \eqref{RWmodel2}, we consider here that the prior is misspecified as
\begin{equation}
\label{eq:misspecified-prior} 
\beta_j\; \overset{iid}{\sim}\; (1-\epsilon_p)\nu_0 + \epsilon_p\nu_{\tau_p^{\prime}},\quad \tau_p^\prime = \sqrt{r^\prime\log(p)},\quad  1\leq j\leq p.
\end{equation}
where $r^{'}\neq r$. Then, we have the following proposition.

\newcommand{\error}{\text{FN $+$ FP}}
\begin{proposition}
\label{prop:mis-specified-prior-hamming-rate}
Consider a linear regression model where $\boldsymbol{\beta}$ satisfies Models \eqref{RWmodel1}-\eqref{RWmodel2}.
Suppose the Gram matrices of the two halves of the data have the form $G^{(1)} = G^{(2)} = I/2$.
Then as $p \to \infty$,
\begin{equation}
    f^\star_{\text{Hamm},  M^\star} \geq 1 - \vartheta - \frac{(r-\vartheta)^2_+}{4r}.
\end{equation}

\end{proposition}

From this proposition we know that having a consistent estimation of the prior information is crucial. Otherwise, the Hamming error of \eqref{opt-mr-stats} could even be worse than that of \eqref{vanilla-mr-stats}.

\section{Prior-Assisted Data Splitting}
\label{sec:bous}

Motivated by the Bayes-optimal mirror statistic derived in Section~\ref{sec:pads}, we emphasize that the two halves of the data serve asymmetric roles. In our two-stage construction, the first half is used for inference: it is used to update the prior belief via Bayes' rule, yielding an ``optimal guess'' about the true parameter $\theta$. It may also be used to prescreen features that are unlikely to be non-null. In contrast, the second half is reserved for calibration, in the sense that it is used to form the sign-flip or the mirroring component that enforces the symmetry required by FDR control procedures. This asymmetry has a direct implication for how one should allocate samples across the two halves.

First, the strength of the prior information can favor an unbalanced split. For instance, in a linear model where the prior is already highly informative about the truth, allocating more observations to the first half yields limited gains in power.
In such regimes, assigning more samples to the second half can be beneficial. 
Second, when the prior itself must be learned from the data (e.g., an empirical-Bayes prior with unknown variance or other hyperparameters), the roles reverse: more data in the first half can improve the estimation of the prior and the accuracy of  the resulting posterior ranking, outweighing the loss of calibration samples. 
While characterizing the optimal allocation in full generality is difficult, the  above discussion motivates us to develop a data-driven mechanism that adapts the split ratio to the problem at hand.

\subsection{Adaptive Data Splitting as optional stopping}

As we have discussed earlier, the power of a data-splitting method depends on the chosen split ratio. We now introduce a data-driven approach to aggregate mirror statistics across multiple split ratios and adaptively select the optimal ratio for each feature. Suppose we gradually increase the amount of data assign to the first half. Let $0 < r_1 < r_2,\cdots,<r_T < 1$ be the candidate splitting ratios, i.e., we assign the first $\lfloor r_t n\rfloor$ samples to the first half for $1\leq t\leq T$.
For each feature $j$ and $t$, denote $\Djkt{1}$ as the first half of data for ratio $r_t$, i.e., $(\boldsymbol{X}_{1:\lfloor r_tn \rfloor,\cdot}, \boldsymbol{y}_{1:\lfloor r_tn \rfloor})$ and $\Djrt \supseteq \Djkt{1}$ as the corresponding \emph{ranking data} for feature $j$. Naturally, $\{\Djkt{1}\}_{t=1}^T$ forms a \emph{filtration}.

Now our goal is to find a good stopping time $\tau$ that adapts to the filtration $\{\Djkt{1}\}_{t=1}^T$ such that the corresponding mirror statistic using $\Djktau{1}$ as the first half of the data not only controls FDR asymptotically, but also maximizes the number of discoveries.  We must address two issues: (i) whether the symmetry of the mirror statistic is preserved at the stopping time $\tau$ under the null, ensuring FDR control as in the fixed-ratio case; and (ii) how to identify the stopping time $\tau$ that maximizes the number of discoveries.


For simplicity, we will omit the subscript $j$ for the remainder of this section.
Let $M^t = r^t \eta^t = r^t \cdot \text{logit} \left( \mathbbm{P}(r^t=1 \mid \Drt)\right)$, denoting the mirror statistic at time $t$. For any stopping time $\tau$, we define the \emph{ADaptive Mirror Statistic (ADMS)} $M^\tau = r^\tau \eta^\tau$. The following proposition guarentees the symmetry property and the sign-magnitude relation for ADMS.

\begin{proposition}
    \label{prop:sign-max-general}
    For any mirror statistic $M$ in \eqref{general-form-mr-stats} and any $\tau$ adapted to filtration $\{\Dkt{1}\}_{t=1}^T$, $M^\tau$ is symmetric around zero under the null. Furthermore, 
    $\mathbbm{P}(M^\tau>0 \mid |M^\tau|) = \sigma(|M^\tau|)=\frac{e^{|M^\tau|}}{1+e^{|M^\tau|}}$.
\end{proposition}

Note that for any $1\leq t \leq T$, $\tau_t \equiv t$ is a stopping time adapted to the filtration. Trivially setting $\tau\equiv t$ allows ADMS to include mirror statistic \eqref{general-form-mr-stats} for any split ratio. According to Proposition \ref{prop:sign-max-general}, any ADMS is symmetric under the null. Next, let us address the second challenge.

Proposition \ref{prop:sign-max-general} also provides a sign-magnitude relation for ADMS, which has been discussed in Remark \ref{rmk:anti-conservatism} for mirror statistic having form \eqref{general-form-mr-stats}. Intuitively, to maximize the number of potential discoveries, one must make $|M^\tau|$ as large as possible. Intuitively, we should stop if and only if we find convincing evidence for $|M^\tau|$ being large enough. At time $t$, we can view $R_t = \E[\sigma(|M^t|)\mid \Dkt{1}]$ as a surrogate for $|M^\tau|$. Therefore, if $R_t$ is large, we stop; otherwise, we keep increasing $t$ until $t=T$. 

\begin{definition}[ADMS--Thr($\ell$)]
\label{def:ADMS--Thr}
    For any predefined threshold $\ell \in(0.5,1)$, let $\tau_\ell$ be the stopping time $\tau_\ell = \inf \{1\leq t\leq T: R_t>\ell\}$; if it never stops, we define $\tau_\ell=T$. Then, we call the resulting adaptive mirror statistic $M^\tau$ the \emph{ADMS--Thr($\ell$)}.
\end{definition}

The stopping time $\tau_\ell$ gets smaller when we decrease $\ell$. Note that $\sigma(|M^t|)\in(0.5,1)$.  When $\ell\to 0.5$, ADMS--Thr($\ell$) becomes $M^{t_1}$, whereas when $\ell\to 1$, ADMS--Thr($\ell$) becomes $M^{t_T}$. It is also worth noting that the stopping time $\tau_\ell$ is feature-wise, i.e., $\tau_\ell$ varies across features. Definition \ref{def:ADMS--Thr} may make one feel that we are aggressively finding evidence against the null, which may break down the FDR control. However, such an exploration is valid due to Proposition \ref{prop:sign-max-general}. As we will demonstrate empirically in Section \ref{sec:simulations}, $\ell \in (0.6,0.8)$ is generally a good choice for optimizing power.

\subsection{Optimal data-splitting via Snell envelope}
Although ADMS--Thr($\ell$) generally works  well, it is still suboptimal and not fully data-driven. This imposes another challenge: is there a way to identify the optimal stopping for maximizing the number of discoveries without relying on any tuning parameter? 
When features are independently and identically sampled, and the number of features goes to infinity, the discovery power is determined by the cumulative distribution function of $|M^\tau|$. In fact, Algorithm~\ref{alg:selecseq} will reject features with their mirror statistics smaller than $c_\alpha = \inf 
\{d:\mathbbm{E}\mathbbm[\sigma(|M^\tau|)\mathbbm{1}(|M^\tau|>d)]\geq 1-\alpha \}$ for a predefined nominal FDR control level $\alpha\in(0,1)$. Thus, to maximize the number of selected features, we need to minimize $c_\alpha$. Therefore, the stopping time that makes the most selections maximizes $\mathbbm{E}\mathbbm[|M^\tau|\mathbbm{1}(|M^\tau|>c_\alpha)]$ for some $c_\alpha$, which is the portion of the features being rejected. Here, we are not limiting ourselves in maximizing the power for a fixed $\alpha$. Instead, for a given positive non-decreasing weight function $\omega(\cdot)$, we intend to maximize $\mathbbm{E}[\omega(|M^\tau|)]$ over all stopping times. Interestingly, the optimal stopping time $\tau^\star$ can be found by constructing a sequence of random variables known as the \textit{Snell envelope} \citep{snell1952applications} via dynamic programming described as follows: 

\textbf{Step 1}. Define the \textit{reward} variables $R_t = \mathbbm{E}[\omega(|M^t|)\mid \Dkt{1}]$ for $t=1,\cdots,T$. 

\textbf{Step 2}. Define $Q_T=R_T$ and inductively for every $t<T$, $Q_t = \max\{R_t, \mathbbm{E}[Q_{t+1} \mid \Dkt{1}]\}$.

\textbf{Step 3}. Compute $\tau^\star$ by $\tau^\star = \min\{t: R_t\geq Q_t\}$.

\begin{proposition}
\label{prop:opt-stop-time}
    For any weight function $\omega(\cdot)$, $\tau^\star = \underset{\tau}{\arg\max}\ \E [\omega(|M^\tau|)]=\underset{\tau}{\arg\max}\ \E[R_\tau]$. 
\end{proposition}

We call the resulting mirror statistic $M^{\tau^\star}$ ADMS--Snell. It is flexible in that the users are allowed to select any monotone increasing function $\omega(\cdot)$ with default choice $\omega(\cdot) = \sigma(\cdot)$. With such a default choice, ADMS--Snell works well as demonstrated later in Section \ref{sec:simulations}. However, the definition of ADMS--Snell imposes a computational challenge. To see this, we note that, in order to approximate $R_t$, one typically has to do Monte Carlo sampling $K$ times, resulting in $O(K)$ time complexity. This will make the computing time for $Q_{T-1}$ to be $O(K^2)$. Inductively, computing ADMS--Snell needs $O(K^T)$ time, which is impossible in practice for large dimensions. 

To address this computational bottleneck, we adopt the Longstaff--Schwartz least-squares Monte Carlo approach \citep{longstaff2001valuing}. Rather than approximating the conditional continuation value $\E\!\left[Q_{t+1}\mid \Dkt{1}\right]$ via nested Monte Carlo, we replace it with an $L^2$ projection of $Q_{t+1}$ onto a finite-dimensional linear sieve spanned by basis functions of the time-$t$ state. Concretely, let
$\phi(\Dkt{1})
=
\big(\phi_1(\Dkt{1}),\ldots,\phi_d(\Dkt{1})\big)^\top$
denote a user-specified feature map (e.g., low-order polynomials of summary statistics). Given $K$ simulated trajectories $\{(\Dkt{1})^{(k)},Q_{t+1}^{(k)})\}_{k=1}^K$, we estimate the continuation value by the fitted regression
\[
\widehat{\E}\!\left[Q_{t+1}\mid \Dkt{1}\right]
=
\phi(\Dkt{1})^\top \hat{\beta}_t,
\qquad
\hat{\beta}_t \in \arg\min_{\beta\in\R^d}
\sum_{k=1}^K
\Big(
Q_{t+1}^{(k)}-\phi((\Dkt{1})^{(k)})^\top\beta
\Big)^2.
\]
This regression-based approximation reduces the time complexity from $O(K^T)$ to $O(KT)$, making ADMS--Snell computationally scalable.

\section{Simulations}
\label{sec:simulations}


\subsection{Many normal-means model}

We consider the many normal-means model \citep{Johnstone2011GaussianEstimation,Dobriban2015OptimalMultipleTesting}: $\boldsymbol{y}_{ij} = \boldsymbol{\mu}_j + \boldsymbol{\epsilon}_{ij}, \quad 1\leq i\leq n, 1\leq j\leq p$, where $\boldsymbol{\epsilon}_{ij} \overset{i.i.d.}{\sim} N(0, I_d)$. In this setting, a ``feature'' $j$ is null if and only if $\boldsymbol{\mu}_j=\boldsymbol{0}_d$ and we are intended to select those with $\boldsymbol{\mu}_j\neq \boldsymbol{0}_d$.
Throughout, we set $d=5,p=1000,n=20$, and consider $\boldsymbol{\mu}_j \overset{i.i.d.}{\sim} \pi(\boldsymbol{\mu_j})$, with a spike-and-slab prior 
$\pi = p_0\nu_{\boldsymbol{0}_d} + (1-p_0) F$ \citep{MitchellBeauchamp1988,GeorgeMcCulloch1993,JohnstoneSilverman2004}. 
Here, $F$ could either be a normal distribution $F=N(\mu_0\boldsymbol{1}_d,\tau^2 I_d)$ (one-sided) or a mixture normal $F=\frac{1}{2}N(\mu_0\boldsymbol{1}_d,\tau^2 I_d)+\frac{1}{2}N(-\mu_0\boldsymbol{1}_d,\tau^2 I_d)$ (two-sided).
We consider the following settings:
\begin{itemize}
    \item \emph{One-sided}: The data are generated from a spike-and-slab prior where $p_0=0.7$ and a one-sided slab $F=N(0.4\boldsymbol{1}_d,0.1^2I_d)$.
    \item \emph{Two-sided}: The data are generated from a spike-and-slab prior where $p_0=0.8$ and a two-sided slab $F=\frac{1}{2}N(0.6\boldsymbol{1}_d,0.1^2 I_d)+\frac{1}{2}N(-0.6\boldsymbol{1}_d,0.1^2 I_d)$.
    \item \emph{$p_0$-misspecified}: The true data are generated from a spike-and-slab prior where  $p_0=0.85$ and $F=\frac{1}{2}N(0.6\boldsymbol{1}_d,0.15^2 I_d)+\frac{1}{2}N(-0.6\boldsymbol{1}_d,0.15^2 I_d)$, whereas the working prior sets $p_0=0.6$.
    \item \emph{$\mu_0$-misspecified}: The true data are generated from a spike-and-slab prior where  $p_0=0.8$ and $F=\frac{1}{2}N(0.8\boldsymbol{1}_d,0.15^2 I_d)+\frac{1}{2}N(-0.8\boldsymbol{1}_d,0.15^2 I_d)$, whereas the working prior sets $\mu_0=0.5$.
\end{itemize}

We compare the Bayes-optimal mirror statistic \eqref{opt-mr-stats} with varying split ratio ($r=0.1,0.2,\cdots,0.8$) (OMS), ADMS--Thr($\ell$) with $\ell$ varying from $0.6$ to $0.8$, ADMS--Snell with default choice $\omega(\cdot)=\sigma(\cdot)$, sign-sum mirror statistic  (SM) and Lfdr. For sign-sum mirror statistic, we naively generalize \eqref{vanilla-mr-stats} to $\sgn\left(\bar{\boldsymbol{y}_j}^{(1)\intercal}\bar{\boldsymbol{y}_j}^{(2)}\right)\left(
\|\bar{\boldsymbol{y}}_j^{(1)}\|+\|\bar{\boldsymbol{y}}_j^{(2)}\|\right)$.
For ADMS--Thr($\ell$) and ADMS--Snell, we set the set of candidates for the split ratio as $r=0.1,0.2,\cdots,0.8$. For Longstaff-Schwartz regression in ADMS--Snell, we use lasso regression with regularization $\lambda$ being set as 1.0, taking the current reward $R_t$, the posterior mean $\E[\boldsymbol{\mu}\mid \Dkt{1}]$, the posterior uncertainty $\sqrt{\text{Tr}(\text{Var}[\boldsymbol{\mu}\mid \Dkt{1}])}$ as covariates. For Lfdr, we select features with the smallest Lfdr values such that the average Lfdr for the selected features does not exceed the nominal FDR level $q$ \citep{oraclecompoundfdr2007}.

As shown in Figure \ref{fig:many-normal-means}, split ratios affect the power of OMS significantly. Although balanced data-splitting seems to work well under two sided prior, when prior information is stronger, like in the one-sided case, the performance will reduce.
Throughout, ADMS--Snell consistently achieves the power comparable to Lfdr while controlling FDP even under misspecified settings. In contrast, we can see Lfdr significantly breaks FDR control in \emph{$p_0$-misspecified} setting.

\begin{figure}[t]
  \centering
  \scalebox{1}[1.0]{%
    \includegraphics[width=\textwidth]{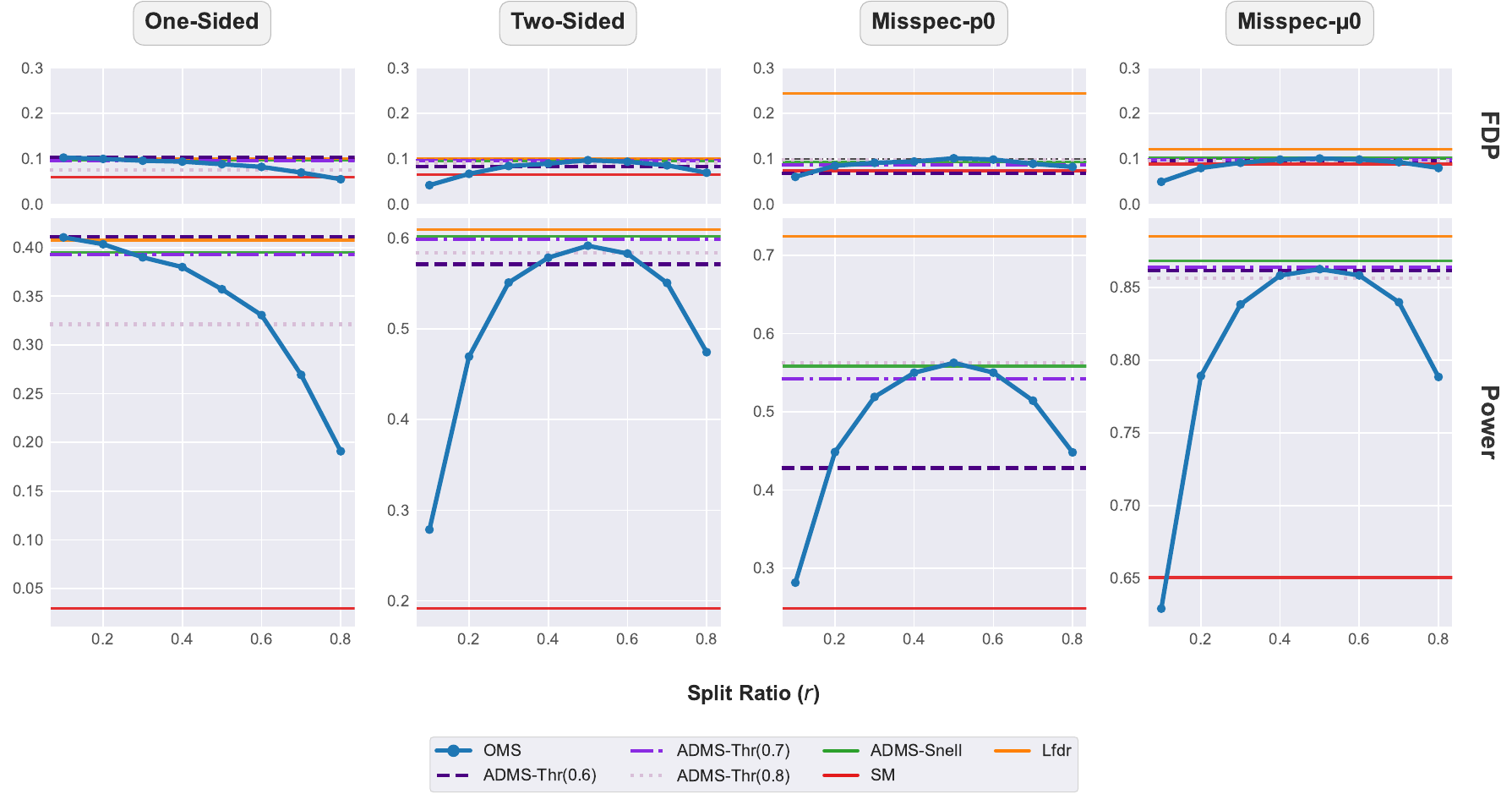}%
  }
  \caption{Power and FDP control ($q=0.1$) for different methods in four different settings.}
  \label{fig:many-normal-means}
\end{figure}

\subsection{Linear regression}
\label{subsec:liner}

We consider high-dimensional linear regression, where we simulate the response vector $\boldsymbol{y}_{n\times 1}$ from the model $\boldsymbol{y}_{n\times 1} = \boldsymbol{X}_{n\times p} \boldsymbol{\beta}_{p\times 1}^\star + \boldsymbol{\epsilon}_{p\times 1}$ and $\boldsymbol{\epsilon}_{p\times 1} \sim \mathcal{N}(0, I_n)$. 
Throughout, we set $n=1200$ and $p=2000$. Each row of the design matrix $\boldsymbol{X}_{n\times p}$ is generated
from a multivariate normal distribution with mean zero and covariance matrix $\Sigma$, which is block-diagonal. In each block, the diagonal entries are one, and the off-diagonal entries decay linearly from $\rho$ to zero.
We adopt the Bayesian setting for the true coefficients: the $\beta_j^\star$'s are i.i.d. drawn from a spike-and-slab prior $\pi$ with a Gaussian slab component, i.e.,
$p_0 \cdot \nu_0 + (1-p_0)\cdot \mathcal{N}(\delta_\mu \sqrt{\log(p)/n}, (\delta \sqrt{\log(p)/n})^2).$

We investigate a few recently developed p-value-free FDR-controlling methods, including DS with the sign-sum mirror statistic \eqref{vanilla-mr-stats} (SM), 
the MLR knockoff \citep{spector2022asymptotically} (MLR), DS with the Bayes-optimal mirror statistic with balanced split (bOMS), DS with the Bayes-optimal mirror statistic with the best unbalanced split (bnbOMS), and DS with ADMS--Snell adaptive split. 
We investigate effects of different correlation structures $\Sigma$ and signal strengths, $\delta$ and $\delta_\mu$, on the power and FDP control for these methods. For DS-based methods, we use the first half of the data to prescreen features with large inclusion rates. The number of the candidate features being selected is set to be half of the size of the second half\footnote{This prescreen procedure is different from \cite{dai2023false}, where Lasso with 10-fold cross-validation is used. Lasso prescreening may not satisfy the sure screening property \citep{fan2008sureindependencescreeningultrahigh} when the true non-null coefficients $\boldsymbol{\beta}^\star$ are not ultra sparse.}.
Throughout, we consider the nominal FDR level $q=0.1$. We examine the following two cases:

\begin{itemize}
    \item \emph{Weak prior information.} When $\delta_\mu = 0$, we call the prior relatively weak, since the sign of $\beta_j$ is not known to us upon knowing $\beta_j \neq 0$. We set $\delta=3.0, \delta_\mu =0.0$.
    \item \emph{Strong prior information.} When $\delta=0$, the slab part degrades to a fix point distribution at $\delta_\mu \sqrt{\log(p)/n}$. 
    In this case, we immediate know the sign of $\beta_j$ upon knowing $\beta_j \neq 0$. 
    We set $\delta=0.0, \delta_\mu =3.0$.
   
\end{itemize}

To mitigate power loss caused by unwanted randomness in FDR-controlling methods we considered, such as the random split in DS and the randomly generated knockoffs in knockoff-based methods, a rich body of literature focuses on stabilizing results across multiple runs. For instance, \cite{dai2023false} proposed the multiple data-splitting algorithm, where features are ranked by their inclusion rates, which are selection frequencies adjusted by selection sizes across multiple replications of DS. Building on e-values, \cite{ren2024derandomised} proposed a derandomized knockoff procedure to combine knockoff statistic from different runs into an e-value for selection. Throughout, we apply a newly developed stabilizing algorithm from \cite{sun2025generalstabilityapproachfalse}, where a stabilized e-value is constructed by aggregating rank statistics generated from multiple runs of the base algorithm, replacing unstable components with aggregate estimates. We compare various FDR-controlling methods under the same conditions as in the single-run simulation. As a baseline, we also assess power and FDP by applying the stabilizing algorithm to DS across all candidate split ratios, referred to here as naiveMDS.

As shown in Figure~\ref{fig:linear}, in both weak and strong prior cases, ADMS demonstrates the highest power among all DS-based methods, with power gains increasing as feature correlation grows. Moreover, applying stabilizing algorithm to ADMS results in a lower FDR level compared to other DS-based methods without suffering any power loss.

\begin{figure}[t]
\centering
\makebox[\textwidth][c]{%
  \scalebox{1}[0.85]{\includegraphics[width=1.0\textwidth]{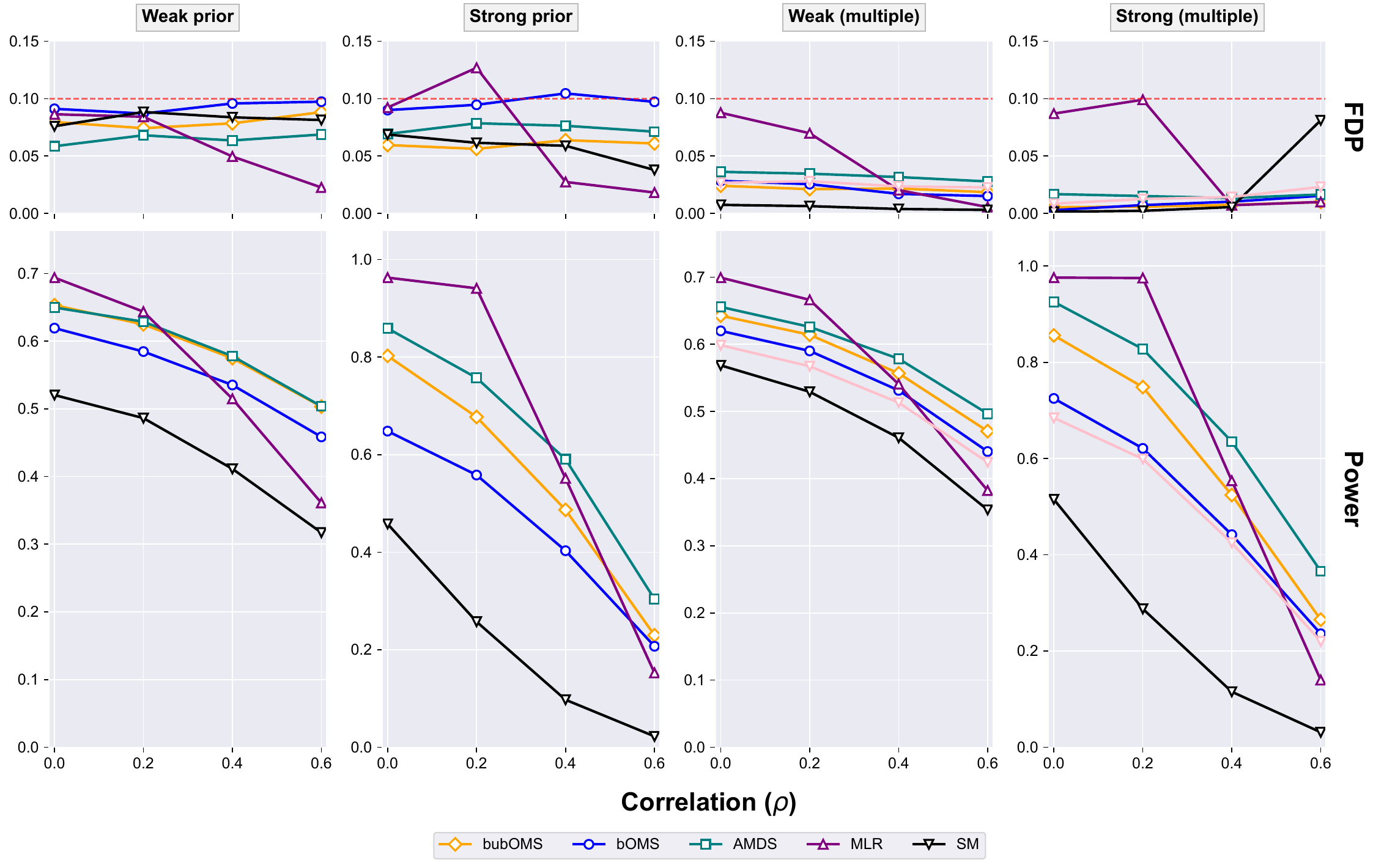}}
}
\caption{Linear regression model with $n=1200,p=2000$ and an expected number of relevant features $p_1=100$. Each row of the design matrix is generated from $N(0, \Sigma(\rho))$. The true coefficients $\beta_j^\star$ are drawn from $\mathcal{N}(0,3^2)$ in the weak prior case (Left2 and Right1) and $\mathcal{N}(3,0^2)$ in the strong prior case (Left1 and Right2). In all panels, we fix the signal and vary the correlation $\rho$. The nominal FDR control level $q=0.1$. For each multiple run, we replete the base algorithm 10 times. Each dot represents the average over 100 independent runs. }
\label{fig:linear}
\end{figure}

\subsection{Logistic regression}
We consider logistic regression, where we simulate the response vector $\boldsymbol{y}_{n\times 1}$ from the logistic model with link function $\boldsymbol{\eta}_{n\times 1} = \boldsymbol{X}_{n\times p} \boldsymbol{\beta}_{p\times 1}^\star$ and $\boldsymbol{y}_{p\times 1} \sim \text{Binomial}(\boldsymbol{\eta}_{p\times 1})$. 
Each row of $\boldsymbol{X}_{n\times p}$ is generated from a centered normal distribution with a block-diagonal covariance matrix $\Sigma$. In each block, the diagonal entries are one, and the off-diagonal entries linearly decay from $\rho$ to zero. We also consider a sparse setting, where the expected number of non-zeros in $\boldsymbol{\beta}_{p\times 1}^\star$ is 30. We adapt a Bayesian setting for $\beta_j^\star$ are i.i.d. drawn from a spike-and-slab prior with a centered Gaussian component $p_0 \cdot \nu_0 + (1-p_0)\cdot \mathcal{N}(0 , \tau^2)$.
We compare different FDR-controlling methods in various scenarios using the same setup as in Section~\ref{subsec:liner}. As shown in Figure~\ref{fig:logit}, ADMS performs the best among all DS-based methods especially in the strong prior cases.

\begin{figure}[t]
\makebox[\textwidth][c]{%
  \scalebox{1}[0.85]{\includegraphics[width=1.0\textwidth]{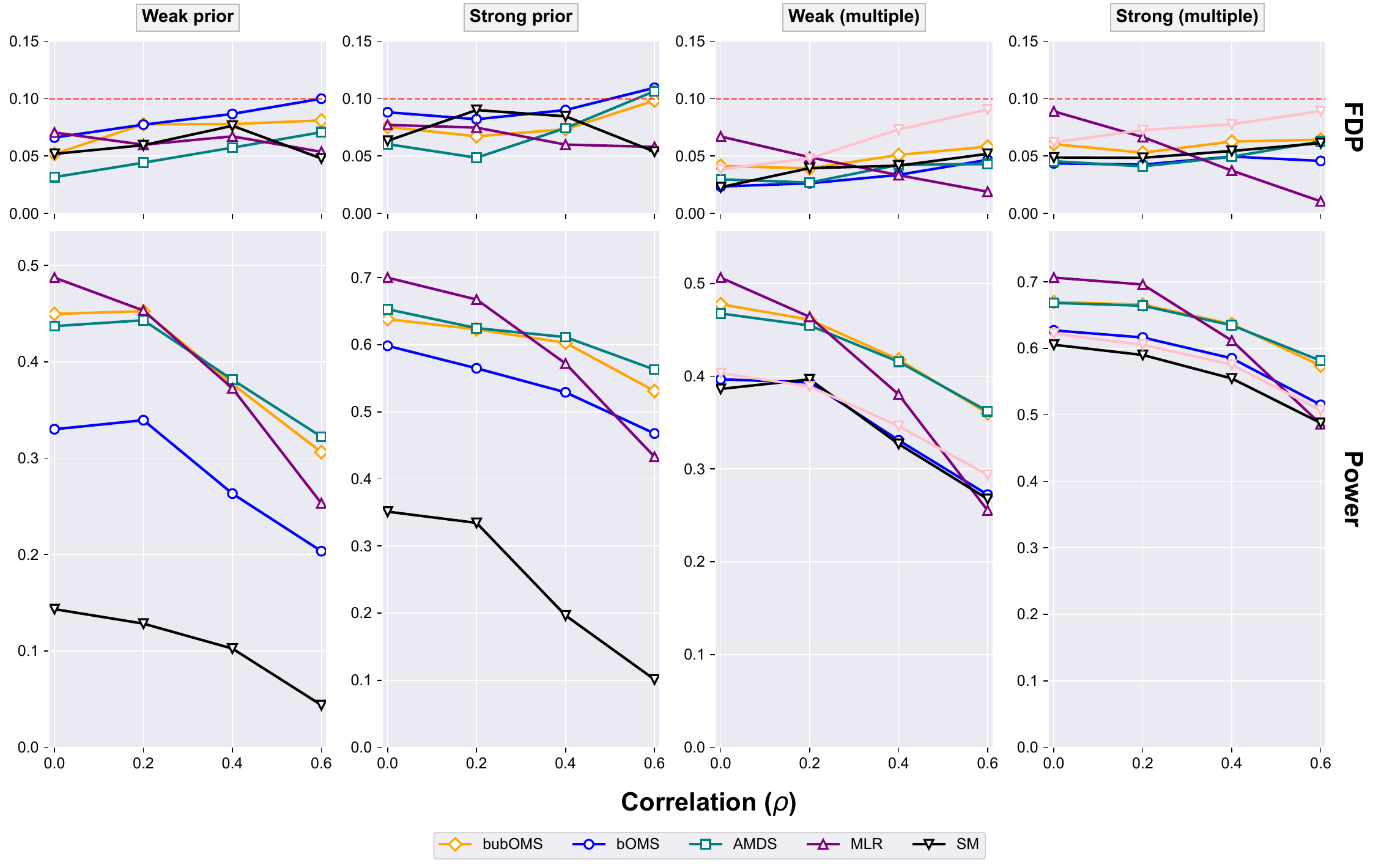}}
}
\caption{Logistic regression model with $n=p=400$ (Left1 and Right2) and $n=p=800$ (Left2 and Right1). The expected number of relevant features is $p_1=30$. Each row of the design matrix is generated from $N(0, \Sigma(\rho))$. The true coefficient $\beta_j^\star$ are drawn from $\mathcal{N}(0,1.5^2)$. In all panels, we fix the signal and vary the correlation $\rho$. The nominal FDR level $q=0.1$. For each multiple run, we replete the base algorithm 10 times. Each dot represents the average over 100 independent runs.}
\label{fig:logit}
\end{figure}

\section{Real data application: HIV drug resistance}

We apply our methods to the HIV drug resistance dataset from \cite{rhee2006hivdrugs}, which has been previously analyzed by \citep{barber2015,dai2023false,spector2022asymptotically}. This dataset contains resistance measurements of seven drugs for protease inhibitors (PIs), six drugs for nucleoside reverse transcriptase inhibitors (NRTIs), and three drugs for nonnucleoside reverse transcriptase inhibitors (NNRTIs). We focus on the first two as in \cite{dai2023false}.

The response vector $\boldsymbol{y}$ consists of log-transformed drug resistance measurements. The design matrix $\boldsymbol{X}$ is binary, where the $j$-th column indicates the presence or absence of the $j$-th mutation. Our task is to select relevant mutations for each drug within the inhibitor class. We preprocess the dataset using the same approach as \cite{barber2015}. 
We remove the patients with missing drug resistance information and exclude those mutations that appear less than three times across all patients. The sample size $n$ and the number of mutations $p$ vary from drug to drug, but are all on the scale of hundreds with $n/p$ ranging from 1.5 to 4. We assume a linear model between the response and features without considering the interaction terms. 

 We compare the numbers of true positives and false positives for different methods, including Benjamini-Hochberg Procedure (BHq), MLR Knockoffs (MLR)  \citep{spector2022asymptotically}, sign-sum mirror statistic \eqref{vanilla-mr-stats} (SM), Bayes-optimal mirror statistic with balanced splitting ratio (bOMS) and ADMS--Snell. 
 Throughout, we set the nominal FDR $q=0.1$. 
 As shown in Figure \ref{fig:hiv_PI}, ADMS outperforms SM and bOMS across all seven PI drugs. It identifies more true positives while maintaining control over false discoveries. Compared to the knockoff method, ADMS achieves comparable results, with a significant improvement observed for the drug Atazanavir (ATV). This indicates that ADMS is effective in enhancing the power of feature selection in the context of PIs.

\begin{figure}[t]
\centering
\scalebox{1}[0.7]{
\includegraphics[width=0.9\textwidth]{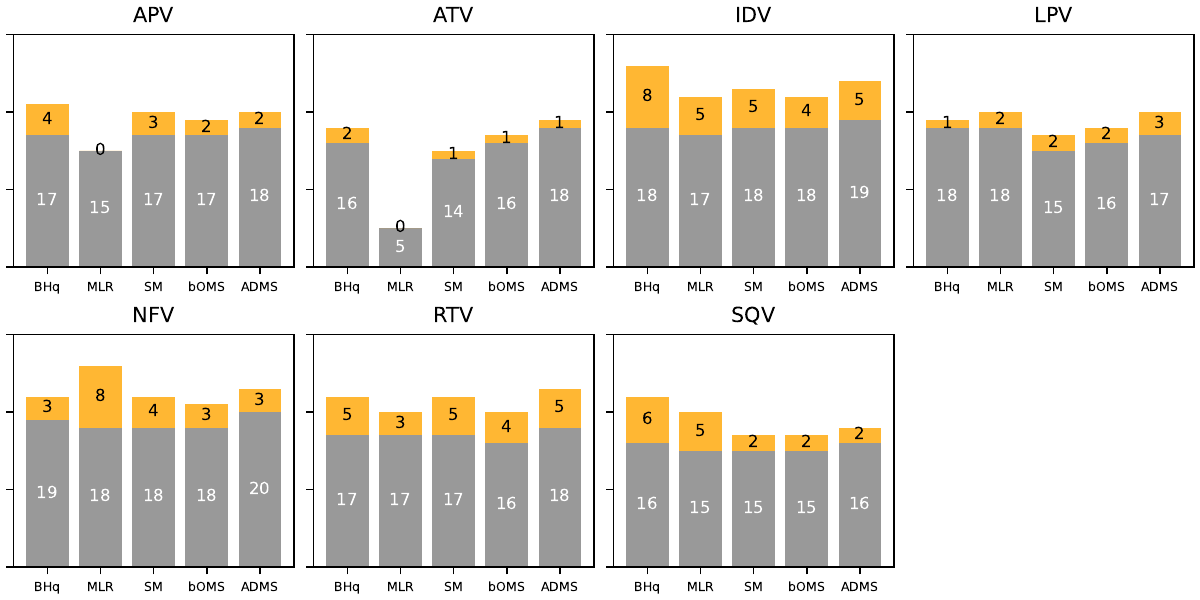}}
\caption{Numbers of the discovered mutations for the seven PI drugs. The grey and orange bars represent
the numbers of true and false positives. The nominal FDR level is $q=0.1$. }
\label{fig:hiv_PI}
\end{figure}

\begin{figure}[t]
\centering
\scalebox{1}[0.75]{\includegraphics[width=0.9\textwidth]{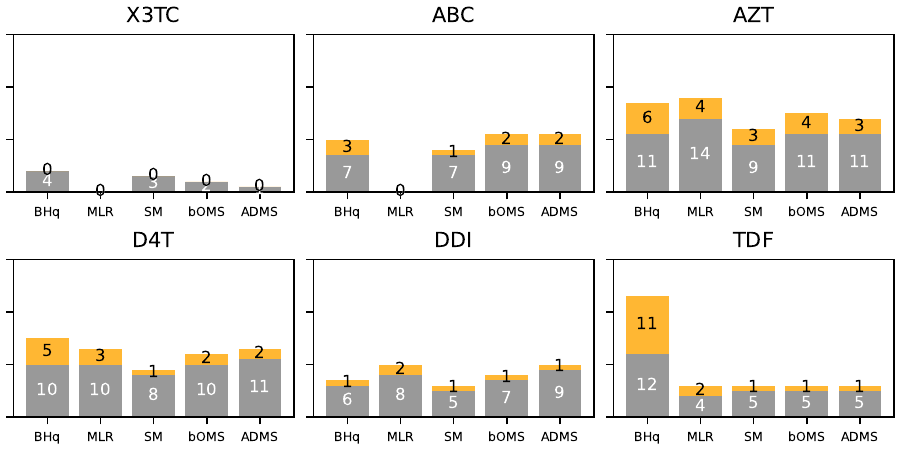}}
\caption{Numbers of the discovered mutations for the six NRTI drugs. The grey and orange bars represent
the numbers of true and false positives. The nominal FDR control level is $q=0.1$. }
\label{fig:hiv_NRTI}
\end{figure}

\section{Discussion}
\label{sec:discussion}
In this paper, we introduce a general class of mirror statistics \eqref{general-form-mr-stats} through a two-stage procedure inspired by \cite{li2021whiteout}. Under mild weak-dependence condition, Proposition~\ref{prop:fdr-control} shows that any mirror statistic in this class yields asymptotic FDR control. Within this class, we derive the Bayes-optimal mirror statistic \eqref{opt-mr-stats} for a given split ratio by separately optimizing (i) the guess in the exploration stage and (ii) the ranking score in the confirmation stage. The resulting statistic admits a self-contrasting interpretation \eqref{opt-mr-stat-self-contrast} and is closely related in spirit to the MLR knockoff statistic \citep{spector2022asymptotically}. A key distinction is that our framework allows the ranking data $\Djr$ to be feature-specific, which accommodates flexible, problem-dependent constructions of the summary statistic $s_j^{(2)}$ and its mirror $\widetilde s_j^{(2)}$, as illustrated in Examples~\ref{ex:negation}, \ref{ex:modelX}, and \ref{ex:modely}.

We also propose an adaptive prior-assisted data-splitting framework that treats the split ratio as a stopping time rather than a fixed predefined choice. By linking the choice of splitting ratio to an optimal stopping problem over a natural filtration of progressively revealed data, we characterize an optimal splitting rule via the Snell envelope. To the best of our knowledge, this is the first work in the FDR-controlling literature that formulates data-driven sample splitting through an explicit optimal-stopping perspective. Finally, to mitigate the computational burden of Snell-envelope evaluation, we adopt a Longstaff--Schwartz-type regression approximation. Both simulations and real data experiment show that our ADMS--Snell has superior performance with light computational cost. Still, it would be of interest to investigate alternative approximation strategies for scalable implementations for ADMS--Snell.

\bigskip

\spacingset{0.8}
\bibliographystyle{chicago} 
\bibliography{sample}

\newpage
\appendix
\numberwithin{equation}{section}

\section{Proofs}

\subsection{Proof for Theorem~\ref{thm:equi-twostage-generalclass}}

\begin{proof}
For simplicity, we assume there is no tie and let $\pi$ be a permutation such that
\[
\eta_{\pi(1)} \ge \eta_{\pi(2)} \ge \cdots \ge \eta_{\pi(p)}.
\]

Define $M_j=r_j\eta_j$ with $\eta_j\ge 0$. Then for any $t>0$,
\[
\{j: M_j>t\}=\{j: r_j=1,\ \eta_j>t\},\qquad
\{j: M_j<-t\}=\{j: r_j=-1,\ \eta_j>t\}.
\]
Hence the SeqStep FDP estimate can be written as
\begin{equation}\label{eq:fdp_t_rewrite}
\widehat{\mathrm{FDP}}(t)
=\frac{\sum_{j=1}^p \mathbf 1\{r_j=-1,\ \eta_j>t\}}{\sum_{j=1}^p \mathbf 1\{r_j=1,\ \eta_j>t\}\ \vee\ 1}.
\end{equation}

Next, let $\eta_{(1)}\ge\cdots\ge \eta_{(p)}$ denote the order statistics so that
$\eta_{(k)}=\eta_{\pi(k)}$, and define the top-$k$ prefix set
\[
A_k := \{\pi(1),\ldots,\pi(k)\}.
\]
Fix $k\in\{0,1,\ldots,p\}$ and choose any threshold $t$ such that
\[
\eta_{(k)} > t \ge \eta_{(k+1)},
\]
with the convention $\eta_{(0)}:=+\infty$ and $\eta_{(p+1)}:=-\infty$.
Then $\{j:\eta_j>t\}=A_k$. Plugging this into \eqref{eq:fdp_t_rewrite} yields
\[
\widehat{\mathrm{FDP}}(t)
=\frac{\sum_{j\in A_k}\mathbf 1\{r_j=-1\}}{\sum_{j\in A_k}\mathbf 1\{r_j=1\}\ \vee\ 1}
=: \widehat{\mathrm{FDP}}[k],
\]
which is exactly the two-stage estimated FDP after examining $k$ features.

Note that the map $t\mapsto \widehat{\mathrm{FDP}}(t)$ is a step function that is constant on each
interval $(\eta_{(k+1)},\eta_{(k)}]$ and equals $\widehat{\mathrm{FDP}}[k]$ there.
Therefore,
\[
\tau_q := \min\{t>0:\widehat{\mathrm{FDP}}(t)\le q\}
\]
is attained on the interval corresponding to the \emph{largest} index $\widehat{k}$ such that
$\widehat{\mathrm{FDP}}[\widehat{k}]\le q$, i.e., the two-stage stopping index.

Finally, the SeqStep selection set is
\[
\widehat S_{\mathrm{SeqStep}}
=\{j: M_j>\tau_q\}
=\{j: r_j=1,\ \eta_j>\tau_q\}
=\{j\in A_{\widehat{k}}: r_j=1\},
\]
which matches the two-stage discoveries (the correct guesses among the examined prefix).
This proves the claimed equivalence.
\end{proof}

\subsection{Proof for Proposition~\ref{prop:fdr-control}}

First, we introduce a lemma demonstrating the weak dependence for mirror statistic from Assumption \ref{ass:weak-cor-among-nulls} and the rest of the proof largely follows \cite{dai2023false}. 

\begin{lemma}
\label{lemma:weak-cor-among-nulls-for-MS}[Assumption 2 in \cite{dai2023false}]
    The mirror statistic $M_j$'s are continuous random variables and there exist constant $c$ and $\alpha\in (0,2)$ such that
    \begin{equation*}
        \Var \left(\sum_{j\in S_0} \mathbbm{1}(M_j>t)\right) \leq cp_0^\alpha, \quad \forall t\in \mathbb{R}, \quad \text{where } p_0=|S_0|.
    \end{equation*}
\end{lemma}

For the ease of presentation, we introduce the following notations. For any \( t \in \mathbb{R} \), denote
\[
\hat{G}_p^0(t) = \frac{1}{p_0} \sum_{j \in S_0} 1(M_j > t), \quad G_p^0(t) = \frac{1}{p_0} \sum_{j \in S_0} \mathbb{P}(M_j > t),
\]
\[
\hat{G}_p^1(t) = \frac{1}{p_1} \sum_{j \in S_1} 1(M_j > t), \quad \hat{V}_p(t) = \frac{1}{p_0} \sum_{j \in S_0} 1(M_j < -t).
\]
Let \( r_p = p_1 / p_0 \). In addition, denote
\[
\text{FDP}_p(t) = \frac{\hat{G}_p^0(t)}{\hat{G}_p^0(t) + r_p \hat{G}_p^1(t)},  \quad \text{FDP}_p(t) = \frac{G_p^0(t)}{G_p^0(t) + r_p \hat{G}_p^1(t)}.
\]
We also know that $\estfdp$ can be written as
\[
    \estfdp = \frac{\hat{V}_p^0(t)}{\hat{G}_p^0(t) + r_p \hat{G}_p^1(t)}.
\]
The rest of the proof is essentially to show that $\estfdp$ approximates $\fdp$ fairly well when $p\to\infty$. Consequently, controlling $\estfdp$ generalizes controlling $\fdp$. 

\begin{lemma}
\label{lemma:GV-converge-in-probability}
    Under Lemma \ref{lemma:weak-cor-among-nulls-for-MS}, and assuming \( p_0 \to \infty \) as \( p \to \infty \), we have, in probability,
    \[
    \sup_{t \in \mathbb{R}} \left| \hat{G}_p^0(t) - G_p^0(t) \right| \to 0, \quad \sup_{t \in \mathbb{R}} \left| \hat{V}_p(t) - G_p^0(t) \right| \to 0.
    \]
\end{lemma}

\textbf{Proof of Lemma \ref{lemma:weak-cor-among-nulls-for-MS}}.

Since $\mathbbm{1}(M_j>t)$ is a random bounded function of the summary statistics $s_j$ taking value between $-1$ and $1$. This lemma is a direct consequence of Assumption \ref{ass:weak-cor-among-nulls}.

\textbf{Proof of Lemma \ref{lemma:GV-converge-in-probability}}. For any \( \epsilon \in (0, 1) \), denote \( -\infty = \alpha_0^p < \alpha_1^p < \cdots < \alpha_{N_\epsilon}^p = \infty \) with \( N_\epsilon = \lceil 2 / \epsilon \rceil \), such that \( G_p^0(\alpha_{k-1}^p) - G_p^0(\alpha_k^p) \leq \epsilon / 2 \) for \( k = 1, \dots, N_\epsilon \). By Lemma \ref{lemma:weak-cor-among-nulls-for-MS}, such a sequence \( \{ \alpha_k^p \} \) exists since \( G_p^0(t) \) is a continuous function for \( t \in \mathbb{R} \). We have
\[
\mathbb{P} \left( \sup_{t \in \mathbb{R}} \left| \hat{G}_p^0(t) - G_p^0(t) \right| > \epsilon \right) \leq \mathbb{P} \left( \bigcup_{k=1}^{N_\epsilon} \sup_{t \in [\alpha_{k-1}^p, \alpha_k^p]} \left| \hat{G}_p^0(t) - G_p^0(t) \right| > \epsilon \right),
\]
\begin{equation}
\label{eq:sup-Ghat-G-control}
    \leq \sum_{k=1}^{N_\epsilon} \mathbb{P} \left( \sup_{t \in [\alpha_{k-1}^p, \alpha_k^p]} \left| \hat{G}_p^0(t) - G_p^0(t) \right| > \epsilon \right).
\end{equation}

We note that both \( \hat{G}_p^0(t) \) and \( G_p^0(t) \) are monotonically decreasing. Therefore, \( \forall k \in \{1, \dots, N_\epsilon\} \), we have
\[
\sup_{t \in [\alpha_{k-1}^p, \alpha_k^p]} \left| \hat{G}_p^0(t) - G_p^0(t) \right| \leq \left| \hat{G}_p^0(\alpha_{k-1}^p) - G_p^0(\alpha_{k-1}^p) \right| + \epsilon / 2.
\]

By Equation \eqref{eq:sup-Ghat-G-control}, Lemma \ref{lemma:weak-cor-among-nulls-for-MS}, and the Chebyshev's inequality, it follows that
\[
\mathbb{P} \left( \sup_{t \in \mathbb{R}} \left| \hat{G}_p^0(t) - G_p^0(t) \right| > \epsilon \right) \leq \sum_{k=1}^{N_\epsilon} \mathbb{P} \left( \left| \hat{G}_p^0(\alpha_{k-1}^p) - G_p^0(\alpha_{k-1}^p) \right| > \frac{\epsilon}{2} \right) \leq \frac{4 c N_\epsilon}{p_0^2 \epsilon^2} \to 0, \quad \text{as } p \to \infty.
\]

Similarly, we can show that
\[
\mathbb{P} \left( \inf_{t \in \mathbb{R}} \left( \hat{G}_p^0(t) - G_p^0(t) \right) < -\epsilon \right) \leq \sum_{k=1}^{N_\epsilon} \mathbb{P} \left( \hat{G}_p^0(\alpha_k^p) - G_p^0(\alpha_k^p) < -\frac{\epsilon}{2} \right) \leq \frac{4cN_\epsilon}{p_0^2 \epsilon^2} \to 0, \quad \text{as } p \to \infty.
\]
This concludes the proof of the first claim in Lemma \ref{lemma:GV-converge-in-probability}. The second claim follows similarly using the symmetric property of the mirror statistic \( M_j \) for \( j \in S_0 \).

\textbf{Proof of Proposition \ref{prop:fdr-control}}. We first show that for any \( \epsilon \in (0, q) \), we have
\[
\mathbb{P}(\tau_q \leq t_{q - \epsilon}) \geq 1 - \epsilon,
\]
in which \( t_{q - \epsilon} > 0 \) satisfying \( \mathbb{P}(\text{FDP}_p(t_{q - \epsilon}) \leq q - \epsilon) \to 1 \). Since the variances of the mirror statistics
are upper bounded and also bounded away from 0, by Lemma \ref{lemma:GV-converge-in-probability}, we have
\[
\sup_{0 < t \leq c} \left| \text{FDP}_p^+(t) - \text{FDP}_p(t) \right| \to 0
\]
for any constant \( c > 0 \). By the definition of \( \tau_q \), i.e., \( \tau_q = \inf \{ t > 0 : \text{FDP}_p(t) \leq q \} \), we have
\[
\mathbb{P}(\tau_q \leq t_{q - \epsilon}) \geq \mathbb{P}(\text{FDP}_p^+(t_{q - \epsilon}) \leq q) \geq \mathbb{P}(|\text{FDP}_p^+(t_{q - \epsilon}) - \text{FDP}_p(t_{q - \epsilon})| \leq \epsilon, \, \text{FDP}_p(t_{q - \epsilon}) \leq q - \epsilon) \geq 1 - \epsilon
\]
for \( p \) large enough. Conditioning on the event \( \tau_q \leq t_{q - \epsilon} \), we have
\[
\limsup_{p \to \infty} \mathbb{E} \left[ \text{FDP}_p(\tau_q) \right] \leq \limsup_{p \to \infty} \mathbb{E} \left[ \text{FDP}_p(\tau_q) \mid \tau_q \leq t_{q - \epsilon} \right] \mathbb{P}(\tau_q \leq t_{q - \epsilon}) + \epsilon
\]
\[
\leq \limsup_{p \to \infty} \mathbb{E} \left[ \left| \text{FDP}_p(\tau_q) - \text{FDP}_p(t_{q - \epsilon}) \right| \mid \tau_q \leq t_{q - \epsilon} \right] \mathbb{P}(\tau_q \leq t_{q - \epsilon})
\]
\[
+ \limsup_{p \to \infty} \mathbb{E} \left[ \left| \text{FDP}_p^+(t_{q - \epsilon}) - \text{FDP}_p(t_{q - \epsilon}) \right| \mid \tau_q \leq t_{q - \epsilon} \right] \mathbb{P}(\tau_q \leq t_{q - \epsilon}) + \epsilon
\]
\[
\leq \limsup_{p \to \infty} \sup_{0 < t \leq t_{q - \epsilon}} \left| \text{FDP}_p^+(t) - \text{FDP}_p(t) \right| 
\]
\[
+ \limsup_{p \to \infty} \sup_{0 < t \leq t_{q - \epsilon}} \left| \text{FDP}_p^+(t) - \text{FDP}_p(t) \right|
\]
\[
+ \limsup_{p \to \infty} \text{FDP}_p(t_{q - \epsilon}) + \epsilon.
\]

The first two terms are 0 based on Lemma \ref{lemma:GV-converge-in-probability} and the dominated convergence theorem. For the third term, we have \( \text{FDP}_p(\tau_q) \leq q \) by the definition of \( \tau_q \). This concludes the proof of Proposition \ref{prop:fdr-control}. \qed

\subsection{Proof for Proposition~\ref{prop:optimality}}
\label{sec:proof-concentration}

\begin{theorem}[\citet{doukhan2007probability}]
\label{thm:DN2007}
Suppose that $X_1,\ldots,X_n$ are mean-zero random variables taking values in $[-1,1]$ such that
\[
\mathrm{Var}(\bar X_n)\le C_0 n
\]
for a constant $C_0>0$. Let $L_1,L_2<\infty$ be constants such that for any $i\le j$,
\[
\bigl|\mathrm{Cov}(X_i,X_j)\bigr|\le 4\,\varphi(j-i),
\]
where $\{\varphi(k)\}_{k\in\mathbb{N}}$ is a nonincreasing sequence satisfying
\[
\sum_{s=0}^{\infty}(s+1)^k\,\varphi(s)\le L_1 L_2^{k}\,k!\qquad \text{for all } k\ge 0.
\]
Then for all $t\in(0,1)$, there exists a universal constant $C_1>0$ only depending on $C_0,L_1,L_2$ such that
\[
\mathbb{P}(\bar X_n\ge t)
\le \exp\!\left(-\frac{t^2}{C_0 n + C_1 t^{7/4} n^{7/4}}\right)
\le \exp\!\left(-C' t^2 n^{1/4}\right),
\]
where $C'$ is a universal constant only depending on $C_0,L_1,L_2$.
\end{theorem}

If we take $\varphi(s)=c\rho^{s}$, this yields the following corollary.

\begin{corollary}[\citet{spector2022asymptotically}]
\label{cor:geom-cov}
Suppose that $Z_1,\ldots,Z_p$ are mean-zero random variables taking values in $[-1,1]$. Suppose
that for some $C\ge 0$ and $\rho\in(0,1)$, the sequence satisfies
\begin{equation}
\label{eq:C1}
\bigl|\mathrm{Cov}(Z_i,Z_j)\bigr|\le C\rho^{|i-j|}.
\tag{C.1}
\end{equation}
Then there exists a universal constant $C'$ depending only on $C$ and $\rho$ such that
\begin{equation}
\label{eq:C2}
\mathbb{P}(\bar Z_p \ge t)\le \exp\!\bigl(-C' t^2 p^{1/4}\bigr).
\tag{C.2}
\end{equation}
Furthermore, let $\pi:[p]\to[p]$ be any permutation. For $m\le p$, define
\[
\bar Z^{(\pi)}_{m}:=\frac{1}{m}\sum_{i=1}^{m} Z_{\pi(i)}
\]
to be the sample mean of the first $m$ random variables after permuting $(Z_1,\ldots,Z_p)$ according to $\pi$.
Then for any $k\in\mathbb{N}$ and $t\ge 0$,
\begin{equation}
\label{eq:C3}
\sup_{\pi\in S_p}\;
\mathbb{P}\!\left(\max_{k\le m\le p}\bigl|\bar Z^{(\pi)}_{m}\bigr|\ge t\right)
\le p\,\exp\!\bigl(-C' t^2 k^{1/4}\bigr),
\tag{C.3}
\end{equation}
where $S_p$ is the symmetric group.
\end{corollary}

\begin{proof}
The proof of~\eqref{eq:C2} follows an observation of Doukhan and Neumann (2007), where we note
$\varphi(s)=C\exp(-as)$ for $a=-\log(\rho)$. Then
\[
\sum_{s=0}^{\infty}(s+1)^k\exp(-as)
\le \sum_{s=0}^{\infty}\prod_{i=1}^{k}(s+i)\exp(-as)
= \left.\frac{d^k}{dp^k}\!\left(\frac{1}{1-p}\right)\right|_{p=\exp(-a)}
= k!\,\frac{1}{\bigl(1-\exp(-a)\bigr)^k}.
\]
As a result,
\[
\sum_{s=0}^{\infty}(s+1)^k\varphi(s)
\le C\left(\frac{1}{1-\exp(-a)}\right)^k k!,
\]
so we take $L_1=\bigl(1-\exp(-a)\bigr)^{-1}$ and $L_2=C$. Lastly, we observe that another geometric
series argument yields
\[
\mathrm{Var}(\bar Z_p)
=\sum_{i=1}^{p}\sum_{j=1}^{p}\mathrm{Cov}(Z_i,Z_j)
\le \sum_{i=1}^{p}\sum_{j=1}^{p} C\rho^{|i-j|}
\le pC\,\frac{2}{1-\rho}.
\]
Thus, we take $C_0=\frac{2C}{1-\rho}$ and apply Theorem~C.1, which yields the first result.
To prove~\eqref{eq:C3}, the main idea is that we can apply~\eqref{eq:C2} to each sample mean
$\bigl|\bar Z^{(\pi)}_{m}\bigr|$, at which point~\eqref{eq:C3} follows from a union bound.

To prove this, note that if we rearrange $(Z_{\pi(1)},\ldots,Z_{\pi(m)})$ into their ``original order,''
then these variables satisfy the condition in~\eqref{eq:C1}. Formally, let
$A=\{\pi(1),\ldots,\pi(m)\}$ and let $\nu:A\to A$ be the permutation such that
$\nu(\pi(i))>\nu(\pi(j))$ if and only if $i>j$, for $i,j\in[m]$. Then define
$Y_i=Z_{\nu(\pi(i))}$ for $i\in[m]$, and note that
\[
\bigl|\mathrm{Cov}(Y_i,Y_j)\bigr|
= \bigl|\mathrm{Cov}(Z_{\nu(\pi(i))},Z_{\nu(\pi(j))})\bigr|
\le C\rho^{|\nu(\pi(i))-\nu(\pi(j))|}
\le C\rho^{|i-j|},
\]
where in the last step, $|i-j|\le |\nu(\pi(i))-\nu(\pi(j))|$ follows by construction of $\nu$.
Since $\bar Y_m=\bar Z^{(\pi)}_{m}$ by construction, this means we may apply~\eqref{eq:C2} to
$\bar Z^{(\pi)}_{m}$ for each $m$.

Thus, by~\eqref{eq:C2}, for any $\pi\in S_p$,
\[
\mathbb{P}\!\left(\max_{k\le m\le p}\bigl|\bar Z^{(\pi)}_{m}\bigr|\ge t\right)
\le \sum_{m=k}^{p}\mathbb{P}\!\left(\bigl|\bar Z^{(\pi)}_{m}\bigr|\ge t\right)
\le \sum_{m=k}^{p}\exp\!\bigl(-C' t^2 m^{1/4}\bigr)
\le p\,\exp\!\bigl(-C' t^2 k^{1/4}\bigr).
\]
This completes the proof.
\end{proof}

\begin{proof}
Fix $q\in(0,1)$ and $\epsilon>0$.  Throughout, we condition on the first half of data $D^{(1)}$ and view all random quantities from the second half as conditionally independent draws given $D^{(1)}$.  For each feature $j$, let $(s_j^{(2)},\tilde s_j^{(2)})$ denote its mirrored pair from the second half, and write the corresponding label variable as
\[
r_j = \psi_j(s_j^{(2)},\tilde s_j^{(2)},D^{(1)}) \in \{\pm1\},
\]
where $\psi_j$ is a measurable function that outputs the mirror sign.  By construction, $r_j$ depends only on $(s_j^{(2)},\tilde s_j^{(2)})$ and the conditioning information $D^{(1)}$.

Let $\sigma=\sigma(D^{(1)})$ be the permutation in Assumption~\ref{ass:exp-corr}, so that features are ordered according to the ranking scores $\eta_j$ and Algorithm~\ref{alg:selecseq} examines $\sigma(1),\ldots,\sigma(p)$.  For each $k$, define the empirical quantities
\[
\widehat X_k = \frac{1}{k}\sum_{i=1}^k \mathbbm{1}\{r_{\sigma(i)}=-1\},\qquad
\widehat Y_k = \frac{1}{k}\sum_{i=1}^k \mathbbm{1}\{r_{\sigma(i)}=1\},\qquad
\widehat{\FDP}[k] = \frac{\widehat X_k}{\widehat Y_k\vee k^{-1}},
\]
and their conditional expectations
\[
\widetilde X_k = \frac{1}{k}\sum_{i=1}^k \mathbbm{P}(r_{\sigma(i)}=-1\mid D^{(1)}),\qquad
\widetilde Y_k = \frac{1}{k}\sum_{i=1}^k \mathbbm{P}(r_{\sigma(i)}=1\mid D^{(1)}),\qquad
\widetilde{\FDP}[k] = \frac{\widetilde X_k}{\widetilde Y_k}.
\]
Because each optimal guess chooses the label with higher posterior probability under its own conditional distribution, we have $\mathbbm{P}(r_j=1\mid D^{(1)})\ge1/2$ and hence $\widetilde Y_k\ge1/2$ for all $k$.

Since $r_j$ is a bounded measurable function of $(s_j^{(2)},\tilde s_j^{(2)},D^{(1)})$, the data-processing property of maximal correlation \cite{Witsenhausen1975} implies
\[
\rho_{\max}\!\left(r_{\sigma(i)},r_{\sigma(j)}\mid D^{(1)}\right)
\le \rho_{\max}\!\left(\,\left(s_{\sigma(i)}^{(2)},\tilde s_{\sigma(i)}^{(2)}\right),\,\left(s_{\sigma(j)}^{(2)},\tilde s_{\sigma(j)}^{(2)}\right)\mid D^{(1)}\right) = \Gamma_{ij},
\]

Define $Z_i = \mathbbm{1}\{r_{\sigma(i)}=\pm 1\}-\mathbbm{P}\{r_{\sigma(i)}=\pm 1\}\in [-1,1]$, we also have $\rho_{\max}(Z_i,Z_j \mid D^{(1)})\leq \Gamma_{ij}$.
Applying Corollary \ref{cor:geom-cov} and Assumption \ref{ass:exp-corr}, we have 
\[
\mathbbm{P}\!\left(|\widehat X_k-\widetilde X_k|\ge t\mid D^{(1)}\right)\le 2\exp\left(-C't^2k^{1/4}\right),
\qquad
\mathbbm{P}\!\left(|\widehat Y_k-\widetilde Y_k|\ge t\mid D^{(1)}\right)\le 2\exp\left(-C't^2k^{1/4}\right),
\]
for some constant $c=c(\alpha,C)>0$.  A union bound over $k\in[k_0,p]$ yields
\[
\mathbbm{P}\!\left(\max_{k_0\le k\le p}\bigl|\widehat X_k-\widetilde X_k\bigr|\ge t\mid D^{(1)}\right)
\le 2p\,\exp\left(-C't^2k_0^{1/4}\right),
\]
and similarly for $\widehat Y_k$.  Taking $k_0\asymp p$ ensures that these probabilities vanish, and thus
\[
\max_{k_0\le k\le p}|\widehat X_k-\widetilde X_k|=o_{\mathbbm{P}}(1),
\qquad
\max_{k_0\le k\le p}|\widehat Y_k-\widetilde Y_k|=o_{\mathbbm{P}}(1),
\qquad
\max_{k_0\le k\le p}|\widehat{\FDP}[k]-\widetilde{\FDP}[k]|=o_{\mathbbm{P}}(1),
\]
where the last relation follows from $\widetilde Y_k\ge1/2$.

Let $(\psi_j^\star,\eta_j^\star)$ denote the optimal guess and its ranking rule, and define $r_j^\star=\psi_j^\star(s_j^{(2)},\tilde s_j^{(2)},D^{(1)})$.  By Bayes optimality,
\[
\mathbbm{P}(r_j^\star=-1\mid D^{(1)})\le\mathbbm{P}(r_j=-1\mid D^{(1)}),\qquad
\mathbbm{P}(r_j^\star=1\mid D^{(1)})\ge\mathbbm{P}(r_j=1\mid D^{(1)}),
\]
for every competing $(\psi_j,\eta_j)$.  Therefore, for any examined prefix $A_k=\{\sigma(1),\ldots,\sigma(k)\}$ determined by a competitor’s order,
\[
\widetilde X_k^\star=\frac{1}{k}\sum_{j\in A_k}\mathbbm{P}(r_j^\star=-1\mid D^{(1)})\le \widetilde X_k,\qquad
\widetilde Y_k^\star=\frac{1}{k}\sum_{j\in A_k}\mathbbm{P}(r_j^\star=1\mid D^{(1)})\ge \widetilde Y_k,
\]
which implies $\widetilde{\FDP}^\star[k]\le\widetilde{\FDP}[k]$ for all $k$ on the same prefix.

Fix $\zeta\in(0,1)$ and suppose $\beta(M,q)>\zeta$.  Then the competing procedure at level $q$ achieves at least $\zeta p_1$ true discoveries, so its stopping index $\hat k$ satisfies $\hat k\ge\zeta p_1\ge\zeta\gamma p$.  On the event where the uniform concentration bounds above hold for all $k\ge k_0=\zeta\gamma p$, the stopping rule guarantees $\widehat{\FDP}[\hat k]\le q$, hence $\widetilde{\FDP}[\hat k]\le q+o_{\mathbbm{P}}(1)$.  By local optimality, $\widetilde{\FDP}^\star[\hat k]\le\widetilde{\FDP}[\hat k]\le q+o_{\mathbbm{P}}(1)$, and the same concentration transfer yields $\widehat{\FDP}^\star[\hat k]\le q+\epsilon$ for all large $p$.  Therefore the optimal procedure at nominal level $q+\epsilon$ will not stop before $\hat k$, that is, $\hat k^\star(q+\epsilon)\ge\hat k$ with probability approaching one.

Because $\hat k^\star(q+\epsilon)\ge\hat k$ and, for every feature, the optimal guess has conditionally no smaller correctness probability than any competitor, the number of true discoveries among the first $\hat k$ inspected features under the optimal rule is asymptotically no smaller.  Dividing by $p_1$ yields
\[
\beta(M^\star,q+\epsilon)\ge \beta(M,q)-o_{\mathbbm{P}}(1),
\]
and taking $\epsilon\downarrow0$  completes the argument:
\[
\underset{0<\epsilon<1-q}{\sup}\liminf_{p\to\infty}\{\beta(M^\star,q+\epsilon)-\beta(M,q)\}\ge0.
\]
\end{proof}

\subsection{Proof for Proposition \ref{prop:M-tildeM-comparison}}
\begin{lemma}[\cite{ke2024power}]
\label{lemma:geometric-insight}
    Fix an integer $d\geq 1$, a vector $\mu \in \mathbb{R}^d$ a covariance matrix $\Sigma\in \mathbb{R}^{d\times d}$ and an open set $S\subset \mathbb{R}^d$ such that $\mu \notin S$. Consider a sequence of random vectors $X_p \in \mathbb{R}^d$, indexed by $p$ satisfying that 
    \begin{equation*}
        X_p \mid (\mu_p, \Sigma_p) \sim \mathcal{N}\left(\mu_p, \frac{1}{2\log p}\Sigma_p\right)
    \end{equation*}
    where $\mu_p\in \mathbb{R}^d$ is a random vector and $\Sigma_p \in \mathbb{R}^{d\times d}$ a random covariance matrix. As $p\to \infty$, suppose for any fix $\gamma>0$ and $L>0$, $\mathbbm{P}(||\mu_p-\mu||>\gamma)\leq p^{-L}$, $\mathbbm{P}(||\Sigma_p-\Sigma||>\gamma)\leq p^{-L}$. Then as $p\to\infty$, 
    \begin{equation*}
        \mathbbm{P}(X_p\in S) = L_p p^{-b}
    \end{equation*}
    where $b = d(\mu, S)^2 = \underset{x\in S}{\inf} ||x-\mu||^2$ is the squre distance from $\mu$ to $S$.
\end{lemma}

\begin{proof}
    See Lemma 7.1 in \cite{ke2024power}.
\end{proof}

Note that 
\begin{equation*}
\begin{aligned}
    \widetilde M_j &= \log \left(\frac{\int e^{-(\widehat\beta_j^{(1)}-\beta_j)^2/4\omega_j-(\widehat\beta_j^{(2)}-\beta_j)^2/4\omega_j}((1-\epsilon_p)\nu_0+\epsilon_p\nu_{\tau_p})d\beta_j}{\int e^{-(\widehat\beta_j^{(1)}-\beta_j)^2/4\omega_j-(\widehat\beta_j^{(2)}+\beta_j)^2/4\omega_j}((1-\epsilon_p)\nu_0+\epsilon_p\nu_{\tau_p})d\beta_j}\right)\\
    &=\log\left(\frac{1-\epsilon_p+\epsilon_pe^{-(2\omega_j)^{-1}(\tau_p^2-\tau_p(\widehat\beta_j^{(1)}+\widehat\beta_j^{(2)}))}}{1-\epsilon_p+\epsilon_pe^{-(2\omega_j)^{-1}(\tau_p^2-\tau_p(\widehat\beta_j^{(1)}-\widehat\beta_j^{(2)}))}}\right).
\end{aligned}
\end{equation*}
Let $U_j = \widehat\beta_j^{(1)}+\widehat\beta_j^{(2)}$ and $V_j = \widehat\beta_j^{(1)}-\widehat\beta_j^{(2)}$, then we know $U_j$ and $V_j$ are independent normal random variables with variance $4\omega_j$. Then $\widetilde M_j>2u\log p$ can be separated according to the value of $V_j$ into two cases.
\begin{itemize}
    \item (Case 1) $V_j\leq \frac{2\omega_j\vartheta+2r}{\sqrt{2r}}\sqrt{\log p}$. In this case $1-\epsilon_p\leq 1-\epsilon_p+\epsilon_pe^{-(2\omega_j)^{-1}(\tau_p^2-\tau_p(\widehat\beta_j^{(1)}-\widehat\beta_j^{(2)}))}\leq 2$. Thus
    \begin{equation*}
        \begin{aligned}
            \mathbbm{P}(\widetilde M_j>2u\log p) &= L_p\mathbbm{P}(-\vartheta\log p-(2\omega_j)^{-1}(2r\log p-\sqrt{2r\log p}U_j)>2u\log p)\\
            & = L_p\mathbbm{P}\left(U_j>\frac{2\omega_j(2u+\vartheta)+2r}{\sqrt{2r}}\sqrt{\log p}\right)
        \end{aligned}
    \end{equation*}
    \item (Case 2) $V_j>\frac{2\omega_j\vartheta+2r}{\sqrt{2r}}\sqrt{\log p}$. In this case $1-\epsilon_p+\epsilon_pe^{-(2\omega_j)^{-1}(\tau_p^2-\tau_p(\widehat\beta_j^{(1)}-\widehat\beta_j^{(2)}))}\leq  2\epsilon_pe^{-(2\omega_j)^{-1}(\tau_p^2-\tau_p(\widehat\beta_j^{(1)}-\widehat\beta_j^{(2)}))}$. Then we have:
    \begin{equation*}
        \begin{aligned}
            \mathbbm{P}(\widetilde M_j>2u\log p) = L_p\mathbbm{P}((2\omega_j)^{-1}\tau_p(U_j-V_j)>2u\log p) = L_p\mathbbm{P}\left(U_j-V_j>\frac{4\omega_ju}{\sqrt{2r}}\sqrt{\log p}\right).
        \end{aligned}
    \end{equation*}
\end{itemize}
We will explain why we get the extra $L_p$ terms in the above probability. Denote
\begin{align*}
    \epsilon &= \log 1-\epsilon_p+\epsilon_pe^{-(2\omega_j)^{-1}(\tau_p^2-\tau_p(\widehat\beta_j^{(1)}-\widehat\beta_j^{(2)}))} \\
    \eta &= \log \frac{1-\epsilon_p+\epsilon_pe^{-(2\omega_j)^{-1}(\tau_p^2-\tau_p(\widehat\beta_j^{(1)}-\widehat\beta_j^{(2)}))}}{\epsilon_pe^{-(2\omega_j)^{-1}(\tau_p^2-\tau_p(\widehat\beta_j^{(1)}-\widehat\beta_j^{(2)}))}}
\end{align*}
Then we have $1\leq \epsilon\leq 2$ in Case 1 and $1\leq \eta \leq 2$ in Case 2. Now we can write $\{\widetilde{M}_j > 2u \log p \}$ as $R_1(u)\cup R_2(u)$ and , $\{\widetilde{M}_j < 2u \log p \}$ as $Q_1(u)\cup Q_2(u)$
\begin{align*}
    R_1(u) &= \{U_j>\frac{2\omega_j(2u+\vartheta+\epsilon/\log p)+2r}{\sqrt{2r}}\sqrt{\log p}, V_j\leq\frac{2\omega_j\vartheta+2r}{\sqrt{2r}}\sqrt{\log p}\} \\
    R_2(u) &= \{U_j-V_j>\frac{2\omega_j(2u+\eta/\log p)}{\sqrt{2r}}\sqrt{\log p}, V_j>\frac{2\omega_j\vartheta+2r}{\sqrt{2r}}\sqrt{\log p}\} \\
    Q_1(u) &= \{U_j<\frac{2\omega_j(2u+\vartheta+\epsilon/\log p)+2r}{\sqrt{2r}}\sqrt{\log p}, V_j\leq\frac{2\omega_j\vartheta+2r}{\sqrt{2r}}\sqrt{\log p}\} \\
    Q_2(u) &= \{U_j-V_j<\frac{2\omega_j(2u+\eta/\log p)}{\sqrt{2r}}\sqrt{\log p}, V_j>\frac{2\omega_j\vartheta+2r}{\sqrt{2r}}\sqrt{\log p}\}
\end{align*}
Define the similar region $R_1^0(u), R_2^0(u), Q_1^0(u), Q_2^0(u)$ without $\epsilon, \eta$ to be:
\begin{align*}
    R_1^0(u) &= \{U_j>\frac{2\omega_j(2u+\vartheta)+2r}{\sqrt{2r}}\sqrt{\log p}, V_j\leq\frac{2\omega_j\vartheta+2r}{\sqrt{2r}}\sqrt{\log p}\} \\
    R_2^0(u) &= \{U_j-V_j>\frac{4\omega_ju}{\sqrt{2r}}\sqrt{\log p}, V_j>\frac{2\omega_j\vartheta+2r}{\sqrt{2r}}\sqrt{\log p}\} \\
    Q_1^0(u) &= \{U_j<\frac{2\omega_j(2u+\vartheta)+2r}{\sqrt{2r}}\sqrt{\log p}, V_j\leq\frac{2\omega_j\vartheta+2r}{\sqrt{2r}}\sqrt{\log p}\} \\
    Q_2^0(u) &= \{U_j-V_j<\frac{4\omega_ju}{\sqrt{2r}}\sqrt{\log p}, V_j>\frac{2\omega_j\vartheta+2r}{\sqrt{2r}}\sqrt{\log p}\}
\end{align*}
It is easy to show that $R_1^0(u+\frac{\log 2}{2\log p}) \subset R_1(u)\subset R_1^0(u)$ and $R_2^0(u+\frac{\log 2}{2\log p}) \subset R_2(u)\subset R_2^0(u)$.
We now analyze the probabilities
\begin{align*}
    &\mathbbm{P}(R_1^0(u)\Delta R_1(u) \mid \beta=0), \quad \mathbbm{P}(R_2^0(u)\Delta R_2(u) \mid \beta=0) \\
    &\mathbbm{P}(R_1^0(u)\Delta R_1(u) \mid \beta\neq 0), \quad \mathbbm{P}(R_2^0(u)\Delta R_2(u) \mid \beta\neq 0)
\end{align*}
We can tell that $R_1^0(u)\Delta R_1(u)$ and $R_2^0(u)\Delta R_2(u)$ are both contained in some strip with infinity length and width at most $\frac{2\omega_j \log 2}{\sqrt{2r\log p}}:=\delta \sqrt{8\omega_j\log p}$. Suppose the distance of the strip to the center of $(U_j,V_j)$ is $d\sqrt{8\omega_j\log p}$ where $d$ is a fix constant.  For any region like that $\mathcal{R}$, we have
\begin{align*}
    \mathbbm{P}((U_j,V_j)\in \mathcal{R}) \lesssim 
    \int_{d}^{d-\delta} \exp\left(-t^2\right) dt 
    \lesssim 
    \frac{1}{\log p} p^{-d^2} 
    = L_p p^{-d^2}
\end{align*}
By Lemma \ref{lemma:geometric-insight}, 
\begin{equation*}
    \mathbbm{P}((U_j,V_j)\in R_1^0(u)) = \widetilde{L_p} p^{-d(u)^2}
\end{equation*}
Thus we have explained the $L_p$ term. 
Combine the above two cases, we have:
\begin{equation*}
    \begin{aligned}
        \mathbbm{P}(\widetilde M_j>2u\log p) & = L_p\mathbbm{P}\left(U_j>\frac{2\omega_j(2u+v)+2r}{\sqrt{2r}}\sqrt{\log p}, V_j\leq\frac{2\omega_j\vartheta+2r}{\sqrt{2r}}\sqrt{\log p}\right)\\
        & +L_p\mathbbm{P}\left(U_j-V_j>\frac{4\omega_ju}{\sqrt{2r}}\sqrt{\log p}, V_j>\frac{2\omega_j\vartheta+2r}{\sqrt{2r}}\sqrt{\log p}\right).
    \end{aligned}  
\end{equation*}
Thus under the null, $U_j, V_j\sim N(0, 4\omega_j)$, we have:
\begin{equation*}
    \mathbbm{P}(\widetilde M_j>2u\log p\mid \beta_j = 0) = L_pp^{-\frac{(2\omega_j(2u+\vartheta)+2r)^2}{16\omega_jr}}.
\end{equation*}
Under the alternative: $U_j\sim N(2\tau_p, 4\omega_j)$, $V_j\sim N(0, 4\omega_j)$. Then 
\begin{equation*}
    \mathbbm{P}(\widetilde M_j<2u\log p\mid \beta_j\neq 0) = L_pp^{-\frac{(2r-2\omega_j(2u+\vartheta))_+^2}{16\omega_jr}}.
\end{equation*}
Now, it is sufficient to show that when $r\geq \omega \vartheta$,
\begin{equation*}
    \inf_u ( u, 
    \ (\sqrt{r/\omega}-\sqrt{u})_+^2+\vartheta, 
    \ \frac{1}{2}r/\omega+\vartheta) = 
    \min (\frac{(r/\omega+\vartheta)^2}{2 r/\omega}, 
    \ \frac{1}{2}r/\omega+ \vartheta)
\end{equation*}
From that we have 
\begin{equation*}
    f_{\text{Hamm},\widetilde M} = 1-\vartheta-\frac{(r-\vartheta)_+^2}{4r}
\end{equation*}
From the proof in \cite{lin2026correlationcaveat}, we know that
\begin{equation*}
    f_{\text{Hamm},M} = 1-\vartheta-\frac{(r-\vartheta)_+^2}{4r}
\end{equation*}

Next we will proceed to prove \eqref{eq:prop-cmp-tilde-ms-with-vanilla-ms-hamming}. Suppose $r>\omega_j\vartheta$ and for any $0<u<\frac{r-\omega_j\vartheta}{8\omega_j}$, we will develop the bound 
\begin{equation*}
    \mathbbm{P}\left(Q_1^0(u+\frac{\log 2}{2 \log p })-Q_1^0(u) \mid \beta_j\neq 0\right) \leq \frac{C}{\sqrt{\log p}} \mathbbm{P}\left(Q_1^0(u) \mid \beta_j\neq 0\right)
\end{equation*}
where $C$ is a fixed constant not depend on $p$. 
\begin{align*}
    &\mathbbm{P}\left(Q_1^0(u+\delta)-Q_1^0(u) \mid \beta_j\neq 0\right) = \frac{1}{\sqrt{8\pi\omega_j}} \int_{t=A}^B \exp(-(t-2\sqrt{2r\log p})^2/8\omega_j) dt \\
    &\mathbbm{P}\left(Q_1^0(u) \mid \beta_j\neq 0\right) = \frac{1}{\sqrt{8\pi\omega_j}} \int_{t=-\infty}^A \exp(-(t-2\sqrt{2r\log p})^2/8\omega_j) dt
\end{align*}
where $A = \frac{2\omega_j(2u+\vartheta)+2r}{\sqrt{2r}} \sqrt{\log p}, B = \frac{2\omega_j(2u+2\delta+\vartheta)+2r}{\sqrt{2r}} \sqrt{\log p}$. Thus
\begin{align*}
    &\mathbbm{P}\left(Q_1^0(u+\delta)-Q_1^0(u) \mid \beta_j\neq 0\right) = \frac{1}{\sqrt{8\pi\omega_j}} \int_{t=\frac{2(r-\omega_j\vartheta-2\omega_ju-2\omega_j\delta)\sqrt{\log p}}{\sqrt{2r}}}^{\frac{2(r-\omega_j\vartheta-2\omega_ju)\sqrt{\log p}}{\sqrt{2r}}} \exp(-t^2/8\omega_j) dt \\
    &\mathbbm{P}\left(Q_1^0(u) \mid \beta_j\neq 0\right) = \frac{1}{\sqrt{8\pi\omega_j}} \int_{t=-\infty}^\frac{-2(r-\omega_j\vartheta-2\omega_ju)\sqrt{\log p}}{\sqrt{2r}} \exp(-t^2/8\omega_j) dt
\end{align*}
By the tail bound for CDF of standard Gaussian distribution, we know that
\begin{equation*}
    \mathcal{O}\left(
    \frac{1}{\sqrt{\log p}} - \frac{1}{(\sqrt{\log p})^3}
    \right) \leq \frac{\int_{t=-\infty}^\frac{-2(r-\omega_j\vartheta-2\omega_ju)\sqrt{\log p}}{\sqrt{2r}} e^{-t^2/8\omega_j} dt}{p^{-d^2}}  
    \leq \mathcal{O}\left(
    \frac{1}{\sqrt{\log p}}
    \right)
\end{equation*}
where $d = \frac{(r-\omega_j\vartheta-2\omega_ju)^2}{4\omega_jr}$,
\begin{equation*}
    \frac{1}{\sqrt{8\pi\omega_j\log p}} \int_{t=\frac{2(r-\omega_j\vartheta-2u-2\delta)\sqrt{\log p}}{\sqrt{2r}}}^{\frac{2(r-\omega_j\vartheta-2u)\sqrt{\log p}}{\sqrt{2r}}} \exp(-t^2/8\omega_j) dt \leq \mathcal{O} (\Delta p^{-d^2}p^{2d\Delta})
\end{equation*}
where $\Delta = \frac{4\delta}{\sqrt{2r}} = \frac{2\log 2}{\sqrt{2r}\log p}$. Combine them we have
\begin{align*}
    \mathbbm{P}\left(Q_1^0(u+\delta)-Q_1^0(u) \mid \beta_j\neq 0\right) &\leq \mathcal{O}(\sqrt{\log p} \Delta p^{-d^2}p^{2d\Delta})
    \mathbbm{P}\left(Q_1^0(u) \mid \beta_j\neq 0\right) \\
    &= \mathcal{O}(\frac{1}{\sqrt{\log p}} )
    \mathbbm{P}\left(Q_1^0(u) \mid \beta_j\neq 0\right)
\end{align*}
Now by Lemma \ref{lemma:geometric-insight}, 
\begin{align*}
    \mathbbm{P}\left(Q_2^0(u+\frac{\log 2}{2 \log p })-Q_2^0(u) \mid \beta_j\neq 0\right) &\leq 
    \mathbbm{P}(V_j\geq \frac{2\omega_j\vartheta + 2r}{\sqrt{2r}}\sqrt{\log p}|\beta_j\neq 0) \leq L_p p^{-\frac{(r+\omega_j\vartheta)^2}{4\omega_jr}} \\
    \mathbbm{P}(Q_1^0(u)|\beta_j\neq 0)&\geq 
    \mathbbm{P}(Q_1^0(0)|\beta_j\neq 0) \geq
    L_p p^{-\frac{(r-\omega_j\vartheta)^2}{4\omega_jr}}
\end{align*}
Thus 
\begin{equation*}
    \mathbbm{P}\left(Q_2^0(u+\frac{\log 2}{2 \log p })-Q_2^0(u) \mid \beta_j\neq 0\right) = \mathcal{O}(L_p p^{-\vartheta})\mathbbm{P}\left(Q_1^0(u) \mid \beta_j\neq 0\right)
\end{equation*}
We have 
\begin{equation*}
    \mathbbm{P}(\widetilde{M}_j<2u\log p|\beta_j\neq 0) \leq (1+o(1))\mathbbm{P}(Q_1^0(u)\cup Q_2^0(u)|\beta_j\neq 0)
\end{equation*}
Furthermore, by monotonicity of $Q_1^0(u)\cup Q_2^0(u)$,
\begin{align}
    \mathbbm{P}(\widetilde{M}_j<2u\log p|\beta_j\neq 0) &\leq (1+o(1))\mathbbm{P}(Q_1^0(u)\cup Q_2^0(u)|\beta_j\neq 0) \cr
\label{proof:hamming-error:bound-Q-u=0}
    &\leq (1+o(1))\mathbbm{P}(Q_1^0(0)\cup Q_2^0(0)|\beta_j\neq 0)
\end{align}
Similarly, since $R_1(u)\subset R_1^0(u)\subset R_1^0(0)$ and $R_2(u)\subset R_2^0(u)\subset R_2^0(0)$, we have
\begin{align}
\label{proof:hamming-error:bound-R-u=0}
    \mathbbm{P}(\widetilde{M}_j>2u\log p|\beta_j= 0)
    &\leq \mathbbm{P}(R_1^0(u)\cup R_2^0(u)|\beta_j= 0) \cr
    &\leq \mathbbm{P}(R_1^0(0)\cup R_2^0(0)|\beta_j= 0)
\end{align}
Now for any $v$, the region $\{M>\frac{4\omega_j v}{\sqrt{2r}}\sqrt{\log p}\}$ can be written as $S_1(v)\cup S_2(v)$
where
\begin{align*}
    S_1(v) &= \{|U_j|>\frac{4\omega_j v}{\sqrt{2r}}\sqrt{\log p}, |V_j|\leq \frac{4\omega_j v}{\sqrt{2r}}\sqrt{\log p}\} \\
    S_2(v) &= \{|U_j|>V_j, |V_j|>\frac{4\omega_j v}{\sqrt{2r}}\sqrt{\log p}\}
\end{align*}
the region $\{M<\frac{4\omega_j v}{\sqrt{2r}}\sqrt{\log p}\}$ can be written as $S_1(v)\cup S_2(v)$
where
\begin{align*}
    T_1(v) &= \{|U_j|<\frac{4\omega_j v}{\sqrt{2r}}\sqrt{\log p}, |V_j|\leq \frac{4\omega_j v}{\sqrt{2r}}\sqrt{\log p}\} \\
    T_2(v) &= \{|U_j|<V_j, |V_j|>\frac{4\omega_j v}{\sqrt{2r}}\sqrt{\log p}\}
\end{align*}
Note that the optimal choice of $v_{opt}=\frac{\omega_j\vartheta+r}{\sqrt{2\omega_j}}$. When $v=v_{opt}$, we know that the region $\{M>\frac{4\omega_j v}{\sqrt{2r}}\sqrt{\log p}\}$ can be written as $S_1\cup S_2$
where
\begin{align*}
    S_1 &= \{|U_j|>\frac{2\omega_j\vartheta+2r}{\sqrt{2r}}\sqrt{\log p}, |V_j|\leq\frac{2\omega_j\vartheta+2r}{\sqrt{2r}}\sqrt{\log p}\} \\
    S_2 &= \{|U_j|>V_j, |V_j|>\frac{2\omega_j\vartheta+2r}{\sqrt{2r}}\sqrt{\log p}\}
\end{align*}
and the region $\{M>\frac{4\omega_j v}{\sqrt{2r}}\sqrt{\log p}\}$ can be written as $T_1\cup T_2$
where
\begin{align*}
    T_1 &= \{|U_j|<\frac{2\omega_j\vartheta+2r}{\sqrt{2r}}\sqrt{\log p}, |V_j|\leq\frac{2\omega_j\vartheta+2r}{\sqrt{2r}}\sqrt{\log p}\} \\
    T_2 &= \{|U_j|<V_j, |V_j|>\frac{2\omega_j\vartheta+2r}{\sqrt{2r}}\sqrt{\log p}\} 
\end{align*}

by symmetric of the regions we have for any $u$,
\begin{equation}
\label{proof:hamming-error:bound-R0-u=0}
    \mathbbm{P}(R_1^0(u)\cup R_2^0(u)|\beta_j=0) \leq 
    \mathbbm{P}(S_1(v_{opt}+u)\cup S_2(v_{opt}+u)|\beta_j=0)
\end{equation}
Since
\begin{equation*}
    \mathbbm{P}(Q_1^0(0)\cup Q_2^0(0)|\beta_j\neq 0) -
    \mathbbm{P}(T_1\cup T_2|\beta_j\neq 0) = 
    \mathbbm{P}(S_2|\beta_j\neq 0) - \mathbbm{P}(R_1-S_1|\beta_j\neq 0)
\end{equation*}
By Lemma \ref{lemma:geometric-insight}, 
\begin{equation*}
    \mathbbm{P}(S_2|\beta_j\neq 0) = o\left(\mathbbm{P}(R_1-S_1|\beta_j\neq 0)\right)
\end{equation*}
Thus for large $p$ we have 
\begin{equation}
\label{proof:hamming-error:bound-Q0-u=0}
    \mathbbm{P}(Q_1^0(0)\cup Q_2^0(0)|\beta_j\neq 0) \leq 
    \mathbbm{P}(T_1\cup T_2|\beta_j\neq 0)
\end{equation}
Similarly, since
\begin{equation*}
    \mathbbm{P}(R_1^0(0)\cup R_2^0(0)|\beta_j= 0) -
    \mathbbm{P}(S_1\cup S_2|\beta_j= 0) = \mathbbm{P}(R_1-S_1|\beta_j= 0) - 
    \mathbbm{P}(S_2|\beta_j= 0)
\end{equation*}
By Lemma \ref{lemma:geometric-insight}, 
\begin{equation*}
    \mathbbm{P}(S_2|\beta_j= 0) = o\left(\mathbbm{P}(R_1-S_1|\beta_j= 0)\right)
\end{equation*}
Thus for large $p$ we have 
\begin{equation}
\label{proof:hamming-error:bound-R0-u=0}
    \mathbbm{P}(R_1^0(0)\cup R_2^0(0)|\beta_j= 0) \leq 
    \mathbbm{P}(S_1\cup S_2|\beta_j= 0)
\end{equation}
From \eqref{proof:hamming-error:bound-Q-u=0}, \eqref{proof:hamming-error:bound-R-u=0}, \eqref{proof:hamming-error:bound-R0-u=0} and \eqref{proof:hamming-error:bound-Q0-u=0} we have 
\begin{align*}
    \text{FP}_{\widetilde{M}} \leq \text{FP}_M(v_{opt}), \quad
    \text{FN}_{\widetilde{M}} \leq \text{FN}_M(v_{opt}) 
\end{align*}
Thus 
\begin{equation}
    \label{proof:hamming-error-first-case}
    \text{Hamm}^\star_{\widetilde{M}}  = (1+o(1))\cdot\text{Hamm}_M(v_{opt})
\end{equation}
Now for any $0<u<\frac{r-\omega_j\vartheta}{8\omega_j}$, 
\begin{align*}
    \mathbbm{P}(R_1(u)-S_1(v_{opt}+u)|\beta_j\neq 0) &\geq 
    \mathbbm{P}(R_1(\frac{r-\omega_j\vartheta}{8\omega_j})-S_1(v_{opt}+\frac{r-\omega_j\vartheta}{8\omega_j})|\beta_j\neq 0) \\
    &\geq L_p p^{-\frac{16(r+\omega_j\vartheta)^2+(r-\omega_j\vartheta)^2}{32\omega_j}}
\end{align*}
while 
\begin{equation*}
    \mathbbm{P}(S_2(v_{opt}+u)|\beta_j\neq 0)\leq \mathbbm{P}(S_2(v_{opt})) = L_p p^{-\frac{(3r+\omega_j\vartheta)^2}{4\omega_j}}
\end{equation*}
Thus for large $p$, we have
\begin{equation*}
    \mathbbm{P}(S_2(v_{opt}+u)|\beta_j\neq 0) \leq \mathbbm{P}(R_1(u)-S_1(v_{opt}+u)|\beta_j\neq 0)
\end{equation*}
Thus
\begin{align}
\label{proof:hamming-error:bound-Q0}
    \mathbbm{P}(Q_1^0(u)\cup Q_2^0(u)|\beta_j&\neq 0) -
    \mathbbm{P}(T_1(v_{opt}+u)\cup T_2(v_{opt}+u)|\beta_j\neq 0) \cr
    &= 
    \mathbbm{P}(S_2(v_{opt}+u)|\beta_j\neq 0) - \mathbbm{P}(R_1(u)-S_1(v_{opt}+u)|\beta_j\neq 0) \leq 0
\end{align}
Similarly, 
\begin{equation*}
    \mathbbm{P}(R_1(u)-S_1(v_{opt}+u)|\beta_j= 0) \leq 
    \mathbbm{P}(R_1(0)-S_1(v_{opt})|\beta_j= 0) 
    \leq L_p p^{-\frac{(r+\omega_j\vartheta)^2}{4\omega_j}}
\end{equation*}
while
\begin{equation*}
     \mathbbm{P}(S_2(v_{opt}+u)|\beta_j= 0) \leq 
     \mathbbm{P}(S_2(v_{opt}+\frac{r-\omega_j\vartheta}{8\omega_j})|\beta_j= 0) 
     \leq L_pp^{-\frac{(r+\omega_j\vartheta+(r-\omega_j\vartheta)/4)^2}{4\omega_j}}
     \leq L_pp^{-\frac{5^2(r+\omega_j\vartheta)^2}{4\cdot 4^2\omega_j}}
\end{equation*}
Thus
\begin{align}
\label{proof:hamming-error:bound-R0}
    \mathbbm{P}(R_1^0(u)\cup R_2^0(u)|&\beta_j= 0) -
    \mathbbm{P}(T_1(v_{opt}+u)\cup T_2(v_{opt}+u)|\beta_j= 0) \cr
    &= 
    \mathbbm{P}(R_1(u)-S_1(v_{opt}+u)|\beta_j= 0)- \mathbbm{P}(S_2(v_{opt}+u)|\beta_j= 0)  \leq 0
\end{align}
From \eqref{proof:hamming-error:bound-Q-u=0}, \eqref{proof:hamming-error:bound-R-u=0}, \eqref{proof:hamming-error:bound-Q0} and \eqref{proof:hamming-error:bound-R0}, for any $0<u<\frac{r-\omega_j\vartheta}{8\omega_j}$, we have
\begin{equation}
    \label{proof:hamming-error-second-case}
    \text{Hamm}^\star_{\widetilde{M}} \leq \text{Hamm}_{\widetilde{M}}(u)\leq (1+o(1))\text{Hamm}_{M}(v_{opt}+u)
\end{equation}
For $u\geq \frac{r-\omega_j\vartheta}{8\omega_j}$, we have
\begin{align*}
    \mathbbm{P}(T_1(v_{opt}+u)\cup T_2(v_{opt}+u)|\beta_j\neq 0) &\geq \mathbbm{P}(T_1(v_{opt}+\frac{r-\omega_j\vartheta}{8\omega_j})\cup T_2(v_{opt}+\frac{r-\omega_j\vartheta}{8\omega_j})|\beta_j\neq 0)\cr
    &\geq L_p p^{-\frac{7^2(r-\omega_j\vartheta)^2}{4\cdot 8^2\omega_j}}
\end{align*}
Thus for large $p$, 
\begin{align*}
    p^{-\vartheta}\mathbbm{P}(T_1(v_{opt}+u)\cup T_2(v_{opt}+u)|\beta_j\neq 0) &\geq p^{-\vartheta}\mathbbm{P}(T_1(v_{opt})\cup T_2(v_{opt})|\beta_j\neq 0) \cr
    &+
    \mathbbm{P}(S_1(v_{opt})\cup S_2(v_{opt})|\beta_j= 0)
\end{align*}
Combined with \eqref{proof:hamming-error-first-case}, we have
\begin{equation}
\label{proof:hamming-error-3rd-case}
\text{Hamm}^\star_{\widetilde{M}}\leq  (1+o(1))\text{Hamm}_M(v_{opt})\leq (1+o(1))\text{Hamm}_M(v_{opt}+u)
\end{equation}
For any $0\leq v < v_{opt}$, we have 
\begin{equation*}
    \text{Hamm}_M(v) - \text{Hamm}_M(v_{opt}) = 
    \int_{(U_j,V_j)\in \mathcal{R}_1} \frac{1}{\sqrt{8\pi\omega_j}} e^{-\frac{V_j^2}{8\omega_j}}\left(e^{-\frac{U_j^2}{8\omega_j}}
    - e^{-\frac{(U_j-\sqrt{2r\log p})^2}{8\omega_j}-\vartheta}\right)
\end{equation*}
where 
\begin{equation*}
    \mathcal{R}_1 = S(v)-S(v_{opt}) \subset \{|U_j|\leq \frac{2(r+\omega_j)}{\sqrt{2r}}\sqrt{\log p}\}
\end{equation*}
Thus the integrand is always positive, combined with \eqref{proof:hamming-error-first-case} we have
\begin{equation}
\label{proof:hamming-error-final-case}\text{Hamm}^\star_{\widetilde{M}} \leq (1+o(1))\text{Hamm}_M(v_{opt}) \leq (1+o(1))\text{Hamm}_M(v)
\end{equation}
Thus from \eqref{proof:hamming-error-first-case}, \eqref{proof:hamming-error-second-case}, \eqref{proof:hamming-error-3rd-case} and \eqref{proof:hamming-error-final-case}, we have completed the proof.

\subsection{Proof for Proposition \ref{prop:tildeM-Mstar-comparison}}

Without lose of generality we can set $j\in 2\mathbbm{Z}, j+1 \in 2\mathbbm{Z}+1$.
\begin{itemize}
    \item Calculation of $M_j^\star$. Let $\Omega = \begin{pmatrix}
           1 & \rho\\
           \rho & 1
       \end{pmatrix}$. Define 
    \begin{equation*}
    \begin{aligned}
        I &= (1-\epsilon_p)\exp\left(-\frac{(\widehat\beta_j^{(1)}, \widehat\beta_{j+1}^{(1)}) \Omega(\widehat\beta_j^{(1)},\widehat\beta_{j+1}^{(1)})^\intercal}{4}\right)+\epsilon_p\exp\left(-\frac{(\widehat\beta_j^{(1)}, \widehat\beta_{j+1}^{(1)}-\tau_p) \Omega(\widehat\beta_j^{(1)},\widehat\beta_{j+1}^{(1)}-\tau_p)^\intercal}{4}\right)\\
        II &= (1-\epsilon_p)\exp\left(-\frac{(\widehat\beta_j^{(1)}-\tau_p, \widehat\beta_{j+1}^{(1)}) \Omega(\widehat\beta_j^{(1)}-\tau_p,\widehat\beta_{j+1}^{(1)})^\intercal}{4}\right),\\
        &\quad\quad\quad\quad
        +\epsilon_p\exp\left(-\frac{(\widehat\beta_j^{(1)}-\tau_p, \widehat\beta_{j+1}^{(1)}-\tau_p) \Omega(\widehat\beta_j^{(1)}-\tau_p,\widehat\beta_{j+1}^{(1)}-\tau_p)^\intercal}{4}\right).
    \end{aligned}
    \end{equation*}
    Then 
\begin{equation*}
    M_j^\star = \log \left(\frac{1-\epsilon_p+\epsilon_pe^{-\tau_p(\tau_p-2\widehat\beta_j^{(2)})/4\omega_j}II/I}{1-\epsilon_p+\epsilon_pe^{-\tau_p(\tau_p+2\widehat\beta_j^{(2)})/4\omega_j}II/I}\right)
\end{equation*}
    \item $-\vartheta\log p-\tau_p(\tau_p+2\widehat\beta_j^{(2)})/4\omega_j+\log(II/I)<0$. In this case, 
    \begin{equation*}
    \begin{aligned}
     \mathbbm{P}(M_j>2u\log p) &= L_p \mathbbm{P}(-\vartheta\log p-\tau_p(\tau_p-2\widehat\beta_j^{(2)})/4\omega_j+\log II/I>2u\log p).
    \end{aligned}
    \end{equation*}
    \item $-\vartheta\log p-\tau_p(\tau_p+2\widehat\beta_j^{(2)})/4\omega_j+\log(II/I)>0$. In this case, 
    \begin{equation*}
        \mathbbm{P}(M_j>2u\log p) = L_p\mathbbm{P}(\widehat\beta_j^{(2)}>2\omega_ju\log p/\tau_p).
    \end{equation*}
\end{itemize}

The reason for these $L_p$ terms is the same as in the proof for Proposition \ref{prop:M-tildeM-comparison}. 
Now, regarding term $\log II/I$, we have:
\begin{itemize}
    \item $\rho\widehat\beta_j^{(1)}+\widehat\beta_{j+1}^{(1)}-(1/2+\rho)\tau_p>2v\log p/\tau_p$. In this case:
    \begin{equation*}
        \begin{aligned}
        I&= (1+o(1))\cdot\epsilon_p\exp\left(-\frac{(\widehat\beta_j^{(1)}, \widehat\beta_{j+1}^{(1)}-\tau_p) \Omega(\widehat\beta_j^{(1)},\widehat\beta_{j+1}^{(1)}-\tau_p)^\intercal}{4}\right),\\
            II&= (1+o(1))\cdot\epsilon_p\exp\left(-\frac{(\widehat\beta_j^{(1)}-\tau_p, \widehat\beta_{j+1}^{(1)}-\tau_p) \Omega(\widehat\beta_j^{(1)}-\tau_p,\widehat\beta_{j+1}^{(1)}-\tau_p)^\intercal}{4}\right).
        \end{aligned}
    \end{equation*}
    Thus we have:
    \begin{equation*}
        \log II/I = \tau_p(\widehat\beta_j^{(1)}+\rho\widehat\beta_{j+1}^{(1)}-(1/2+\rho)\tau_p)/2 +o(1).
    \end{equation*}
    \item $\rho\widehat\beta_j^{(1)}+\widehat\beta_{j+1}^{(1)}-(1/2+\rho)\tau_p<2v\log p/\tau_p$ but $\rho\widehat\beta_j^{(1)}+\widehat\beta_{j+1}^{(1)}-\tau_p/2>2v\log p/\tau_p$. In this case,
    \begin{equation*}
        \begin{aligned}
        I&= (1+o(1))\cdot \epsilon_p\exp\left(-\frac{(\widehat\beta_j^{(1)}, \widehat\beta_{j+1}^{(1)}-\tau_p) \Omega(\widehat\beta_j^{(1)},\widehat\beta_{j+1}^{(1)}-\tau_p)^\intercal}{4}\right),\\
            II&= (1+o(1))\cdot(1-\epsilon_p)\exp\left(-\frac{(\widehat\beta_j^{(1)}-\tau_p, \widehat\beta_{j+1}^{(1)}) \Omega(\widehat\beta_j^{(1)}-\tau_p,\widehat\beta_{j+1}^{(1)})^\intercal}{4}\right).
        \end{aligned}
    \end{equation*}
    Thus we have:
    \begin{equation*}
        \log II/I = v\log p+\tau_p((1-\rho)\widehat\beta_j^{(1)}-(1-\rho)\widehat\beta_{j+1}^{(1)})/2+o(1).
    \end{equation*}
    \item $\rho\widehat\beta_j^{(1)}+\widehat\beta_{j+1}^{(1)}-\tau_p/2<2v\log p/\tau_p$. In this case,
    \begin{equation*}
        \begin{aligned}
        I&= (1+o(1))\cdot(1-\epsilon_p)\exp\left(-\frac{(\widehat\beta_j^{(1)}, \widehat\beta_{j+1}^{(1)}) \Omega(\widehat\beta_j^{(1)},\widehat\beta_{j+1}^{(1)})^\intercal}{4}\right),\\
            II&= (1+o(1))\cdot(1-\epsilon_p)\exp\left(-\frac{(\widehat\beta_j^{(1)}-\tau_p, \widehat\beta_{j+1}^{(1)}) \Omega(\widehat\beta_j^{(1)}-\tau_p,\widehat\beta_{j+1}^{(1)})^\intercal}{4}\right),
        \end{aligned}
    \end{equation*}
    \begin{equation*}
        \log II/I =  \tau_p(\widehat\beta_j^{(1)}+\rho\widehat\beta_{j+1}^{(1)}-\tau_p/2)/2 +o(1).
    \end{equation*}
\end{itemize}

First we consider when $\rho>0$. For the first case of $M_j^\star.$ In this case we have:
\begin{equation*}
    \vartheta\log p+\tau_p^2/4\omega_j+2u\log p-\tau_p\widehat\beta_j^{(2)}/2\omega_j<\log II/I< \vartheta\log p+\tau_p^2/4\omega_j+\tau_p\widehat\beta_j^{(2)}/2\omega_j.
\end{equation*}
For the above to hold, $\widehat\beta_j^{(2)}> 2\omega_ju\log p/\tau_p$. Thus, we consider three different cases:
\begin{itemize}
    \item (1.1)\ \ $\rho\widehat\beta_j^{(1)}+\widehat\beta_{j+1}^{(1)}>2v\log p/\tau_p+(1/2+\rho)\tau_p$. Pugging in $\log II/I$, we have:
    \begin{equation*}
        \begin{aligned}
        \widehat\beta_j^{(1)}+\rho\widehat\beta_{j+1}^{(2)}&<2\vartheta\log p/\tau_p+(1/2+\rho)\tau_p+\tau_p/2\omega_j+\widehat\beta_j^{(2)}/\omega_j,\\
        \widehat\beta_j^{(1)}+\rho\widehat\beta_{j+1}^{(2)}&>2\vartheta\log p/\tau_p+(1/2+\rho)\tau_p+\tau_p/2\omega_j+4u\log p/\tau_p-\widehat\beta_j^{(2)}/\omega_j.\\   
        \end{aligned}
    \end{equation*}
     
    \item (1.2)\ \ $2\vartheta\log p/\tau_p+\tau_p/2<\rho\widehat\beta_j^{(1)}+\widehat\beta_{j+1}^{(1)}<2v\log p/\tau_p+ (1/2+\rho)\tau_p$. Plugging in $\log II/I$, we have:
    \begin{equation*}
        \begin{aligned}
        (1-\rho)\widehat\beta_j^{(1)}-(1-\rho)\widehat\beta_{j+1}^{(1)}&<\tau_p/2\omega_j+\widehat\beta_j^{(2)}/\omega_j,\\
        (1-\rho)\widehat\beta_j^{(1)}-(1-\rho)\widehat\beta_{j+1}^{(1)}&>\tau_p/2\omega_j+4u\log p/\tau_p-\widehat\beta_j^{(2)}/\omega_j.
        \end{aligned}
    \end{equation*}
    \item (1.3)\ \ $\rho\widehat\beta_j^{(1)}+\widehat\beta_{j+1}^{(1)}<2v\log p/\tau_p+ \tau_p/2$. Pugging in $\log II/I$, we have:
    \begin{equation*}
        \begin{aligned}
        \widehat\beta_j^{(1)}+\rho\widehat\beta_{j+1}^{(2)}&<2\vartheta\log p/\tau_p+\tau_p/2+\tau_p/2\omega_j+\widehat\beta_j^{(2)}/\omega_j,\\
        \widehat\beta_j^{(1)}+\rho\widehat\beta_{j+1}^{(2)}&>2\vartheta\log p/\tau_p+\tau_p/2+\tau_p/2\omega_j+4u\log p/\tau_p-\widehat\beta_j^{(2)}/\omega_j.\\   
        \end{aligned}
    \end{equation*}
\end{itemize}
Then we plot the region for the second case of $M_j^\star$. We further have three cases:
\begin{itemize}
    \item (2.1)\ \ 
    $\rho\widehat\beta_j^{(1)}+\widehat\beta_{j+1}^{(1)}>2v\log p/\tau_p+(1/2+\rho)\tau_p$. Pugging in $\log II/I$, we have:
    \begin{equation*}
        \begin{aligned}
        \widehat\beta_j^{(1)}+\rho\widehat\beta_{j+1}^{(2)}&>2\vartheta\log p/\tau_p+(1/2+\rho)\tau_p+\tau_p/2\omega_j+\widehat\beta_j^{(2)}/\omega_j,\\ 
        \widehat\beta_j^{(2)}/\omega_j&>2u\log p/\tau_p.
        \end{aligned}
    \end{equation*}

    \item (2.2)\ \ $2\vartheta\log p/\tau_p+\tau_p/2<\rho\widehat\beta_j^{(1)}+\widehat\beta_{j+1}^{(1)}<2v\log p/\tau_p+ (1/2+\rho)\tau_p$. Plugging in $\log II/I$, we have:
    \begin{equation*}
        \begin{aligned}
        (1-\rho)\widehat\beta_j^{(1)}-(1-\rho)\widehat\beta_{j+1}^{(1)}&>\tau_p/2\omega_j+\widehat\beta_j^{(2)}/\omega_j,\\
        \widehat\beta_j^{(2)}/\omega&>2u\log p/\tau_p.
        \end{aligned}
    \end{equation*}
    
    \item (2.3)\ \ $\rho\widehat\beta_j^{(1)}+\widehat\beta_{j+1}^{(1)}<2v\log p/\tau_p+ \tau_p/2$. Pugging in $\log II/I$, we have:
    \begin{equation*}
        \begin{aligned}
        \widehat\beta_j^{(1)}+\rho\widehat\beta_{j+1}^{(2)}&>2\vartheta\log p/\tau_p+\tau_p/2+\tau_p/2\omega_j+\widehat\beta_j^{(2)}/\omega_j,\\
        \widehat\beta_j^{(2)}/\omega_j&>2u\log p/\tau_p . 
        \end{aligned}
    \end{equation*}
    
\end{itemize}


\newcommand{\singletheta}{4\vartheta}
\newcommand{\doubletheta}{8\vartheta}
\newcommand{\fracomega}{1+1/\omega_j}

Since we are working on Spike-and-Slab prior, there are 4 cases to be considered, i.e. 
\begin{align*}
    &\beta_j=\beta_{j+1}=0, \\
    &\beta_j=0, \beta_{j+1}= \tau_p,\\
    &\beta_j=\tau_p, \beta_{j+1}=0, \\
    &\beta_j=\tau_p, \beta_{j+1}=\tau_p.
\end{align*}
The probability assignment for each case is different
\begin{align*}
    &\mathbbm{P}(\beta_j=\beta_{j+1}=0) = (1-p^{-\vartheta})^2 \to 1, \\
    &\mathbbm{P}(\beta_j=0, \beta_{j+1}\neq 0) = 
    \mathbbm{P}(\beta_j\neq 0, \beta_{j+1}= 0) = p^{-\vartheta}(1-p^{-\vartheta}) \to p^{-\vartheta}, \\
    &\mathbbm{P}(\beta_j\neq 0, \beta_{j+1}\neq 0) = p^{-2\vartheta}.
\end{align*}
Then the Hamming error is 
\begin{align*}
    \text{Hamm} &= \min_{u} \sum_{\beta_j=0}\mathbbm{P}(M_j>2u\log p|\beta_j,\beta_{j+1})\mathbbm{P}(\beta_j,\beta_{j+1}) \\
    &+ \sum_{\beta_j\neq0}\mathbbm{P}(M_j<2u\log p|\beta_j,\beta_{j+1})\mathbbm{P}(\beta_j,\beta_{j+1}). 
\end{align*}
The Hamming error is related to the probability of $\{M_j>2u\log p\}$ or $\{M_j<2u\log p\}$, which is a region of 
\begin{equation*}
    (\xi_{j,1},\xi_{j,2}, \xi_{j+1,1}) \sim \mathcal{N}((\xi_{j,1}^0,\xi_{j,2}^0 .\xi_{j+1,1}^0), \frac{1}{2\log p}\Sigma) \in \mathbb{R}^3
\end{equation*}
where 
\begin{align*}
    &\xi_{j,1} = \widehat{\beta}_j^{(1)} / \sqrt{\log p} \\ 
    &\xi_{j,2} = \widehat{\beta}_j^{(2)} / \sqrt{\log p} \\ 
    &\xi_{j+1,1} = \widehat{\beta}_{j+1}^{(1)} / \sqrt{\log p} \\ 
    &\Sigma^{-1} = \frac{1}{4}  
    \begin{pmatrix}
       1 & \rho & 0\\
       \rho & 1 & 0 \\
       0 & 0 & 1-\rho^2
    \end{pmatrix}
\end{align*} 
Denote the square Mahalanobis distance of $(\xi_{j,1}^0,\xi_{j,2}^0 \xi_{j+1,1}^0)$ to the corresponding region, i.e., $\{M_j>2u\log p\}$ for case 1 and 2 and $\{M_j<2u\log p\}$ for case 3 and 4, as $D_1^2,D_2^2,D_3^2,D_4^2$. By Lemma \ref{lemma:geometric-insight}, we have
\begin{align*}
    f_\text{Hamm} &= \min_{u} 
    \max\{1-D_1^2, 1-\vartheta+D_2^2, 1-\vartheta+D_3^2, 1-2\vartheta+D_4^2\} \\
    &= \min_{u} 
    1- \frac{1}{4}\min\{d_1(u),d_2(u),d_3(u),d_4(u)\} \\
    &\leq 1- \frac{1}{4}\min\{d_1(0),d_2(0),d_3(0),d_4(0)\}
\end{align*}
where $d_1(u),d_2(u),d_3(u),d_4(u)$ defined above is the  square distance absorbing the $\vartheta$ term. And we define $d_k=d_k(0)$ for $k=1,2,3,4$. Denote 
\begin{align*}
    &D = \singletheta + \frac{((1/2+1/2\omega_j)\tau_p-2\vartheta/\tau_p)_+^2}{\fracomega} \\
    &\widetilde{D} = 4\vartheta + (1-\rho)(3+\rho)\tau_p^2/4 \\
    & \widetilde{\widetilde{D}} = 8\vartheta + (1-\rho^2)\tau_p^2 \geq \widetilde{D}
\end{align*}
where $\omega_j=1/(1-\rho^2)$. In each case, we further define 
\begin{align*}
    &(x^0,y^0,z^0) \in \text{ corresponding region of $(\xi_{j,1},\xi_{j,2}, \xi_{j+1,1})$} \\
    &(x,y,z) = (x^0,y^0,z^0) - (\xi_{j,1}^0,\xi_{j,2}^0, \xi_{j+1,1}^0)
\end{align*}
Thus the distance is simplified as 
\begin{equation*}
    d_k = \argmin_{(x,y,z) \in \text{ corresponding region}} x^2+z^2+2\rho xz + (1-\rho^2)y^2
\end{equation*}
for $k=1,2,3,4$. To simplify the notation, we ignore the $\sqrt{\log p}$ term during the following proof, so we treat $\tau_p = \sqrt{2r}$ for the rest of the proof.

\begin{itemize}
    \item Case (1.1) and (2.1)
    \begin{itemize}
        \item $\beta_j=\beta_{j+1}=0$.
        In both Case 1.1 and Case 2.1 we have 
        \begin{equation*}
            x + \rho z + y/\omega_j > 2\vartheta/\tau_p + (1/2+1/2\omega_j+\rho)\tau_p
        \end{equation*}
        By Cauchy inequality,
        \begin{equation*}
            d_1 =  \frac{(2\vartheta/\tau_p +(1/2+1/2\omega_j+\rho)\tau_p)^2}{\fracomega} \geq D
        \end{equation*}
        \item $\beta_j=0, \beta_{j+1}=\tau_p$. 
        In both Case 1.1 and Case 2.1 we have
        \begin{equation*}
            x + \rho z + y/\omega_j > 2\vartheta/\tau_p + (1/2+1/2\omega_j)\tau_p
        \end{equation*}
        By Cauchy inequality,
        \begin{equation*}
            d_2 =  4\vartheta + \frac{(2\vartheta/\tau_p +(1/2+1/2\omega_j)\tau_p)^2}{\fracomega} \geq D
        \end{equation*}

        \item $\beta_j=\tau_p, \beta_{j+1}=0$.
        In Case 1.1 we have
        \begin{align*}
            &\min_{x,y,z} \ (1-\rho^2)(x^2+y^2) + (\rho x+z)^2 \geq (1-\rho^2)(x+y)^2/2 + (\rho x+z)^2\\
            &s.t. \ \quad x+\rho z+y/\omega \leq 2\vartheta/\tau_p +\rho \tau_p - (1/2+1/2/\omega)\tau_p \\
            &and \quad \rho x+ z \geq 2\vartheta/\tau_p+\tau_p/2
        \end{align*}
        Combine two above constrains, we have
        \begin{align*}
            (1-\rho^2)(x+y) &\leq 2\vartheta/\tau_p + \rho \tau_p - (1-\rho^2/2)\tau_p - 2\rho\vartheta/\tau_p - \rho\tau_p/2 \\
            &= (\rho^2/2+\rho/2-1)\tau_p + 
            2(1-\rho)\vartheta/\tau_p \\
            &=-(1+\rho/2)(1-\rho)\tau_p + 2(1-\rho)\vartheta/\tau_p 
        \end{align*}
        To prove $d_3\geq D$, it's sufficient to prove that,
        \begin{equation*}
            \frac{((1/2+1/2/\omega)\tau_p-2\vartheta/\tau_p)_+^2}{1+1/\omega} \leq 
            (\tau_p/2+2\vartheta/\tau_p)^2+
            \frac{1-\rho}{2(1+\rho)}((1+\rho/2)\tau_p-2\vartheta/\tau_p)_+^2
        \end{equation*}
        which can be proved by comparing the coefficient of $(\tau_p, 2\vartheta/\tau_p)$.
        
        In Case 2.1, we have 
        \begin{align*}
            &\rho x + z \geq 2\vartheta/\tau_p + 1/2 \tau_p \\ 
            &y\leq -\tau_p
        \end{align*}
        From that we have
        \begin{equation*}
            d_3 \geq 4\vartheta + (1-\rho^2)\tau_p^2 + (2\vartheta/\tau_p+\tau_p/2)^2/2 \geq D
        \end{equation*}

        \item $\beta_j=\beta_{j+1}=\tau_p$. 
        In Case 1.1, we have 
        \begin{equation*}
            x + \rho z + y/\omega_j \geq 2\vartheta/\tau_p + (1/2+\rho+1/2\omega_j)\tau_p
        \end{equation*}
        By Cauchy inequality, 
        \begin{equation*}
             d_4 \geq \doubletheta + \frac{((1/2+1/2\omega_j)\tau_p-2\vartheta/\tau_p )_+^2}{\fracomega} \geq D
        \end{equation*}
        In Case 2.1, we have
        \begin{align*}
            &\rho x + z \geq 2 \vartheta/\tau_p - 1/2\tau_p \\
            & y \leq -\tau_p
        \end{align*}
        From that we have
        \begin{equation*}
            d_4 \geq 8\vartheta + (1-\rho^2)\tau_p^2 \geq \widetilde{D}
        \end{equation*}
        
    \end{itemize}

    \item Case (1.3) and Case (2.3)
    \begin{itemize}
        \item $\beta_j=\beta_{j+1}=0$. In both Case 1.3 and Case 2.3 we have
        \begin{equation*}
            x+\rho z + y/\omega_j \geq 2\vartheta/\tau_p +(1/2+1/2\omega_j)\tau_p
        \end{equation*}
        By Cauchy inequality, 
        \begin{equation*}
            d_1 =  \frac{(2\vartheta/\tau_p +(1/2+1/2\omega_j)\tau_p)^2}{\fracomega} \geq D
        \end{equation*}

        \item $\beta_j=0,\beta_{j+1}=\tau_p$. 
        In Case 1.3, we have 
        \begin{align*}
            &\rho x + z \leq 2\vartheta/\tau_p - \tau_p/2 \\
            &x+\rho z + y/\omega_j \geq 2\vartheta/\tau_p + (1/2\omega_j+1/2-\rho)\tau_p
        \end{align*}
        From that we have 
        \begin{align*}
            &\rho x + z \leq 2\vartheta/\tau_p - \tau_p/2 \\
            &(1-\rho^2)(x+y) \geq (1+\rho/2)(1-\rho)\tau_p + 2(1-\rho)\vartheta/\tau_p
        \end{align*}
        By Cauchy inequality,
        \begin{align*}
            d_2 &\geq 4\vartheta + (\tau_p/2-2\vartheta/\tau_p)_+^2+
            \frac{1-\rho}{2(1+\rho)}((1+\rho/2)\tau_p+2\vartheta/\tau_p)^2 \geq D
        \end{align*}
        To prove the second inequality, we can simply compare the coefficient of $\tau_p^2, \vartheta, 2\vartheta/\tau_p$, which leads to prove:
        \begin{align*}
            \frac{1}{4}(1+\frac{1}{\omega}) &\leq \frac{1}{4} + \frac{1-\rho}{2(1+\rho)}(1+\rho/2)^2 \\
            -2 &\leq -2 + \frac{2(1-\rho)(1+\rho/2)}{(1+\rho)} \\
            \frac{1}{1+1/\omega} &\leq 1 + \frac{1-\rho}{2(1+\rho)}
        \end{align*}
        In Case 2.3, we have
        \begin{equation*}
            x + \rho z + y/\omega_j \geq 2\vartheta/\tau_p + (1/2+1/2\omega_j)\tau_p
        \end{equation*}
        By Cauchy inequality,
        \begin{equation*}
            d_2 \geq 4\vartheta + \frac{(2\vartheta/\tau_p +(1/2+1/2\omega_j)\tau_p)^2}{\fracomega} \geq D
        \end{equation*}

        \item $\beta_j=\tau_p,\beta_{j+1}=0$. 
        In Case 1.3, we have 
        \begin{equation*}
            x + \rho z + y/\omega_j \leq -(1/2+1/2\omega_j)\tau_p + 2\vartheta/\tau_p
        \end{equation*}
        By Cauchy inequality,
        \begin{equation*}
            d_2 \geq 4\vartheta + \frac{((1/2+1/2\omega_j)\tau_p-2\vartheta/\tau_p)_+^2}{\fracomega} = D
        \end{equation*}
        In Case 2.3, we have 
        \begin{align*}
            &\rho x + z \leq 2\vartheta/\tau_p +\tau_p/2 - \rho \tau_p
        \end{align*}
        From that we have
        \begin{align*}
            d_4 &\geq 4\vartheta + (1-\rho^2)y^2 + (x+\rho z)^2 \\
            &\geq 4\vartheta +(1-\rho^2)\tau_p^2 +  ((\rho-1/2)\tau_p-2\vartheta/\tau_p)^2_+ \\
            &\geq D
        \end{align*}
        The last inequality comes from comparing the coefficients for $(\tau_p, \vartheta, 2\vartheta/\tau_p)$, which is 
        \begin{align*}
            \frac{1}{4}(1+\frac{1}{\omega}) &\leq (1-\rho^2) + (\rho-1/2)_+^2 \\ 
            -2 &\leq -2(2\rho-1)_+\vartheta \\
            \frac{1}{1+1/\omega} &\leq 1
        \end{align*}

        \item $\beta_j=\beta_{j+1}=\tau_p$.
        In Case 1.3, we have 
        \begin{equation*}
            x + \rho z + y/\omega_j \leq -2\vartheta/\tau_p -(1/2+1/2\omega_j)\tau_p
        \end{equation*}
        By Cauchy inequality,
        \begin{equation*}
            d_4 \geq 8\vartheta + \frac{((1/2+1/2\omega_j)\tau_p+2\vartheta/\tau_p)^2}{\fracomega} \geq D
        \end{equation*}
        In Case 2.3, we have 
        \begin{align*}
            &y\leq -\tau_p \\
            &\rho x +z \leq 2\vartheta/\tau_p - (1/2+\rho)\tau_p
        \end{align*}
        From that we have
        \begin{align*}
            d_4 &\geq 8\vartheta + (1-\rho^2)y^2 + (\rho x+z)^2 \\
            &\geq 4\vartheta + (1-\rho^2)\tau_p^2 +  ((1/2+\rho)\tau_p - 2\vartheta/\tau_p)^2_+/2 \\
            &\geq D
        \end{align*}
    \end{itemize}
    
    
    \item Case (1.2) and Case (2.2)
    \begin{itemize}
        \item $\beta_j=\beta_{j+1}=0$.
        In both Case 1.2 and Case 2.2, we have 
        \begin{equation*}
            x + \rho z + y/\omega_j \geq 2\vartheta/\tau_p +(1/2+1/2\omega_j)\tau_p
        \end{equation*}
        By Cauchy inequality,
        \begin{equation*}
            d_1 =  \frac{(2\vartheta/\tau_p +(1/2+1/2\omega_j)\tau_p)^2}{\fracomega} \geq D
        \end{equation*}
        
        \item $\beta_j=0,\beta_{j+1}=\tau_p$.
        In both Case 1.2 and Case 2.2, we have
        \begin{equation*}
            (1-\rho)(x-z)+y/\omega_j\geq (1-\rho+1/2\omega_j)\tau_p
        \end{equation*}
        From that,
        \begin{align*}
            d_2 &= 4\vartheta + \frac{1+\rho}{2}(x+z)^2+\frac{1-\rho}{2}(x-z)^2 +y^2/\omega_j \\
            &\geq 4\vartheta + (1-\rho)(3+\rho)\tau_p^2/4
        \end{align*} 

        \item $\beta_j=\tau_p, \beta_{j+1}=0$.
        In Case 1.2 we have 
        \begin{equation*}
            (1-\rho)(x-z)+y/\omega_j \leq -(1/2\omega_j+1-\rho)\tau_p
        \end{equation*}
        From that,
        \begin{align*}
            d_3 &= 4\vartheta + \frac{1+\rho}{2}(x+z)^2+\frac{1-\rho}{2}(x-z)^2 +y^2/\omega_j \\
            &\geq 4\vartheta + (1-\rho)(3+\rho)\tau_p^2/4
        \end{align*} 

        In Case 2.2 we have 
        \begin{align*}
            &\rho x + z \leq 2\vartheta/\tau_p - 1/2\tau_p \\
            &y\leq -\tau_p
        \end{align*}
        From that we have
        \begin{align*}
            d_3 &\geq 4\vartheta + (1-\rho^2)\tau_p^2 +  (1/2\tau_p-2\vartheta/\tau_p)^2_+ \\
            &\geq 4\vartheta + (1-\rho^2)\tau_p^2 +  ((\rho-1/2)\tau_p-2\vartheta/\tau_p)^2_+\\
            &\geq D
        \end{align*}

        \item $\beta_j=\beta_{j+1}=\tau_p$.
        In Case 1.2 we have
        \begin{equation*}
            (1-\rho^2)(x+y)\leq -(1+\rho/2)(1-\rho)\tau_p + 2(1-\rho)\vartheta/\tau_p
        \end{equation*}
        From that we have 
        \begin{align*}
            d_4 &\geq 8\vartheta + (\tau_p/2-2\vartheta/\tau_p)_+^2+
            \frac{1-\rho}{2(1+\rho)}((1+\rho/2)\tau_p-2\vartheta/\tau_p)_+^2 \geq D
        \end{align*}
        In Case 2.2, we have
        \begin{align*}
            &\rho x + z \leq 2\vartheta/\tau_p - 1/2\tau_p\\
            &y \leq -\tau_p
        \end{align*}
        From that we have
        \begin{align*}
            d_4 &\geq 8\vartheta + (1-\rho^2)\tau_p^2 +  (1/2\tau_p-2\vartheta/\tau_p)^2_+ \\
            &\geq 4\vartheta + (1-\rho^2)\tau_p^2 +  ((\rho-1/2)\tau_p-2\vartheta/\tau_p)^2_+\\
            &\geq D
        \end{align*}
    \end{itemize}
\end{itemize}


    
    

Next, we analyze when $\rho<0$. 
For the first case of $M_j^\star.$, we have:
\begin{equation*}
    \vartheta\log p+\tau_p^2/4\omega_j+2u\log p-\tau_p\widehat\beta_j^{(2)}/2\omega_j<\log II/I< \vartheta\log p+\tau_p^2/4\omega_j+\tau_p\widehat\beta_j^{(2)}/2\omega_j
\end{equation*}
\begin{itemize}
    \item (1.1)\ \ $\rho\widehat\beta_j^{(1)}+\widehat\beta_{j+1}^{(1)}>2v\log p/\tau_p+\tau_p/2$. Pugging in $\log II/I$, we have:
    \begin{equation*}
        \begin{aligned}
        \widehat\beta_j^{(1)}+\rho\widehat\beta_{j+1}^{(2)}&<2\vartheta\log p/\tau_p+(1/2+\rho)\tau_p+\tau_p/2\omega_j+\widehat\beta_j^{(2)}/\omega_j\\
        \widehat\beta_j^{(1)}+\rho\widehat\beta_{j+1}^{(2)}&>2\vartheta\log p/\tau_p+(1/2+\rho)\tau_p+\tau_p/2\omega_j+4u\log p/\tau_p-\widehat\beta_j^{(2)}/\omega_j\\   
        \end{aligned}
    \end{equation*}
     
    \item (1.2)\ \ $2\vartheta\log p/\tau_p+(1/2+\rho)\tau_p<\rho\widehat\beta_j^{(1)}+\widehat\beta_{j+1}^{(1)}<2v\log p/\tau_p+ \tau_p/2$. Plugging in $\log II/I$, we have:
    \begin{equation*}
        \begin{aligned}
        (1+\rho)\widehat\beta_j^{(1)}+(1+\rho)\widehat\beta_{j+1}^{(1)}&<4\vartheta\log p/\tau_p+\tau_p/2\omega_j+(1-\rho)\tau_p+\widehat\beta_j^{(2)}/\omega_j\\
        (1+\rho)\widehat\beta_j^{(1)}+(1+\rho)\widehat\beta_{j+1}^{(1)}&<4\vartheta\log p/\tau_p+\tau_p/2\omega_j+(1-\rho)\tau_p+4u\log p/\tau_p-\widehat\beta_j^{(2)}/\omega_j
        \end{aligned}
    \end{equation*}
    
    \item (1.3)\ \ $\rho\widehat\beta_j^{(1)}+\widehat\beta_{j+1}^{(1)}<2v\log p/\tau_p+ (1/2+\rho)\tau_p$. Pugging in $\log II/I$, we have:
    \begin{equation*}
        \begin{aligned}
        \widehat\beta_j^{(1)}+\rho\widehat\beta_{j+1}^{(2)}&<2\vartheta\log p/\tau_p+\tau_p/2+\tau_p/2\omega_j+\widehat\beta_j^{(2)}/\omega_j\\
        \widehat\beta_j^{(1)}+\rho\widehat\beta_{j+1}^{(2)}&>2\vartheta\log p/\tau_p+\tau_p/2+\tau_p/2\omega_j+4u\log p/\tau_p-\widehat\beta_j^{(2)}/\omega_j\\   
        \end{aligned}
    \end{equation*}
\end{itemize}

Then we plot the region for the second case of $M_j^\star$.
\begin{itemize}
    \item (2.1)\ \ 
    $\rho\widehat\beta_j^{(1)}+\widehat\beta_{j+1}^{(1)}>2v\log p/\tau_p+\tau_p/2$. Pugging in $\log II/I$, we have:
    \begin{equation*}
        \begin{aligned}
        \widehat\beta_j^{(1)}+\rho\widehat\beta_{j+1}^{(2)}&>2\vartheta\log p/\tau_p+(1/2+\rho)\tau_p+\tau_p/2\omega_j+\widehat\beta_j^{(2)}/\omega_j\\ 
        \widehat\beta_j^{(2)}/\omega_j&>2u\log p/\tau_p
        \end{aligned}
    \end{equation*}

    \item (2.2)\ \ $2\vartheta\log p/\tau_p+(1/2+\rho)\tau_p<\rho\widehat\beta_j^{(1)}+\widehat\beta_{j+1}^{(1)}<2v\log p/\tau_p+ \tau_p/2$. Plugging in $\log II/I$, we have:
    \begin{equation*}
        \begin{aligned}
        (1+\rho)\widehat\beta_j^{(1)}+(1+\rho)\widehat\beta_{j+1}^{(1)}&>4\vartheta\log p/\tau_p+\tau_p/2\omega_j+(1-\rho)\tau_p+\widehat\beta_j^{(2)}/\omega_j\\
        (1+\rho)\widehat\beta_j^{(1)}+(1+\rho)\widehat\beta_{j+1}^{(1)}&<4\vartheta\log p/\tau_p+\tau_p/2\omega_j+(1-\rho)\tau_p+4u\log p/\tau_p-\widehat\beta_j^{(2)}/\omega_j
        \end{aligned}
    \end{equation*}
    
    \item (2.3)\ \ $\rho\widehat\beta_j^{(1)}+\widehat\beta_{j+1}^{(1)}<2v\log p/\tau_p+ (1/2+\rho)\tau_p$. Pugging in $\log II/I$, we have:
    \begin{equation*}
        \begin{aligned}
        \widehat\beta_j^{(1)}+\rho\widehat\beta_{j+1}^{(2)}&>2\vartheta\log p/\tau_p+\tau_p/2+\tau_p/2\omega_j+\widehat\beta_j^{(2)}/\omega_j\\
        \widehat\beta_j^{(2)}/\omega_j&>2u\log p/\tau_p  
        \end{aligned}
    \end{equation*}

    By considering the region one by one, we can prove the same result for $\rho<0$ similarly.

\end{itemize}

\subsection{Proof for Proposition \ref{prop:mcmc-best-group-case}}

The mirror statistic is 
\begin{equation*}
    M_j^\star = \log 
    \frac{1-\epsilon_p+\epsilon_p\exp\left(
    -\frac{1}{4}\left(
    3\tau_p^2-2\sqrt{\log p}\cdot\tau_p(\sqrt{2}t+s)\right)\right)}{1-\epsilon_p+\epsilon_p\exp\left(
    -\frac{1}{4}\left(
    3\tau_p^2-2\sqrt{\log p}\cdot \tau_p(\sqrt{2}t-s)\right)\right)},
\end{equation*}
where 
\begin{align*}
    t &= \frac{\widehat{\beta}_j^{(1)}+\widehat{\beta}_j^{(2)}}{\sqrt{2\log p}}, \\
    s&=\frac{\widehat{\beta}_{j+1}^{(1)}}{\sqrt{\log p}}, \\
    (t,s) &\sim \mathcal{N}\left((0,0),\frac{4I_2}{2\log p}\right) \text{ under null}, \\
    (t,s) &\sim \mathcal{N}\left((2\sqrt{r},\sqrt{2r}),\frac{4I_2}{2\log p}\right) \text{ under alternative}.
\end{align*}
The rejection region for $\{M_j>2u\log p\}$ can be roughly written as $R = R_1 \cup R_2$, where
\begin{align*}
    &R_1 = \{\sqrt{2}t-s\leq \frac{3r+2\vartheta}{\sqrt{2r}}; 
    \sqrt{2}t+s\geq \frac{4u+3r+2\vartheta}{\sqrt{2r}}\}, \\
    &R_2 = \{\sqrt{2}t-s\geq \frac{3r+2\vartheta}{\sqrt{2r}}; 
    s\geq \frac{2u}{\sqrt{2r}}\}.
\end{align*}
The different only contributes to the $L_p$ terms as in the proof of Proposition \ref{prop:M-tildeM-comparison}.
Denote $d_0(u)$ the square distance of $(0,0)$ to $R$ and $d_1(u)$ the square distance of $(0,0)$ to $R^c$ plus $4\vartheta$. Then we have
\begin{equation*}
    f_{\text{Hamm}} = 1-\frac{1}{4}\max_u\min\{d_0(u),d_1(u)\}.
\end{equation*}
The distance $d_1(u)$ is 
\begin{equation*}
    d_1 = 4\vartheta+\min \{
    \frac{(3r-4u-2\vartheta)^2_+}{6r},
    \frac{2(r-u)^2_+}{r}\}.
\end{equation*}
When $u\leq (3r+2\vartheta)/2$, we have
\begin{equation*}
    d_0 = \frac{(4u+3r+2\vartheta)^2}{6r} \geq d_1.
\end{equation*}
When $u\geq (3r+2\vartheta)/2$, we have
\begin{equation*}
    d_0 = |OG|^2 = \frac{(2u+3r+2\vartheta)^2+8u^2}{4r} \geq 4\vartheta.
\end{equation*}
When $u\leq (3r+2\vartheta)/2$, we have
\begin{equation*}
    \min\{d_0(u),d_1(u)\} \leq 
    4\vartheta + \frac{(3r-2\vartheta)_+^2}{6r},
\end{equation*}
and the inequality became equality when $u=0$. When $u\geq (3r+2\vartheta)/2$, we have 
\begin{equation*}
    \min\{d_0(u),d_1(u)\} = d_1(u) = 4\vartheta.
\end{equation*}
Thus we have 
\begin{equation*}
    f_{\text{Hamm}} = 1-\vartheta-\frac{(3r-2\vartheta)_+^2}{24r}.
\end{equation*}

\subsection{Proof for Proposition \ref{prop:mis-specified-prior-hamming-rate}}
Similar to the proof for  Proposition \ref{prop:M-tildeM-comparison}, we have 
\begin{equation*}
    f_{\text{Hamm}} = 1-\frac{1}{8} \max_u \min\{d_0,d_1\},
\end{equation*}
where 
\begin{equation*}
    d_0 = \frac{2(r+\vartheta+2u)^2}{r}.
\end{equation*}
\begin{equation*}
    d_1 = 8\vartheta + 2\min\{\frac{(2\sqrt{r_0r}-r-\vartheta-2u)_+^2}{r},
    \frac{4(\sqrt{r_0r}-u)_+^2}{r}\}.
\end{equation*}
When $r_0\leq r$,
\begin{equation*}
    d_0\geq d_1.
\end{equation*}
Since $d_0$ and $d_1$ are monotone functions w.r.t $u$, we conclude that 
\begin{equation*}
    \min\{d_0,d_1\} \leq 8\vartheta + \frac{2(2\sqrt{r_0r}-r-\vartheta-2u)_+^2}{r},
\end{equation*}
The above inequality became equality when $u=0$. Thus
\begin{align*}
    f_{\text{Hamm}} &= 1-\vartheta - \frac{(2\sqrt{r_0r}-r-\vartheta)_+^2}{4r} \\
    &\geq 1-\vartheta - \frac{(r_0-\vartheta)_+^2}{4r_0}.
\end{align*}
When $r_0\geq r \geq \vartheta$,
\begin{equation*}
    f_{\text{Hamm}} \geq 1-\frac{(r+\vartheta)^2}{4r} \geq 1-\vartheta - \frac{(r_0-\vartheta)_+^2}{4r_0}.
\end{equation*}
When $\vartheta\geq r_0\geq r$,
\begin{equation*}
    f_{\text{Hamm}} = 1 - \vartheta -
    \frac{((\xi+1)^2r-4\vartheta)_+^2}{16(\xi+1)^2r} = 1-\vartheta - \frac{(r_0-\vartheta)_+^2}{4r_0},
\end{equation*}
where 
\begin{equation*}
    \xi = 2\sqrt{r_0/r}-1 \geq 1.
\end{equation*}
When $r_0\geq \vartheta \geq r$, 
\begin{equation*}
    f_{\text{Hamm}} = 1 - \vartheta = 1-\vartheta - \frac{(r_0-\vartheta)_+^2}{4r_0}.
\end{equation*}

\subsection{Proof for Proposition \ref{prop:sign-max-general}}

Define the $\sigma$-field $\mathcal{E}^\tau = \sigma\{B\cap\boldsymbol{1}(\tau=t):1\geq t\geq T, B\in \Drt\}$. Since $\Drt \supseteq \Dkt{1}$, $\mathcal{E}^\tau  \supseteq \Dktau{1}$. Note that $\Dktau{1}$ is a $\sigma$-field.
Define $\eta^t = \text{logit}\left( q^t\right)$ and $\eta^\tau = \text{logit}\left( q^\tau\right)$.
We now show that $q^\tau = \mathbbm{P}(r^\tau=1\mid \mathcal{E}^\tau )$. 
First, since $\tau$ is a stopping time for filtration $\{ \Dkt{1}\}_{t=1}^T$, $\{\tau=t\}$ is $ \Dkt{1}$-measurable and hence $\Drt$-measurable.
Also, on the event $\{\tau=t\}$, $q^\tau=q^t=\mathbbm{P}(r^t=1\mid \Drt)$ is $\Drt$-measurable. Therefore, $q^\tau$ is $\mathcal{E}^\tau$-measurable. For each generator $A = B\cap \boldsymbol{1}(\tau=t)$, 
\[
\int_A q^\tau dP = \int_{B\cap \boldsymbol{1}(\tau=t)} q^t dP = \int_{B\cap \boldsymbol{1}(\tau=t)} \mathbbm{P}(r^t=1\mid \Drt) dP = \int_A \boldsymbol{1}(r^t=1) dP.
\]
Since the generator $\{B\cap\boldsymbol{1}(\tau=t):1\geq t\geq T, B\in D^t\}$ is a $\pi$-system, which determines the integral for the sigma filed $\mathcal{E}^\tau$, the above equality holds for every event $A\in \mathcal{E}^\tau$. Thus, $q^\tau = \mathbbm{P}(r^\tau=1\mid \mathcal{E}^\tau )$.
To this end, we can write $M^\tau = r^\tau \cdot \text{logit}\mathbbm{P}(r^\tau=1\mid \mathcal{E}^\tau )$. It is clear that $\mathbbm{P}(M^\tau>0 \mid \mathcal{E}^\tau) = \frac{1}{1+e^{-|M^\tau|}}$. Therefore, $\mathbbm{P}(M^\tau>0 \mid |M^\tau|) = \frac{1}{1+e^{-|M^\tau|}}$. 

Similarly, under null, $r^t$ is a coin flip independent with $\Drt$. Therefore,
\[
P(r^\tau\in A, B\cap\boldsymbol{1}(\tau=t)) = 
P(r^t\in A, B\cap\boldsymbol{1}(\tau=t)) = 
P(r^t\in A)P(B\cap\boldsymbol{1}(\tau=t))
\]
Since the generator $\{B\cap\boldsymbol{1}(\tau=t):1\geq t\geq T, B\in D^t\}$ is a $\pi$-system, 
we know that $r^\tau$ is independent with $\mathcal{E}^\tau$. Therefore, $M^\tau=r^\tau \cdot \text{logit}\mathbbm{P}(r^\tau=1\mid \mathcal{E}^\tau )$ is symmetry under null.

\subsection{Proof for Proposition \ref{prop:opt-stop-time}}

First, $\tau^\star = \underset{\tau}{\arg\max} \  \E[R_\tau]$ is the optimal stopping theorem introduced in \cite{snell1952applications}. Now we only need to show that, for any stopping time $\tau$ and weight function $\omega(\cdot)$, $\E[\omega(|M^\tau|)] = \E[R_\tau]$. For $\tau\equiv t$. By tower rule, $\E[R_t] = \E[\E[\omega(|M^t|)\mid \Dkt{1}]] = \E[\omega(|M^t|)]$. In the proof of Proposition \ref{prop:sign-max-general}, we show that $M^\tau=r^\tau \cdot \text{logit}\mathbbm{P}(r^\tau=1\mid \mathcal{E}^\tau )$. Therefore, $\omega(|M^\tau|) = \omega \circ \text{logit} (q^\tau):=\widetilde \omega(q^\tau)$. Since $R_t = \E[\widetilde\omega(q^t)\mid \Dkt{1}]$, by stopping time theory, we know $R_\tau = \E[\widetilde\omega(q^\tau)\mid \Dktau{1}]$. Therefore, by tower rule again, we have $\E[R_\tau] = \E[\widetilde \omega(q^\tau)] = \E[\omega(|M^\tau|)]$.

\section{Sampling Details for Spike-and-slab Prior}
\label{sec:sampling-detailsdoukhan2007probability}
In this section, we will provide details in sampling from the posterior distribution for computing Bayes-optimal mirror statistics.   
\subsection{Spike-and-slab Gibbs update rules for linear model}

Suppose we have the following hierarchical model,
\begin{align*}
    &y_i \sim \mathcal{N}(x_i \beta, \sigma^2), \\
    &\beta_j \sim (1-\pi) \delta_0 + \pi \mathcal{N}(\mu, \tau^2), \\
    &\pi_j \sim \text{ Bern }(\theta), \\
    &\sigma^2 \sim \text{ InvGamma }(a_1, a_2), \\
    &\tau^2 \sim \text{ InvGamma }(b_1, b_2), \\
    &\theta \sim \text{ Beta }(A,B), \\
    &\mu \sim \mathcal{N}(\mu_0, \tau_0^2).
\end{align*}
Each iteration consists of regression update and hyper-parameters update. 
\begin{itemize}
    \item \textbf{Update $\pi_j$}
    
    In regression update step, we update each feature $j$ randomly and once at a time, so it's essentially the same as in the 1d regression problem.
    \begin{align}
    \label{eq:linear-update-pi}
        \frac{P(\pi_j=0 \mid \cdot)}{P(\pi_j=1 \mid \cdot)} &= \frac{p_0}{1-p_0}\times \frac{P(y \mid X, \pi_j=0)}{P(y \mid X, \pi_j=1)} \cr
        &=\frac{p_0}{1-p_0}
        \times \text{det} (Q)^{1/2} \exp(-\frac{\tau^2}{2\sigma^4}(y-\mu)^T X Q X^T (y-\mu)) \cr
        &\quad \times \exp(-\frac{1}{2\sigma^2}||y||^2 +\frac{1}{2\sigma^2}||y-\mu||^2),
    \end{align}
    where $Q = Q_j = 1 + \frac{\tau^2}{\sigma^2} ||x_j||^2$ is a scalar and we denote the residual
    $r_i := y_i - \sum_{l\neq j} x_{il} \beta_l$, which satisfies
    \begin{equation*}
        r_i := y_i - \sum_{l\neq j} x_{il}\beta_l = x_{ij} \beta_j + \epsilon_i.
    \end{equation*}
    
    \item \textbf{Update $\beta_j$}

    When the other features except $j$ are fixed, consider the model $r_i = x_{ij}\beta_j + \epsilon_i$. 
    \begin{equation*}
    \beta_j \mid \mathbf{r}, \pi_j, \tau^2, \sigma^2 \sim
    \begin{cases}
    \delta_0 & \text{if } \pi_j = 0 \\[10pt]
    \mathcal{N} \left( \frac{\sum_{i=1}^n r_i x_i}{\sum_{i=1}^n x_i^2 + \frac{1}{\tau^2}}, \frac{\sigma^2}{\sum_{i=1}^n x_i^2 + \frac{1}{\tau^2}} \right) & \text{if } \pi_j = 1
    \end{cases}
    \end{equation*}

    \item \textbf{Update Hyper-parameters}

    If the problem is one dimensional, we have the following update rules,
    \begin{align*}
        &\theta \mid \cdot = \theta \mid \pi \sim \text{ Beta } (A+\pi, B+1-\pi), \\
        &\tau^2 \mid \beta, \mu, \pi=1 \sim \text{ InvGamma } (b_1 + \frac{1}{2}, b_2 + \frac{(\beta-\mu)^2}{2}),\\
        &\tau^2 \mid \beta, \mu, \pi=0 \sim \text{ InvGamma } (b_1, b_2), \\
        &\mu \mid \beta, \tau^2, \pi=1 \sim \mathcal{N}\left(\frac{\frac{\mu_0}{\tau_0^2}+\frac{\beta}{\tau^2}}{\frac{1}{\tau_0^2}+\frac{1}{\tau^2}}, \frac{1}{\frac{1}{\tau_0^2}+\frac{1}{\tau^2}}\right),\\
        &\mu \mid \beta, \tau^2, \pi=0 \sim \mathcal{N}(\mu_0,\tau^2_0), \\ 
        &\sigma^2 \mid y, \beta \sim \text{ InvGamma } (a_1 +\frac{n}{2}, a_2 + \frac{\sum_{i=1}^n (y_i-x_i\beta)^2}{2}).
    \end{align*}
    
    Now in general case, we have
    \begin{align*}
        &\theta \mid \cdot = \theta \mid \pi \sim \text{ Beta } (A+\sum_j \pi_j, B+\sum_j 1-\pi_j), \\
        &\tau^2 \mid \beta, \mu, \pi \sim \text{ InvGamma } (b_1 + \frac{1}{2} \sum_j \pi_j, b_2 + \frac{1}{2} \sum_j \pi_j(\beta_j-\mu)^2)\quad,   \\
        &\mu \mid \beta, \tau^2, \pi=1 \sim \mathcal{N}\left(\frac{\frac{\mu_0}{\tau_0^2}+\frac{\sum_j\beta_j}{\tau^2}}{\frac{1}{\tau_0^2}+\frac{p}{\tau^2}}, \frac{1}{\frac{1}{\tau_0^2}+\frac{p}{\tau^2}}\right),  \\ 
        &\sigma^2 \mid y, \beta \sim \text{ InvGamma } (a_1 +\frac{n}{2}, a_2 + \frac{\sum_{i=1}^n (y_i-x_i\beta)^2}{2}).
\end{align*}

\end{itemize}


\subsection{Spike-and-slab update rules for logistic model}

Deriving scalable sampling methods for logistic model has been actively studied in the statistical literature \cite{Ray_2021, biswas2022scalablespikeandslab, barndorff1982normal, biane2001probability, polson2013bayesianinferencelogisticmodels}. However, the sampling for Spike-and-Slab logistic model is much more trickier. The approach we use follows from the $t$-distribution approximation for the logistic distribution in \cite{obrien2004, narisetty2019skinny}. We can approximate $\text{Logistic}(x_i^T\boldsymbol{\beta},1)$ with $x_i^T\boldsymbol{\beta}+wt_\nu$, where $t_\nu$ denotes a t-distribution with $\nu$ degrees of freedom and $w$ is a multiplicative factor. The choice of the constants $\nu = 7.3$ and $w^2 = \frac{\pi^2(\nu-2)}{3\nu}$ were recommended by \cite{obrien2004}. This allows us to re-write the Spike-and-Slab logistic model as
\begin{align*}
    z_j &\overset{iid}{\sim} \text{Bernoulli}(\theta_j), \\
    \beta_j \mid z_j &\overset{ind}{\sim} (1-z_j)\delta_0 + z_j \mathcal{N}(\mu,\tau^2),\\
    \widetilde\sigma^2 &\overset{iid}{\sim}\text{InvGamma}(\frac{\nu}{2},\frac{w^2\nu}{2}),\\
    \widetilde y_i \mid \boldsymbol{\beta}, \widetilde\sigma^2 &\overset{ind}{\sim} \mathcal{N}(x_i^T\boldsymbol{\beta},\widetilde\sigma^2), \\
    y_i &= \mathbbm{1}(\widetilde y_i > 0).
\end{align*}
Thus it reduces to sample from the spike-and-slab linear model discussed in the previous sub-section.

\end{document}